\documentclass[10pt,superscriptaddress]{article}

\usepackage[margin=2.7cm]{geometry}

\usepackage{graphicx}
\usepackage{amsmath}
\usepackage{amssymb}
\usepackage{amsfonts}
\usepackage{mathrsfs}
\usepackage{amsthm}
\usepackage{bm}
\usepackage{url}
\usepackage{youngtab}
\usepackage[T1]{fontenc}
\usepackage{csquotes}
\MakeOuterQuote{"}
\usepackage{authblk}
\newtheoremstyle{note}
  {\topsep/2}               % ABOVE SPACE
  {\topsep/2}               % BELOW SPACE
  {}                      % BODY FONT
  {\parindent}            % INDENT (empty value is the same as 0pt)
  {\itshape}              % HEAD FONT
  {.}                     % HEAD PUNCTUATION
  {5pt plus 1pt minus 1pt}% HEAD SPACE
  {}

\newtheorem{theorem}{Theorem}
\newtheorem{lemma}{Lemma}
\newtheorem{conjecture}{Conjecture}
\newtheorem{corollary}{Corollary}
\newtheorem{proposition}{Proposition}

\newtheorem{remark}{Remark}
\newtheorem{definition}{Definition}

\def\vec#1{\bm{#1}} %% overiding the original command

\newcommand{\mrm}[1]{\mathrm{#1}}
\newcommand{\tr}{\operatorname{tr}}

\newcommand{\sgn}{\operatorname{sgn}}

\newcommand{\var}{\mathrm{Var}}

\newcommand{\rmd}{\mathrm{d}}
\newcommand{\rme}{\mathrm{e}}
\newcommand{\rmi}{\mathrm{i}}

\newcommand{\rmC}{\mathrm{C}}
\newcommand{\rmE}{\mathrm{E}}
\newcommand{\rmO}{\mathrm{O}}
\newcommand{\rmT}{\mathrm{T}}
\newcommand{\rmU}{\mathrm{U}}

\newcommand{\be}{\begin{equation}}
\newcommand{\ee}{\end{equation}}
\newcommand{\ba}{\begin{align}}
\newcommand{\ea}{\end{align}}

\def\<{\langle}  %% overiding the original command \<
\def\>{\rangle}  %% overiding the original command \>
\newcommand{\ket}[1]{| #1\>}
\newcommand{\bra}[1]{\< #1|}

\newcommand{\braket}[2]{\<#1|#2\>}
\def\outer#1#2{|#1\>\<#2|}       %% overiding the original command \outer

\newcommand{\CC}{\mathbb{C}}
\newcommand{\RR}{\mathbb{R}}
\newcommand{\FF}{\mathbb{F}}
\newcommand{\bbF}{\mathbb{F}}
\newcommand{\bbR}{\mathbb{R}}
\newcommand{\bbC}{\mathbb{C}}
\newcommand{\bbQ}{\mathbb{Q}}
\newcommand{\bbZ}{\mathbb{Z}}
\newcommand{\Sp}{\mrm{Sp}}

\newcommand{\Cli}{\mathrm{C}}
\newcommand{\PCli}{\overline{\mathrm{C}}}

\newcommand{\orb}{\mathrm{orb}}
\newcommand{\Id}{\mathrm{I}}

\def\eqref#1{\textup{(\ref{#1})}}  %% overiding the original command \eqref
\newcommand{\eref}[1]{Eq.~\textup{(\ref{#1})}}
\newcommand{\Eref}[1]{Equation~\textup{(\ref{#1})}}
\newcommand{\esref}[1]{Eqs.~\textup{(\ref{#1})}}

\newcommand{\tref}[1]{Table~\ref{#1}}

\newcommand{\sref}[1]{Sec.~\ref{#1}}

\newcommand{\thref}[1]{Theorem~\ref{#1}}

\newcommand{\lref}[1]{Lemma~\ref{#1}}
\newcommand{\Lref}[1]{Lemma~\ref{#1}}
\newcommand{\lsref}[1]{Lemmas~\ref{#1}}

\newcommand{\crref}[1]{Corollary~\ref{#1}}

\newcommand{\pref}[1]{Proposition~\ref{#1}}
\newcommand{\Pref}[1]{Proposition~\ref{#1}}

\newcommand{\cref}[1]{Conjecture~\ref{#1}}
\newcommand{\Cref}[1]{Conjecture~\ref{#1}}

\newcommand{\aref}[1]{Appendix~\ref{#1}}

\newcommand{\rcite}[1]{Ref.~\cite{#1}}
\newcommand{\rscite}[1]{Refs.~\cite{#1}}

\DeclareMathOperator{\Sym}{Sym}
\DeclareMathOperator{\Hom}{Hom}
\DeclareMathOperator{\Harm}{Harm}
\DeclareMathOperator{\vectorize}{vec}
\newcommand{\Ten}{\tau}

\begin{document}

\title{The Clifford group fails gracefully to be a unitary 4-design}

\author[1]{Huangjun Zhu}
\author[1]{Richard Kueng}
\author[2]{Markus Grassl}
\author[1]{David Gross}

\affil[1]{Institute for Theoretical Physics, University of Cologne, Germany}
\affil[2]{Max Planck Institute for the Science of Light, Erlangen,  Germany}

\date{\today}
\maketitle

\begin{abstract}
	A \emph{unitary $t$-design} is a set of unitaries that is ``evenly
	distributed'' in the sense that the average of any $t$-th order
	polynomial over the design equals the average over the entire unitary
	group.
	In various fields -- e.g.\ quantum information theory -- one
	frequently encounters constructions that rely on matrices drawn
	uniformly at random from the unitary group. 
	Often, it suffices to sample these matrices
	from a unitary $t$-design, for sufficiently high $t$.
	This results in more explicit, \emph{derandomized}
	constructions.
	The most prominent unitary $t$-design considered in quantum
	information is the \emph{multi-qubit Clifford group}. 
	It is known to be a unitary $3$-design, but, unfortunately, not a $4$-design.
	Here, we give a simple, explicit characterization of the way in
	which the Clifford group fails to constitute a $4$-design. 
	Our results show that for various applications in quantum
	information theory and in the theory of convex signal recovery,
	Clifford orbits perform almost as well as those of true $4$-designs.
	Technically, it turns out that in a precise sense, the $4$th tensor power of the
	Clifford group affords only one more invariant subspace than the $4$th
	tensor power of the unitary group. That additional subspace is a
	\emph{stabilizer code} -- a structure extensively studied in the field
	of quantum error correction codes.
	The action of the Clifford group on this stabilizer code can be decomposed explicitly into previously known irreps of the discrete symplectic group.
	We give various constructions of exact complex projective
	4-designs or approximate 4-designs of arbitrarily high precision
	from Clifford orbits. Building on results from coding theory, we
	give strong evidence suggesting that these orbits actually
	constitute complex projective $5$-designs.
\end{abstract}

\tableofcontents
\section{Introduction}

\subsection{Motivation: Designs and derandomizations}

A $d$-dimensional \emph{complex projective design} is a
configuration of vectors that are ``evenly distributed'' on the complex unit
sphere in $\bbC^d$. 
More precisely,
a set of unit-length
vectors  is a complex projective  $t$-design, if sampling
uniformly from the set gives rise to a random vector  whose first $2t$
moments agree with the moments of the uniform distribution on the sphere  \cite{DelsGS77,Hogg82,ReneBSC04, Scot06,AmbaE07}. 
This property makes designs a useful tool for the derandomization of
constructions that rely on random vectors. \emph{Unitary designs} are an analog of complex projective designs on the unitary group, which are equally useful for the derandomization of
constructions that rely on  random unitaries \cite{Dank05the,DankCEL09, GrosAE07, RoyS09}. In addition, unitary designs provide a simple way for constructing projective designs.

Applications of designs abound, with examples including randomized benchmarking \cite{KnilLRB08,MageGE11,WallF14}, quantum state tomography \cite{HayaHH05,Scot06, ZhuE11,Zhu14IOC}, quantum process tomography \cite{Scot08,KimmL15},  quantum cryptography \cite{AmbaBW09}, data hiding \cite{DiViLT02},  decoupling \cite{AbeyDHW09,SzehDTR13}, and tensor networks \cite{NezaW16}. 

\subsection{Outline of result: Overcoming the ``$t=3$-barrier''}

One major drawback of the program of using complex projective and unitary designs
for derandomization is that there has been little progress in
constructing explicit families of $t$-designs for $t>3$. There are
various constructions using ``structured randomness'' -- most notably
the \emph{random circuit model} that yields approximate designs in any
dimension and of any degree  \cite{HarrL09,BranHH15}. While the
resulting designs are sufficiently well-structured for some tasks in
quantum information theory, they are arguably not as explicit as one
could hope for. 

The only explicit infinite family of unitary $3$-designs known so far is the complex Clifford group, while the only explicit infinite family of  projective $3$-designs are  the orbits of the  Clifford group \cite{KuenG15, Zhu15MC, Webb15} (and even these are very recent results). 
Unfortunately, it has also been shown that the Clifford group is not a unitary $4$-design, and their orbits are not, in general, projective $4$-designs \cite{KuenG15, Zhu15MC, Webb15}. 

This situation seems all the more unsatisfactory, as there are various
applications -- including the two examples given in Section~\ref{sec:twoApplications} below -- where
$2$-designs are essentially useless \cite{MattWW09,GroKraKue15a},
$3$-designs give first non-trivial improvements \cite{GroKraKue15a}, and
$4$-designs show  already optimal behavior 
\cite{KabaKRT15, AmbaE07, MattWW09}.
The case $t=4$ treated here is thus not another step in an infinite series of potential papers, but rather seems to constitute a fundamental threshold.
Other prominent applications of 4-designs include randomized benchmarking \cite{WallF14} and quantum process tomography \cite{KimmL15}.

The main result of the present work is that while Clifford orbits fall
short of constituting $4$-designs, their $8$th moments can be
 calculated explicitly. The results are sufficiently well-behaved that
for several applications, Clifford orbits turn out to perform nearly as well
as $4$-designs or uniform random vectors would.
Moreover, even exact 4-designs can easily be constructed from Clifford orbits. 
In order to establish these statements, we give an explicit description of the irreducible representations of the $4$th tensor power of the Clifford group. 
In a precise sense, it turns out that this tensor power affords only one more invariant subspace than the $4$th tensor power of the unitary group.  That additional subspace is a \emph{stabilizer code} -- a structure extensively studied in the field of quantum error correction codes \cite{Gott97the, NielC00book}. 
This feature allows for an explicit analysis.

This paper contains only the representation-theoretic analysis of the
4th tensor power of the Clifford group. 
In two companion papers we apply this technical result to problems from signal analysis \cite{StabilizerLRMR16} and quantum information theory \cite{PovmNorm16} respectively.
These applications are briefly sketched below.
The reason for splitting our discussion three-ways is that we
target two different scientific communities that have come to employ very different languages.

\subsection{Two applications}
\label{sec:twoApplications}

Here, we give a high-level description of two seemingly very different problems that originally motivated our work and to which our results can be applied.

\subsubsection{Application: Phase Retrieval}

The signal analysis example is the problem of \emph{phase retrieval}:
Let $x$ be an unknown vector in $\bbC^d$. Assume we have access
to a set of ``phase insensitive linear measurements''
\begin{equation}\label{eqn:phaseless}
	y_i = |(a_i, x)|, \qquad i = 1, \dots, m.
\end{equation}
Here, the $a_i\in\bbC^d$ are a given set of \emph{measurement
vectors}. The task now is to recover $x$ given $y_1, \dots, y_n$.
There are many practical applications -- for example in optical
microscopy, where information about a sample is encoded in the
electro-magnetic light field, but where only phase-insensitive
intensity measurements are usually feasible.
From a mathematical point of view,
the absolute value in \eref{eqn:phaseless} means that we are facing a
\emph{non-linear inverse problem} -- which are often difficult to solve in
theory and in practice. 

A recent research program has investigated the use of algorithms based
on \emph{convex optimization} for the purpose of solving the phase
retrieval problem. First theoretical results have shown that certain
convex algorithms do indeed recover $x$ with high probability, if the
measurements $a_i$ are random Gaussian vectors or drawn uniformly from
the unit sphere in  $\bbC^d$
\cite{CanStrVor13,CanLi14}. 
However, in
many practical applications, such measurements cannot be realized.
Therefore, we are facing the task of re-proving  those guarantees 
for measurements that are ideally deterministic, or, if
randomized, at least drawn from a ``smaller'' and ``more
structured'' set of vectors than from the entire unit sphere. Such
\emph{derandomized} versions of convex-optimization algorithms have indeed been
established for a variety of models -- see e.g.\ Refs.~\cite{CanLiSol15, GroKraKue15a}.

Starting from \rcite{GroKraKue15a}, some of the present authors have been
interested in using spherical designs as a ``general-purpose'' tool for
randomizing phase retrieval algorithms. The basic insight is that
protocols that ostensibly require Gaussian vectors often only rely on
certain measure-concentration estimates that can be derived already
from information about finite moments. 
Initial results \cite{GroKraKue15a} showed that while a $2$-design property alone does not give rise to non-trivial recovery guarantees, this changes from $t=3$ onward. 
Later, it has been proven that $4$-designs essentially match the performance of random Gaussian measurements \cite{KabaKRT15}. 
In accordance with our initial hope, the results of \cite{KabaKRT15} were first proven for Gaussian measurements and then generalized -- with comparatively few changes in the argument --  to the design case.

In \rcite{StabilizerLRMR16}, we use the theory of the present paper to establish near-optimal performance guarantees for phase retrieval from measurements drawn randomly from Clifford orbits. 
This includes the case of stabilizer measurements. 
Generalizations to the recovery of low-rank matrices are also proven.

\subsubsection{Application: Quantum state distinguishability}

Our second example comes from \emph{quantum information theory}. 
In quantum mechanics, the \emph{state} of a $d$-level system is encoded in a positive semi-definite $d\times d$-matrix, the so-called \emph{density operator} or \emph{density matrix}. 
A \emph{measurement} maps density operators
to classical probability distributions over a space of outcomes. The
fundamental property of \emph{quantum complementarity} means that 
measurements necessarily entail a loss of information about the
quantum system. 

One way of precisely measuring this information loss is as follows: The (single-shot)
statistical distinguishability of two classical probability
distributions $p, q$ is measured by the \emph{total variational
distance}, or half their $\ell_1$-norm distance $d_c(p,q) := \frac12
\| p - q \|_{\ell_1}$. 
Analogously, the optimal probability of distinguishing between two
quantum states $\rho, \sigma$ is given by one half the
\emph{Schatten-1 norm} (or \emph{trace norm} or \emph{nuclear norm})
distance: $d_q(\rho, \sigma) := \frac12 \|\rho-\sigma\|_{1}$.
Quantum measurements are
represented by positive-operator-valued measures (POVMs), which realize
(certain) linear maps $\Lambda$ from the set of density matrices to
the set of classical probability distributions. 
The fact that ``information is lost'' in such a process can e.g.\ be
made precise by stating that $\Lambda$ is a contraction:
\begin{equation*}
	d_c(\Lambda(\rho), \Lambda(\sigma)) \leq  d_q(\rho,\sigma).
\end{equation*}
The information loss of a given $\Lambda$ can be upper-bounded via the \emph{POVM norm constant} $C_\Lambda<1$, which is the largest constant so that
\begin{equation*}
	d_c(\Lambda(\rho), \Lambda(\sigma)) \geq  
	C_\Lambda d_q(\rho,\sigma)
\end{equation*}
holds for any pair $\rho, \sigma$.
It thus makes sense to ask for an optimal measurement $\Lambda$, i.e.\ one that
maximizes $C_\Lambda$. 

It has been shown  that the
\emph{uniform POVM} achieves this goal \cite{MattWW09}.
This measurement maps quantum states to probability distributions on
the complex unit sphere, where the density $p(\psi)$ at the vector
$\psi$ is
proportional to $\tr(\rho\,|\psi\rangle\langle \psi|)$. 

The situation is now very similar to the one considered in the phase
retrieval example above: 
The uniform POVM is optimal, but impractical to implement in large
quantum experiments. 
However, as has been shown in Ref.~\cite{AmbaE07,MattWW09}, 
restricting the uniform POVM to a set of vectors that form a $4$-design
gives rise to a quantum measurement which already  matches the optimal scaling behavior.
Again, an analogous statement for $2$-designs does not hold \cite{MattWW09}.

Building on the theory developed below, we analyze quantum state distinguishability as measured by POVM norms  of Clifford orbits in \rcite{PovmNorm16}. 
For states with high purity, near-optimal results are established for all Clifford orbits, including stabilizer measurements.
As an auxiliary result, we also establish tighter entropic uncertainty relations 
\cite{WehnW10,coles_entropic_2015}
for stabilizer measurements.

\subsection{Note added}

While finalizing this paper, we became aware of a work by Helsen, Wallman, and Wehner that analyses a closely related representation of the Clifford group with the aim to derive improved bounds for randomized benchmarking \cite{helsen2016representations}.  
More precisely, they work with the representation $\tau^{(2,2)}$ in the sense of 
Eq.~(\ref{eqn:partialTransposeCommutant}).
As described below, this means that the main results of our respective papers are largely equivalent.
The proof methods seem to be rather different.

\section{Mathematical Background} 

In this section we review the mathematical background on complex projective designs, unitary designs, the Pauli group, Clifford group, and stabilizer states.
  
\subsection{Projective $t$-designs}
Complex projective designs are configurations of vectors that are evenly distributed on the complex unit sphere in $\bbC^d$. They are an analog of spherical designs on the real unit sphere \cite{DelsGS77,Hogg82,ReneBSC04,BannB09}; c.f.~\aref{sec:polynomials}. Such designs are
interesting to a number research areas, such as approximation theory, combinatorics, experimental designs etc. Recently, they have also found increasing applications
in many quantum information processing and signal analysis tasks, such as quantum state tomography \cite{HayaHH05,Scot06, ZhuE11,Zhu14IOC}, quantum state discrimination \cite{AmbaE07, MattWW09}, and phase retrieval  \cite{GroKraKue15a}. Here we review three equivalent definitions of (complex projective) $t$-designs.

Let  $\Hom_{(t,t)}(\bbC^d)$ be the space of  polynomials homogeneous of degree $t$ in the coordinates of $\psi \in \bbC^d$ with respect to a given basis and  homogeneous of degree $t$ in the coordinates of the complex conjugate $\psi^*$ (refer to \aref{sec:polynomials} for detailed notes on multivariate polynomials). 
\begin{definition}
	A set of $K$ normalized vectors $\{\psi_j\}$ in dimension $d$  is a (complex projective) \emph{$t$-design} if 
	\begin{equation}
\frac{1}{K}\sum_j p(\psi_j)= \int p(\psi)\rmd \psi\quad \forall p\in \Hom_{(t,t)}(\bbC^d),
	\end{equation}
 where the integral is taken with respect to the normalized uniform  measure on the complex unit sphere in $\bbC^d$. 
 \end{definition}

To derive simpler criteria on $t$-designs, we need to introduce several additional concepts. Let $\Sym_t(\bbC^d)$ be the $t$-partite symmetric subspace of $(\bbC^d)^{\otimes t}$ with corresponding projector $P_{[t]}$.
The dimension of   $\Sym_t(\bbC^d)$ reads
\begin{equation}
D_{[t]}=\binom{d+t-1}{t}.
\end{equation}
The $t$th frame potential of $\{\psi_j\}$ is defined by
\begin{equation}
\Phi_t(\{\psi_j\}):=\frac{1}{K^2}\sum_{j,k} |\langle \psi_j|\psi_k\rangle|^{2t}. 
\end{equation}

\begin{proposition}\label{prop:designequiv}
The following statements  are equivalent:
\begin{enumerate}
\item $\{\psi_j\}$ is a $t$-design.

\item 
\label{item:notquiteright}
$\frac{1}{K} \sum_j (\outer{\psi_j}{\psi_j})^{\otimes t}=P_{[t]}/D_{[t]}$, where $K=|\{\psi_j\}|$. 

\item $\Phi_t(\{\psi_j\})=1/D_{[t]}$.
\end{enumerate}
\end{proposition}

\begin{remark}
In general,	$\Phi_t(\{\psi_j\})\geq 1/D_{[t]}$, and the lower bound is saturated iff $\{\psi_j\}$ is a $t$-design.
\end{remark}

\begin{proof}
Let  $L(\Sym_t(\bbC^d))$ be the space of linear operators acting on $\Sym_t(\bbC^d)$.
There is a one-to-one correspondence (Lemma~\ref{lem:complexPolarization}) between polynomials $p\in \Hom_{(t,t)}(\bbC^d)$ and  operators $A\in L(\Sym_t(\bbC^d))$,
\begin{equation}
A\mapsto p_A,\quad p_A(\psi):=\tr\bigl[A (\outer{\psi}{\psi})^{\otimes t}\bigr].
\end{equation}
Therefore,
\begin{equation}
\frac{1}{K}\sum_j p_A(\psi_j)=\frac{1}{K}\tr\biggl[A \sum_j (\outer{\psi_j}{\psi_j})^{\otimes t}\biggr], \qquad
 \int p_A(\psi)\rmd\psi=\tr\biggl[A \int(\outer{\psi}{\psi})^{\otimes t}\rmd\psi\biggr]. 
\end{equation} 
It follows that $\{\psi_j\}$ is a $t$-design iff
\begin{equation}\label{eq:tdesignSum}
\frac{1}{K}\sum_j (\outer{\psi_j}{\psi_j})^{\otimes t} =\int(\outer{\psi}{\psi})^{\otimes t}\rmd\psi=\frac{P_{[t]}}{D_{[t]}}.
\end{equation}
Here the second equality follows from the fact that the $t$th symmetric subspace is irreducible under the action of the unitary group. This observation confirms the equivalence of statements 1 and 2. 
The equivalence of statements 2 and 3 is a consequence of the following equation,
\begin{equation}
\biggl\|\frac{1}{K} \sum_j (\outer{\psi_j}{\psi_j})^{\otimes t}-\frac{P_{[t]}}{D_{[t]}}\biggr\|_2^2=\Phi_t(\{\psi_j\})-\frac{1}{D_{[t]}},
\end{equation}
where $\|\cdot\|_2$ denotes the Hilbert-Schmidt norm or the Frobenius norm. This equation  implies that $\Phi_t(\{\psi_j\})\geq 1/D_{[t]}$, and the lower bound is saturated iff \eref{eq:tdesignSum} is satisfied. 
\end{proof}

\Pref{prop:designequiv}  suggests several useful measures for characterizing the deviation of  $\{\psi_j\}$ from $t$-designs.  For example, two common measures 
are the operator norm and trace norm of the deviation operator
\begin{equation}
\frac{D_{[t]}}{K} \sum_j (\outer{\psi_j}{\psi_j})^{\otimes t}-P_{[t]}.
\end{equation}
Another measure is the ratio of the frame potential over the minimum frame potential, that is, $D_{[t]}\Phi_t(\{\psi_j\})$.

Any $t$-design in dimension $d$ has at least
\begin{equation}\label{eq:tdesignLBound} \binom{d+\lceil
	t/2\rceil-1}{\lceil t/2\rceil}\binom{d+\lfloor t/2\rfloor-1}{
	\lfloor t/2\rfloor}
\end{equation}
elements,
where $\lceil
t/2\rceil$ denotes the smallest integer not smaller than $t/2$, and $\lfloor t/2\rfloor$ the largest integer not larger than $t/2$ \cite{Hogg82, Leve98, Scot06}.
The bound is equal to $d,d^2,d^2(d+1)/2, d^2(d+1)^2/4$ for $t=1,2,3,4$, respectively. A $t$-design is \emph{tight} if the lower bound is saturated.
A 1-design is tight iff it defines an orthonormal basis; a 2-design is tight iff it defines a symmetric informationally complete measurement (SIC) \cite{Zaun11,
	ReneBSC04, Scot06, ScotG10, ApplFZ15G}. Another interesting example of  2-designs are complete sets of mutually unbiased bases (MUB) \cite{Ivan81,
	WootF89, KlapR05M, DurtEBZ10}. The only known explicit infinite family of 3-designs are the orbits of the (multiqubit) Clifford group, among which the set of stabilizer states is particularly prominent \cite{KuenG15,Zhu15MC,Webb15}.

\begin{definition}
	A set of weighted normalized vectors $\{\psi_j, w_j\}$ in dimension $d$ with $w_j\geq0$ and $\sum_j w_j=1$ is a weighted (complex projective) \emph{$t$-design} if 
	\begin{equation}
   \sum_j w_j p(\psi_j)= \int p(\psi)\rmd \psi\quad \forall p\in \Hom_{(t,t)}(\bbC^d).
	\end{equation}
\end{definition}
A weighted $t$-design reduces to an ordinary unweighted  $t$-design when all weights are equal. 
In many contexts, weighted designs are equally useful as unweighted designs.
In the current paper, we construct unweighted 4-designs for dimensions that are a power of two.
They can easily be turned into weighted 4-designs for arbitrary dimensions $\tilde d$. 
Indeed, let $\tilde d\leq d$, let $P$ be a projection operator onto an arbitrary $\tilde d$-dimensional subspace of $\CC^{d}$, and let $\{\psi_j\}$ be a $t$-design.
Then one can verify immediately that with
\begin{equation*}
	\tilde\psi_j = \frac{1}{\|P\psi_j\|_2}\,\psi_j, \qquad
	\tilde w_j = \|P\psi_j\|_2,
\end{equation*}
the $\{\tilde\psi_j, \tilde w_j\}$ forms a weighted $t$-design.
This way, the findings of the present paper have consequences for any dimension -- a power of two or not.

Almost all conclusions about $t$-designs mentioned above, including \pref{prop:designequiv}, also apply  to 
weighted $t$-designs
provided that the operator $\frac{1}{K} \sum_j (\outer{\psi_j}{\psi_j})^{\otimes t}$ is replaced by $\sum_j w_j (\outer{\psi_j}{\psi_j})^{\otimes t}$, and the frame potential $\Phi_t(\{\psi_j\})$ is replaced by 
\begin{equation}
\Phi_t(\{\psi_j,w_j\}):=\sum_{j,k} w_j w_k |\langle \psi_j|\psi_k\rangle|^{2t}. 
\end{equation}

\subsection{Unitary $t$-designs}
\label{sec:intro:unitary}
Unitary designs are configurations of unitary operators that are ``evenly distributed'' on the unitary group, in analogy to spherical designs and complex projective designs. They are particularly useful in derandomizing constructions that rely on random unitaries, such as
randomized benchmarking \cite{KnilLRB08,MageGE11,WallF14},   quantum process tomography \cite{KimmL15},   quantum cryptography \cite{AmbaBW09}, data hiding \cite{DiViLT02},  and decoupling \cite{AbeyDHW09,SzehDTR13}.

Let  $\Hom_{(t,t)}(\rmU(d))$ be the space of  polynomials homogeneous of degree $t$ in the matrix elements of $U\in \rmU(d)$ and homogeneous of degree $t$ in the matrix elements of  $U^*$ (the complex conjugate of $U$; the Hermitian conjugate of $U$ is denoted by $U^\dag$). 
\begin{definition}
A set of $K$ unitary operators $\{U_j\}$ in dimension $d$  is a  \emph{unitary $t$-design} if 
\begin{equation}\label{eq:U2design}
\frac{1}{K} \sum _j p(U_j) =\int \rmd U p(U) \quad \forall p\in \Hom_{(t,t)}(\rmU(d)),
\end{equation}
where the integral is taken over the normalized Haar measure.
\end{definition}
The above equation remains intact even if $U_j$ are multiplied by arbitrary phase factors, so what we are concerned are actually projective unitary $t$-designs.
The $t$th frame potential of $\{U_j\}$ is defined as
 \begin{equation}
 \Phi_t(\{U_j\}):=\frac{1}{K^2}\sum_{j,k} |\tr(U_jU_k^\dag)|^{2t}.
 \end{equation}
As shown in the proof of \pref{pro:tdesignU} below,
 \begin{equation}
 \Phi_t(\{U_j\})\geq \gamma(t,d):=\int \rmd U |\tr(U)|^{2t},
 \end{equation}
 and the lower bound is saturated iff $\{U_j\}$ is a unitary $t$-design \cite{GrosAE07, Scot08, RoyS09}. The value of $\gamma(t,d)$ has been computed explicitly: it is equal to the number of permutations of $\{1,2,\ldots, t\}$ with no increasing subsequence of length larger than $d$  \cite{DiacS94,Rain98}.
Here we  only need the formula in  two special cases \cite{Scot08},
 \begin{equation}\label{eq:FramePmin}
 \gamma(t,d)=\begin{cases}
 \frac{(2t)!}{t! (t+1)!} &d=2,\\
 t!& d\geq t.
 \end{cases}
 \end{equation}

Like projective $t$-designs, there are many equivalent definitions of unitary $t$-designs.
\begin{proposition}\label{pro:tdesignU}
The following statements are equivalent:
\begin{enumerate}
\item $\{U_j\}$ is a unitary $t$-design.

\item $\frac{1}{K} \sum _j \tr\bigl[BU_j^{\otimes t} A(U_j^{\otimes t})^\dag\bigr] =\int \rmd U \tr\bigl[B U^{\otimes t}A(U^{\otimes t})^\dag\bigr]$ for all $A,B\in L((\bbC^d)^{\otimes t})$.

\item $\frac{1}{K} \sum _j U_j^{\otimes t} A(U_j^{\otimes t})^\dag =\int \rmd U U^{\otimes t}A(U^{\otimes t})^\dag$ for all $A\in L((\bbC^d)^{\otimes t})$.

\item $\frac{1}{K} \sum _j U_j^{\otimes t} \otimes (U_j^{\otimes t})^\dag =\int \rmd U U^{\otimes t}\otimes(U^{\otimes t})^\dag$.

\item $\frac{1}{K} \sum _j U_j^{\otimes t} \otimes (U_j^{\otimes t})^* =\int \rmd U U^{\otimes t}\otimes(U^{\otimes t})^*$.

\item $\Phi_t(\{U_j\})= \gamma(t,d)$.
\end{enumerate}
\end{proposition}
\begin{proof}
Note that $\tr\bigl[BU^{\otimes t} A(U^{\otimes t})^\dag\bigr]$ is a homogeneous polynomial in $\Hom_{(t,t)}(\rmU(d))$ and that all polynomials of this form for $A,B\in L((\bbC^d)^{\otimes t})$ span $\Hom_{(t,t)}(\rmU(d))$. Therefore, statements 1 and 2 are equivalent. The equivalence of statements 2 and 3 is obvious.

The equivalence of statements 1 and 4  follows from the following equation,
\begin{equation}
\tr\bigl\{V(B\otimes A) [U^{\otimes t} \otimes (U^{\otimes t})^\dag]\bigr\}=\tr\bigl\{BU^{\otimes t} A (U^{\otimes t})^\dag\bigr\},
\end{equation}
where $V$ is the swap operator of  parties $1,2, \ldots, t$ with the parties $t+1, t+2, \ldots, 2t$. The equation in statement 5 is a partial transposition of the one in statement 4. 

The equivalence of statements 5 and 6 follows from the following equation
\begin{equation}
\biggl\| \frac{1}{K} \sum _j U_j^{\otimes t} \otimes (U_j^{\otimes t})^* -\int \rmd U U^{\otimes t}\otimes(U^{\otimes t})^*\biggr\|_2=\Phi_t(\{U_j\})-\gamma(t,d).
\end{equation}
\end{proof}

Most known examples of unitary designs are constructed from subgroups of the unitary group, which are referred to as (unitary) group designs henceforth. Given a finite group $G$ of unitary operators, the frame potential of  $G$ takes on the form
\begin{equation}\label{eq:FramePotG}
\Phi_t(G)=\frac{1}{|G|}\sum_{U\in G} |\tr(U)|^{2t}.
\end{equation}
Let  $\overline{G}$ be the quotient of $G$ over the phase factors. Then 
\begin{equation}
\Phi_t(G)=\Phi_t(\overline{G})=\frac{1}{|\overline{G}|}\sum_{U\in \overline{G}} |\tr(U)|^{2t}.
\end{equation}
This formula is applicable whenever $\overline{G}$ is a finite group even if $G$ is not. 
Note that $\Phi_t(G)$ is equal to    the sum of squared multiplicities of irreducible components of
\begin{equation}
\Ten^t(G):=\{U^{\otimes t}|U\in G\},
\end{equation}
 which coincides with the dimension of the commutant of $\Ten^t(G)$ \cite{GrosAE07}. Recall that the commutant
$\mathcal{A}'$ of a set of operators $\mathcal{A}$ is the algebra of all operators that commute with every element of $\mathcal{A}$,
\begin{equation}
\mathcal{A}'=\{B| [A,B]=0\; \forall A\in \mathcal{A}\}.
\end{equation}
Let  $H$ be a subgroup in  $G$.
It is clear that
every irreducible representation of $\Ten^t(G)$ on
$\big(\CC^d\big)^{\otimes t}$ is also invariant under $\Ten^t(H)$ and
thus forms a representation space of $H$.  However, these spaces need not be irreducible under the
action of $H$. As a consequence, $\Phi_t(H)\geq \Phi_t(G)$ for any subgroup $H$ in $G$, and the equality is saturated iff every irreducible component of $\Ten^t(G)$ is also  irreducible when restricted to $\Ten^t(H)$; that is, $\Ten^t(G)$ and $\Ten^t(H)$ decompose into the same number of irreducible components.

At this point, it is instructive to 
 review the representation theory of the unitary group $\rmU(d)$ on
the space of all tensors $(\CC^d)^{\otimes t}$ from the point of view of \emph{Schur-Weyl duality} \cite{GoodW09book,Proc07book}. By definition the unitary group $\rmU(d)$ acts on $\bbC^d$. The action extends to the \emph{diagonal action} on  $(\CC^d)^{\otimes t}$, 
\begin{equation}
U\mapsto\tau^t(U): |\psi_1\rangle\otimes|\psi_2\rangle\otimes\cdots\otimes|\psi_t\rangle\mapsto U|\psi_1\rangle\otimes U|\psi_2\rangle\otimes\cdots\otimes U|\psi_t\rangle\quad  \forall |\psi_j\rangle \in\bbC^d,\;  \forall U\in \rmU(d).
\end{equation}
Meanwhile,  the symmetric group $S_t$ acts on the tensor product space
$(\bbC^d)^{\otimes t}$ by permuting the tensor factors:
\begin{equation}
\pi (|\psi_1\rangle\otimes|\psi_2\rangle\otimes\cdots\otimes|\psi_t\rangle) 
= 
|\psi_{\pi_1}\rangle\otimes|\psi_{\pi_2}\rangle\otimes\cdots\otimes|\psi_{\pi_t}\rangle
\qquad
\forall |\psi_j\rangle \in\bbC^d,\; \forall  \pi\in S_t.
\end{equation}
The diagonal action of $\rmU(d)$ and the permutation action of  $S_t$  on $(\bbC^d)^{\otimes t}$ commute with each other. Schur-Weyl duality states that
$(\bbC^d)^{\otimes t}$ decomposes into multiplicity-free
irreducible representations of $\rmU(d)\times S_t$ \cite{GoodW09book}. More precisely, 
\begin{equation}\label{eq:schurweyl}
\left(\bbC^{d}\right)^{\otimes t}
=
\bigoplus_{\lambda}H_\lambda =\bigoplus_{\lambda} W_\lambda \otimes
S_\lambda.
\end{equation}
Here the $\lambda$'s are non-increasing partitions of $t$ into no
more than $d$ parts, $W_\lambda$ is the \emph{Weyl module} carrying
the irrep of $\rmU(d)$ associated with $\lambda$, and $S_\lambda$ the
\emph{Specht module} on which $S_t$ acts irreducibly. 
We denote the 
dimensions of $S_\lambda$ and $W_\lambda$ by $d_\lambda$
and $D_\lambda$, respectively.
Note that $d_\lambda$ equals the multiplicity of the Weyl
module $W_\lambda$, and, likewise,  $D_\lambda$ is the multiplicity of
the Specht module $S_\lambda$. As an implication, the commutant of the diagonal action of the unitary group is generated by all permutations of the tensor factors. 
If $\lambda=[t]$ is the trivial partition, then
$W_\lambda=\Sym_t(\CC^d)$ and $S_t$ acts trivially on $S_\lambda\simeq
\CC$.
In particular, it follows that the space $\Sym_{t}(\CC^d)$ carries an
irreducible representation of $\rmU(d)$.

The discussion above leads to a number of equivalent characterizations of $t$-designs constructed from groups.
\begin{proposition}The following statements concerning $G\leq \rmU(d)$ are equivalent:
\begin{enumerate}
\item $G$ is a unitary $t$-design.	

 \item $\Phi_t(G)= \gamma(t,d)$.
 
 \item $\Ten^t(G)$ decomposes into the same number of irreps as  $\Ten^t(\rmU(d))$.

\item Every irreducible component in $\Ten^t(\rmU(d))$ is still irreducible when restricted to $\Ten^t(G)$.

\item $\Ten^t(G)$ and $\Ten^t(\rmU(d))$ has the same commutant.

\item The  commutant of  $\Ten^t(G)$ is generated by all the permutations of the tensor factors. 
\end{enumerate}	
\end{proposition}
For example, $G$ is a 1-design iff it is irreducible; in that case, $\overline{G}$  has at least $d^2$ elements, and the lower bound is saturated iff it defines a nice error basis, that is, $\tr(U_jU_k^\dag)=d\delta_{jk}$ for $U_j, U_k\in \overline{G}$ \cite{Knil96N,KlapR02}.
The  group $G$ is a unitary 2-design iff $\Ten^2(G)$ has only two irreducible components, which correspond to the symmetric and antisymmetric subspaces of the bipartite Hilbert space. Prominent examples of unitary group 2-designs include Clifford groups and restricted Clifford groups in prime power dimensions \cite{DiViLT02, Chau05, Dank05the, DankCEL09, GrosAE07}.

Complex projective designs and unitary designs are connected by the following proposition. 
\begin{proposition}
	Any orbit of normalized vectors of a unitary group $t$-design forms a complex projective $t$-design.
\end{proposition}

\begin{proof}
Let $G$ be a unitary group $t$-design, then $\Ten^t(G)$ acts irreducibly on  $\Sym_{t}(\CC^d)$. Therefore,
\begin{equation}
\frac{1}{|G|}\sum_{U\in \overline{G}} \bigl(U\outer{\psi}{\psi}U^\dagger\bigr)^{\otimes t}=\frac{1}{|G|}\sum_{U\in \overline{G}}U^{\otimes t }(\outer{\psi}{\psi})^{\otimes t}(U^{\otimes t})^\dag =\frac{P_{[t]}}{D_{[t]}} 
\end{equation}
for any normalized vector $\psi$. It follows that any orbit of pure states of $G$ forms a complex projective $t$-design.
\end{proof} 

\subsection{\label{sec:PauliCli}Pauli group, Clifford group, and stabilizer codes}

The Pauli group and Clifford group play a crucial role in  quantum computation \cite{Gott97the, GottC99, NielC00book, BravK05}, quantum error correction \cite{Gott97the,NielC00book},  randomized
benchmarking \cite{KnilLRB08,MageGE11,WallF14}, and quantum state tomography with compressed sensing \cite{GrosLFB10,Gros11,KimmL15}. They are also closely related to many interesting discrete structures, such as discrete Wigner functions \cite{Gros06,Zhu16P,gross2008quantum}, mutually unbiased bases \cite{DurtEBZ10}. 
Many nice properties of the Clifford group are closely related
to the fact that the group forms a unitary 2-design \cite{DiViLT02, Chau05,
	Dank05the, DankCEL09, GrosAE07, Scot08, RoyS09, HarrL09, ClevLLW15}.  Recently, it was shown that 
the multiqubit Clifford group is actually a unitary 3-design, but not a 4-design \cite{Zhu15MC,Webb15,KuenG15}. In the rest of this paper we assume that the dimension is a power of 2 when referring to the Pauli group or the Clifford group. 

Let $\FF_2=\mathbb{Z}_2=\{0,1\}$ be the finite field of integers with
arithmetic modulo $2$. We label the \emph{Pauli matrices} on a single qubit 
by elements of $\FF_2^2$ in the following way:
\begin{equation*}
	\sigma_{(0,0)} =  
	\left(
		\begin{array}{cc}
			1 & 0 \\
			0 & 1 
		\end{array}
	\right),
	\qquad
	\sigma_{(0,1)} =  
	\left(
		\begin{array}{cc}
			0 & 1 \\
			1 & 0 
		\end{array}
	\right),
	\qquad
	\sigma_{(1,0)} =  
	\left(
		\begin{array}{cc}
			1 & 0 \\
			0 & -1 
		\end{array}
	\right),
	\qquad
	\sigma_{(1,1)} =  
	\left(
		\begin{array}{cc}
			0 & -\rmi \\
			\rmi &0 
		\end{array}
	\right).
\end{equation*}
A \emph{Pauli operator} on $n$ qubits is defined as the tensor product
of $n$ Pauli matrices. Concretely, each $a\in\FF_2^{2n}$  defines a Pauli operator as follows,
\begin{equation*}
	W_a := \sigma_{(a_1,a_2)} \otimes \dots \otimes
	\sigma_{(a_{2n-1},a_{2n})}.
\end{equation*}
Every pair of Pauli operators either commute or anticommute, 
\begin{equation}\label{eqn:pauli_commutation_relation}
	W_a W_b=(-1)^{\langle a,b\rangle}W_b W_a,
\end{equation}
where $\langle a,b\rangle=a^\rmT J b$ is the symplectic form with $J$ being the $2n\times 2n$ block-diagonal
matrix over $\bbF_2$ with $n$ blocks of $\left(\begin{smallmatrix}
0 &1\\ 1 &0
\end{smallmatrix}\right)$ on the diagonal. 
Let 
\begin{equation*}
	\bar{\mathcal{P}}_n = \{ W_a \,|\, a\in\FF_2^{2n} \}
\end{equation*}
be the set of all $n$-qubit Pauli operators. 
The \emph{Pauli group} on $n$-qubits is the group generated by all the Pauli operators in $\bar{\mathcal{P}}_n$,
\begin{equation*}
\mathcal{P}_n = 
\langle \bar{\mathcal{P}}_n \rangle
= \{ \rmi^j W_a \,|\, a\in\bbF_2^{2n}, j \in \bbZ_4 \}.
\end{equation*}
In the following discussion $\bar{\mathcal{P}}_n $ is also identified as the projective Pauli group, the quotient group of $\mathcal{P}_n$ with respect to the phase factors. As a group,  $\bar{\mathcal{P}}_n $ is isomorphic to $\bbF_2^{2n}$.

The $n$-qubit Clifford group is usually defined as the normalizer of the  $n$-qubit Pauli group $\mathcal{P}_n$.
For the convenience of the following discussion, we shall define the Clifford group by specifying explicit generators. 
The single qubit Clifford group $\Cli_1$  is generated by  the Hadamard matrix $H$ and the phase matrix $S$, where
\begin{equation}\label{eqn:hadamardc}
H=\frac{1+\rmi}{2}\begin{pmatrix}
1 & 1\\
1&-1
\end{pmatrix},\qquad
S=\begin{pmatrix}
1 & 0\\
0&-\rmi
\end{pmatrix}.
\end{equation} 
Here our definition of the Hadamard matrix differs from the usual definition by a phase factor of  $\rme^{\pi \rmi/4}$. This convention has a crucial advantage in studying the representation of the Clifford group and symplectic group, as we shall see in \sref{sec:repsofsp}.
 In general, the Clifford group $\Cli_n$ is generated by Hadamard matrices and phase matrices for  respective qubits, as well as  CNOT gates between all pairs of qubits, where  
\begin{equation}\label{eqn:cnot}
\mathrm{CNOT}=\begin{pmatrix}
1 &0 &0& 0\\
0& 1& 0& 0\\
0& 0 & 0& 1\\
0&0&  1& 0
\end{pmatrix}. 
\end{equation}
It can be proved that the Clifford group $\Cli_n$ generated by these matrices is the normalizer of the Pauli group in $\rmU(d, \bbQ[\rmi])$ \cite{Gott97the,NebeRS01}, where $\bbQ[\rmi]$ is the extension of the rational field $\bbQ$ by the imaginary unit $\rmi$ (thanks to our definition of the Hadamard matrix, we do not need the eighth roots of unity), and $\rmU(d, \bbQ[\rmi])$ is the group of unitary operators in dimension $d$ with entries in $\bbQ[\rmi]$. In addition,  the normalizer of $\mathcal{P}_n$ in $\rmU(d)$ is generated by $\Cli_n$ and phase factors. The center of the Clifford group $\Cli_n$ is the order-4 cyclic group generated by the scalar matrix $\rmi$.

Let $\Sp(2n,\FF_2)$ be the symplectic group composed of all $2n\times 2n$ matrices  $F$ over $\bbF_2$ that satisfy the following equation
\begin{equation}
FJF^\rmT=J.
\end{equation}
For every Clifford unitary $U\in \Cli_n$, there is a unique symplectic matrix $F\in\Sp(2n,\FF_2)$ such that
\begin{equation}\label{eqn:symplecticaction}
U W_a U^\dagger = (-1)^{f(a)} W_{Fa}\qquad \forall a\in \bbF_2^{2n},
\end{equation}
where  $f$ is a  function from $\bbF_2^{2n}$ to $\bbF_2$. Conversely, for each symplectic matrix $F\in\Sp(2n,\FF_2)$ there exists a Clifford unitary $U\in \Cli_n$ and a suitable function $f$ such that the above equation is satisfied. Note that the $4d^2$  Clifford unitaries $\rmi^j U W_a$ for $j=0,1,2,3$ and $a\in \bbF_2^{2n}$ induce the same symplectic transformation. 
  Denote by $\PCli_n$  the projective Clifford group. Then both $\Cli_n/\mathcal{P}_n$ and $\PCli_n/\bar{\mathcal{P}}_n$ are isomorphic to $\Sp(2n,\FF_2)$.

The Clifford group $\Cli_n$ is  a unitary 3-design, but not a 4-design \cite{Zhu15MC,Webb15,KuenG15}. Nevertheless, its fourth frame potential is not far from the value of a 4-design [c.f.~\eref{eq:FramePmin}] according to  the  formula
 \cite{Zhu15MC}
\begin{equation}\label{eq:FpClifford}
 \Phi_4(\Cli_n)=\begin{cases}
 15 & n=1,\\
 29& n=2,\\
 30 & n\geq 3.
 \end{cases}
 \end{equation}
 This observation indicates that the fourth tensor power of the Clifford group has only a few more irreducible components than that of the whole unitary group, which will be spelled out more precisely in the next section.

\emph{Stabilizer codes and states} \cite{Gott97the} are certain subspaces of $\bbC^d$ that are of fundamental importance in quantum information theory.
Among other applications, they form the foundation of the theory of \emph{quantum error correction} \cite{NielC00book}.

A \emph{stabilizer group} is an abelian subgroup of the  Pauli group  that does not contain $-1$. 
A \emph{stabilizer code} is the common $+1$-eigenspace of operators in a stabilizer group \cite{Gott97the,NielC00book}. 
Let $S\subset \mathcal{P}_n$ be a stabilizer group. One can easily verify that
\begin{equation*}
	P = \frac1{|S|} \sum_{W\in S} W
\end{equation*}
is the orthogonal projector onto the stabilizer code associated with the group.
The order of any $n$-qubit stabilizer group is a divisor of $d=2^n$. 
If the stabilizer group has order $2^m$ with $m\leq n$, then the stabilizer code has dimension $2^{n-m}$. 
Those  $n$-qubit stabilizer groups of order $d$ are called \emph{maximal}. 
When the stabilizer group is maximal, the stabilizer code has dimension~1. 
Such codes are commonly referred to as \emph{stabilizer states}. 

Stabilizer codes can be described in terms of the geometry of the discrete symplectic vector space $\FF_2^{2n}$. We mention this connection only briefly -- c.f.~Refs.~\cite{Gros06,gross2013stabilizer,KuenG15} for more details.
Any $n$-qubit stabilizer group
$S$ is of the form
\begin{equation*}
	S = \{ (-1)^{f(a)} W_a \,|\, a \in M \subset \FF_2^{2n} \}
\end{equation*}
for some set $M\subset \FF_2^{2n}$ and some function $f: \FF_n^{2n} \to \FF_2$. 
The fact that $S$ forms a group implies that $M$ is a subspace of $\FF_2^{2n}$.
From the fact that $S$ is abelian and \eref{eqn:pauli_commutation_relation}, it follows 
that the symplectic inner product vanishes on $M$. 
Such subspaces are called \emph{isotropic} in symplectic geometry.
So  there is a close correspondence between stabilizer codes and isotropic subspaces of finite symplectic vector spaces.

\section{Decomposition of the fourth tensor power of the Clifford group}

\subsection{\label{sec:specialcode}A special stabilizer code}

To state our main result, we need to introduce a certain stabilizer code. Whenever $k$ is even, the following set of  Pauli operators
\begin{equation}\label{eqn:stabgroup}
	S_{n,k} = \{ \Ten^k(W_a)\,|\, a\in \FF_2^{2n} \}
\end{equation}
commute with each other. The set is also invariant under the diagonal action of the Clifford group. 
If in addition $k$ is a multiple of $4$, then  $S_{n,k}$ is closed under multiplication and thus forms a stabilizer group. Denote by 
$V_{n,k}$  the stabilizer code defined by the 
joint $+1$-eigenspace of operators in $S_{n,k}$. The dimension of the stabilizer code is $d^{k-2}$, and 
the projector onto it  is given by
\begin{equation}\label{eq:StabCodeProjection}
	P_{n,k} 
	= 
	\frac{1}{|S_{n,k}|} \sum_{a\in \FF_2^{2n}} \Ten^k(W_a)
	= 
	\frac{1}{2^{2n}} \sum_{a\in\FF_2^{2n}} \underbrace{W_a \otimes \dots \otimes W_a}_{k \times}.
\end{equation}
The stabilizer code  $V_{n,k}$ and projector $P_{n,k}$ are invariant under the action of the symmetric group $S_k$, which acts on $\big(\CC^d\big)^{\otimes k}$ by
permuting the $k$ tensor factors. Meanwhile, they are also invariant under the diagonal action of the Clifford group. In other words, $V_{n,k}$ affords a representation of the Clifford group $\Cli_n$. 
Our main result stated in Section~\ref{sec:Main}, in a precise sense, $V_{n,4}$ is the only subspace of $(\CC^d)^{\otimes 4}$ stabilized by $\Cli_n$ but not by the unitary group $U(d)$.

Given that $V_{n,k}$ is a common $+1$ eigenspace
of $\Ten^k(W_a)$ for all Pauli operators $W_a$ and that $\rmi^k=1$ when $k$ is a multiple of 4, it follows that the Pauli group $P_n$ acts trivially on $V_{n,k}$. 
Therefore, $V_{n,k}$ affords a   representation of the symplectic group $\Sp(2n,\bbF_2)$, which is isomorphic to $\Cli_n/\mathcal{P}_n$. The property of this representation is discussed in more detail in \sref{sec:repsofsp}.

In the rest of this section, we construct an orthonormal basis for  $V_{n,k}$, though this is not essential to understanding the main result. First consider the special case $n=1$.
Let $u\in \bbF_2^k$ and define $\tilde{u}:=u+(1,1,\ldots,1)$ as the bitwise "NOT" of $u$. If $k$ is a multiple of 4 and $u$ has even weight (even number of digits equal to 1), then the  vector $|\phi_u\rangle:=(|u\rangle+|\tilde{u}\rangle)/\sqrt{2}$ is a common $+1$-eigenvector of $\tau^k(W_a)$ for all $a\in\bbF_2^k$; that is, $|\phi_u\rangle\in V_{1,k}$. Now it is straightforward to verify that the follow set of vectors
\begin{equation}
\{|\phi_u\rangle  \,|\,\mbox{$u\in \bbF_2^k$ has even weight and  $u_1=0$}\}
\end{equation}
forms an orthonormal basis of $V_{1,k}$. 

Simple analysis shows that $V_{n,k}$ and $P_{n,k}$ can be written as tensor products as follows,
\begin{equation}
V_{n,k}=V_{1,k}^{\otimes n},\qquad P_{n,k}=P_{1,k}^{\otimes n}.
\end{equation}
So an orthonormal basis of  $V_{n,k}$ can be constructed by taking a suitable tensor power  of the above basis of $V_{1,k}$.

\subsection{\label{sec:Main}Main results}

The most concise way to state our main result is in terms of the
\emph{commutant} of $\Ten^4(\Cli_n)$.
\emph{Schur-Weyl duality}  states that the
commutant of $\Ten^k(\rmU(d))$ is generated by the symmetric group  $S_k$ permuting the tensor factors of $(\CC^d)^{\otimes k}$.
If $d=2^n$ and we restrict to 
the subgroup $\Ten^4(\Cli_n)$, the commutant becomes larger. Our main
result says that there is only one additional generator: the
stabilizer projector $P_{n,4}$ introduced above. 
\begin{theorem}[Main Theorem]\label{thm:Main}
	The commutant $\Ten^4(\Cli_n)'$ of the diagonal action of the Clifford
	group on $\big(\CC^d\big)^{\otimes 4}$ is generated as an algebra by $S_4$ (permuting the tensor factors) and  the stabilizer
	projector $P_{n,4}$.
\end{theorem}

Next, we will give a more concrete formulation of the main result.
To this end, recall that Schur-Weyl duality can be used to find the
decomposition
\begin{equation}\label{eq:SchurWeyl4}
	\left(\bbC^{d}\right)^{\otimes 4}
	=
	\bigoplus_{\lambda}H_\lambda =\bigoplus_{\lambda} W_\lambda \otimes S_\lambda
\end{equation}
of $\left(\bbC^{d}\right)^{\otimes4}$ into
irreps of $\rmU(d)\times S_4$. 
Here, the $\lambda$'s are partitions of $4$ into no more than $d$
parts, $W_\lambda$ is the Weyl module carrying an irrep of $\rmU(d)$ and
$S_\lambda$ the Specht module on which $S_4$ acts irreducibly; the
group  $\rmU(d)\times S_4$ acts irreducibly on each $H_\lambda$. 
The dimensions of $S_\lambda$ and $W_\lambda$ are denoted by
$d_\lambda$ and $D_\lambda$, respectively, as listed in
\tref{tab:IrreDim}. 
Note that $d_\lambda$ equals the multiplicity of the Weyl
module $W_\lambda$, and, likewise,  $D_\lambda$ is the multiplicity of
the Specht module $S_\lambda$. 
Let $G$ be a subgroup of $\rmU(d)$,  then the number of irreducible
components of $G\times S_4$ on $H_\lambda$ is
equal to the number of irreducible components of $G$  on $W_\lambda$.
In particular, $G\times S_4$ is irreducible on $H_\lambda$ iff $G$ is
irreducible on $W_\lambda$. 
The multiplicity of each irrep of $G$
appearing in $H_\lambda$ is always a multiple of $d_\lambda$.

Now recall that $V_{n,4}$ is the stabilizer code defined above. 
We denote its orthogonal complement by $V_{n,4}^\perp$ and define
the spaces
\begin{equation*}
H_\lambda^{+} := H_\lambda \cap V_{n,4}, \qquad
H_\lambda^{-} := H_\lambda \cap V_{n,4}^\perp.
\end{equation*}
Because $V_{n,4}$ is invariant under the action of $S_4$, and because
the $S_\lambda$ are irreducible under the same action, it follows that
for each $\lambda$, there is a  subspace $W_\lambda^+ \subset W_\lambda$
such that
\begin{equation*}
H_\lambda^+ = W_\lambda^+ \otimes S_\lambda.	
\end{equation*}
Likewise,
\begin{equation*}
H_\lambda^- = W_\lambda^- \otimes S_\lambda,
\end{equation*}
where $W_\lambda^-$ is the ortho-complement, within $W_\lambda$, of
$W_\lambda^+$. Define  $D_\lambda^{\pm}: = \dim W_\lambda^{\pm}$, then $\dim H_\lambda^{\pm}=d_\lambda D_\lambda^{\pm}$.
A major technical stepping stone for establishing our main result are explicit formulas for the dimensions of these spaces.
\begin{lemma}\label{lem:dimensions}
	The values of $D_\lambda^{\pm}$ for nonincreasing partitions $\lambda$ of $4$
	are given in Table~\ref{tab:IrreDim}. 
\end{lemma}

\begin{table}
	\caption{\label{tab:IrreDim}Dimensions of the Specht modules, Weyl modules, and irreducible components of $\Ten^4(\Cli_n)$ that appear in $(\bbC^d)^{\otimes 4}$, where $d=2^n$.}
	
	\centering
	\begin{math}
	\begin{array}{l|cccc}
	\hline\hline
	\lambda		& d_\lambda & D_\lambda & D_\lambda^+ & D_\lambda^- \\  \hline
	[4]	&1& \frac{d(d+1)(d+2)(d+3)}{24} & \frac{(d+1)(d+2)}{6}  &  \frac{(d-1)(d+1)(d+2)(d+4)}{24}  \\ 
	{[1,1,1,1]}&1& \frac{d(d-1)(d-2)(d-3)}{24}  &  \frac{(d-1)(d-2)}{6} &  \frac{(d+1)(d-1)(d-2)(d-4)}{24}  \\ 
	{[2,2]}	&2& \frac{d^2(d^2-1)}{12}  & \frac{(d^2-1)}{3}  & \frac{(d^2-4)(d^2-1)}{12}  \\ 
	{[2,1,1]}	&3& \frac{d(d-2)(d^2-1)}{8}  & 0 &   \frac{d(d-2)(d^2-1)}{8} \\ 
	{[3,1]}	&3& \frac{d(d+2)(d^2-1)}{8}  & 0 & \frac{d(d+2)(d^2-1)}{8}   \\
	\hline\hline
	\end{array} 
	\end{math}
\end{table}

Let $U\in \Cli_n$ be an element of the  Clifford group. 
Because $\Ten^4(U)$ commutes with both $S_4$ and $P_{n,4}$, it is of the form
\begin{equation*}
	\Ten^4(U) = 
	\bigoplus_{\lambda; s=\pm\,|\,D_\lambda^s\neq 0} U_\lambda^s \otimes \Id_\lambda,
\end{equation*}
where $U_\lambda^s$ acts on $W_\lambda^s$ and $\Id_\lambda$ is the identity on $S_\lambda$. 
Therefore, the spaces $W_\lambda^\pm$ carry representations $U\mapsto U_\lambda^\pm$ of the Clifford group $\Cli_n$.
We can now state a more concrete version of the main theorem.
\begin{proposition}\label{prop:Wirre}
	Whenever they are non-trivial, the spaces $W_{\lambda}^{\pm}$
	carry irreducible and inequivalent representations of the $n$-qubit Clifford group
	$\Cli_n$. 	What is more, 
	under the action of $\Cli_n \times S_4$, the 
	space $\big(\CC^d\big)^{\otimes 4}$ decomposes into irreps as
	\begin{equation*}
		\left(\CC^d\right)^{\otimes 4}
		=
		\bigoplus_{\lambda; s=\pm\,|\,D_\lambda^s\neq 0} W_\lambda^{s} \otimes S_\lambda.
	\end{equation*}
\end{proposition}

We remark that following \rcite{VollW01}, the commutant of $\Ten^4(\Cli_n)$ can easily be mapped to the commutant of certain related representations of $\Cli_n$.
Indeed, consider as a first example the representation
\begin{equation}\label{eqn:conjugaterep}
	\Ten^{(3,1)}: U \mapsto U \otimes U \otimes U \otimes \bar U.
\end{equation}
Then
\begin{equation}\label{eqn:partialTransposeCommutant}
	A \in \Ten^{(3,1)}(\Cli_n)' 
	\Leftrightarrow
	A^{\Gamma_4} \in \Ten^4(\Cli_n)'.
\end{equation}
Here, $A^{\Gamma_4}$ is the \emph{partial transpose} of $A$ with respect to the fourth tensor factor.
It is defined on product matrices as
\begin{equation*}
	A_1 \otimes A_2 \otimes A_3 \otimes A_4
	\mapsto
	A_1 \otimes A_2 \otimes A_3 \otimes (A_4)^\rmT,
\end{equation*}
and extended linearly to the general case. 
The transpose is to be understood in the same basis in which the complex conjugate is taken.
To verify this claim,
note that
\begin{align*}
	&\Ten^{(3,1)}(U)
	(A_1\otimes A_2\otimes A_3\otimes A_4)^{\Gamma_4} 
	\Ten^{(3,1)}(U)^\dagger\\
	&=
	(U \otimes U \otimes U \otimes \bar U)
	(A_1\otimes A_2\otimes A_3 \otimes A_4^\rmT)
	(U^\dagger \otimes U^\dagger \otimes U^\dagger \otimes U^\rmT) \\
	&=
	\left[
	  (U \otimes U \otimes U \otimes U)
	  (A_1\otimes A_2\otimes A_3 \otimes A_4)
	  (U^\dagger \otimes U^\dagger \otimes U^\dagger \otimes U^\dagger) 
	\right]^{\Gamma_4},
\end{align*}
so that
\begin{equation*}
	\Ten^{(3,1)}(U)A^{\Gamma_4}\Ten^{(3,1)}(U)^\dagger   - A^{\Gamma_4} = 0 
	\quad
	\Leftrightarrow 
	\quad
	\left[\Ten^{4}(U)A\Ten^{4}(U)^\dagger   - A\right]^{\Gamma_4} = 0
	\quad
	\Leftrightarrow 
	\quad
	\Ten^{4}(U)A\Ten^{4}(U)^\dagger   - A = 0.
\end{equation*}
An analogous reasoning applies to the representations $\Ten^{(k,l)}$ for general $k,l$. 
Particularly relevant are the representations $\Ten^{(k,k)}$, which are isomorphic to the $k$th tensor power of the \emph{adjoint representation}.
Based on this connection, one could work out the irreducible representations of $\Ten^{(k,l)}(\Cli_n)$ by diagonalizing the commutant.
We have not pursued this route any further in the present paper (but see \cite{helsen2016representations}).  

\subsection{Proof of Main Theorem}

In this section, we  prove \lref{lem:dimensions} and conclude from it our main result.
An alternative proof of \lref{lem:dimensions} 
--
which also yields orthonormal bases for $W_{[4]}^+$ and  $W_{[1^4]}^+$
--
is presented in the appendix.

\begin{table}
	\caption{\label{tab:CharS4}Characters of the symmetric group $S_4$.}
	
	\centering
	\begin{math}
	\begin{array}{l|ccccc}
	\hline\hline
	\mbox{cycle type} & (1^4) & (2^2) & (2,1^2) & (3,1)  & (4) \\ 
	\mbox{order}	  &  1  & 2   & 2 & 3 & 4 \\ 
	\#		& 1 & 3 & 6 & 8 &  6\\ 
	\hline
	\chi_1=[4]	& 1 & 1 & 1 & 1 & 1 \\ 
	\chi_2=[1,1,1,1]	    & 1 & 1 & -1 & 1 & -1 \\ 
	\chi_3=[2,2]		& 2 & 2 & 0 & -1 & 0 \\ 
	\chi_4=[2,1,1]		& 3 & -1 & -1 & 0 & 1 \\ 
	\chi_5=[3,1]		& 3 & -1 & 1 &0  & -1\\
	\hline\hline
	\end{array} 
	\end{math}
\end{table}

\begin{proof}[Proof of \lref{lem:dimensions}]
	Let $H_\lambda, W_\lambda, S_\lambda$ be the representation spaces appearing in the Schur-Weyl decomposition in \eref{eq:SchurWeyl4}. Let $P_\lambda$ be the projector onto  $H_\lambda$. We have 
	\begin{equation}\label{eq:Plambda}
	P_\lambda=\frac{d_\lambda}{24}\sum_{\sigma\in S_4} \chi_\lambda(\sigma) U_\sigma,
	\end{equation}
	where  $U_\sigma$ is the unitary operator that realizes the permutation of the tensor factors corresponding  to $\sigma$, and $\chi_\lambda$ is the character of the irrep of $S_4$ corresponding to the partition $\lambda$; see \tref{tab:CharS4}. 
	For example, the projectors onto the symmetric and antisymmetric subspaces are respectively given by
	\begin{align}
	P_{[4]}=\frac{1}{24}\sum_{\sigma\in S_4}U_\sigma,\qquad
	P_{[1^4]}=\frac{1}{24}\sum_{\sigma\in S_4}\sgn(\sigma)U_\sigma,
	\end{align}
	where $\sgn(\sigma)$ is equal to 1 for even permutations and $-1$ for odd permutations.

	Note that $P_\lambda$  commutes with the projector $P_{n,4}$ onto the stabilizer code, so the dimension of  $H_\lambda^+=V_{n,4}\cap H_\lambda$ is given by $d_\lambda D_\lambda^+=\tr(P_{n,4} P_\lambda)$. Therefore, 
	\begin{align}
	D_\lambda^+&=\frac{1}{d_\lambda}\tr(P_{n,4} P_\lambda)=\frac{1}{d^2d_\lambda}\sum_a\tr(W_a^{\otimes4} P_\lambda)=\frac{1}{d^2}\Bigl[D_\lambda+\frac{1}{24}\sum_{\sigma\in S_4}\sum_{0\neq a\in\bbF_2^{2n}} \chi_\lambda(\sigma) \tr\bigl(U_\sigma W_a^{\otimes 4}\bigr)\Bigr].
	\end{align}
Here the trace $\tr\bigl(U_\sigma W_a^{\otimes 4}\bigr)$ with $a\neq 0$ can be computed using the following simple formula,
	\begin{equation}
	\tr\bigl(U_\sigma W_a^{\otimes 4}\bigr)=
	\begin{cases}
	0 & \mbox{$\sigma$ contains a cycle of odd length},\\
	d^{l(\sigma)} & \mbox{otherwise},
	\end{cases}
	\end{equation}
	where  $l(\sigma)$ is the number of cycles in $\sigma$ that have even lengths.
	According to \tref{tab:CharS4}, the symmetric group $S_4$ has three permutations of cycle type $(2^2)$ and six permutations of cycle type (4), while any other permutation contains at least one cycle of odd length. Now the value of $D_\lambda^+$
	can be computed by virtue of  the above two equations, from which $D_\lambda^-=D_\lambda-D_\lambda^+$ follows immediately, as shown in \tref{tab:IrreDim}.
\end{proof}

\begin{proof}[Proof of \pref{prop:Wirre} and \thref{thm:Main}]
	
	From the discussion in \sref{sec:intro:unitary}, the sum of squared multiplicities of irreducible components of $\tau^4(\rmC_n)$ is equal to the fourth frame potential of the Clifford group $\Cli_n$. 
	For now, we restrict to $n\geq 3$.
	In this case, both $H_\lambda^+$ and $ H_\lambda^-$ are nontrivial invariant subspaces of  $\Cli_n\times S_4$ for $\lambda=[4], [1,1,1,1], [2,2]$. 
	So the frame potential of $\Cli_n$ is at least 
	\begin{equation}
		\Phi_4(\rmC_n)\geq	d_{[4]}^2+d_{[1,1,1,1]}^2+d_{[2,2]}^2+\sum_{\lambda} d_\lambda^2=30,
	\end{equation}
	with equality if and only if all the representations of $\Cli_n$ afforded by $W^\pm_\lambda$ for $D^\pm_\lambda\neq0$ are irreducible and inequivalent. 
	However, we know from \eref{eq:FpClifford} that $\Phi(\rmC_n)$ is indeed equal to $30$ for $n\geq 3$. 
	Thus, equality must hold and we have proved the first part of Proposition~\ref{prop:Wirre}.
	The proofs for the  special cases $n=1,2$ are similar. 

	The second part of Proposition~\ref{prop:Wirre} is a straight-forward combination of the first part with Schur-Weyl duality.

	By the second part of Proposition~\ref{prop:Wirre} and Schur's Lemma, every element $B$ of the commutant of $\Ten^4(\rmC_n)$ is of the form
	\begin{equation*}
		B 
		= 
		\bigoplus_{\lambda; s=\pm\,|\,D_\lambda^s\neq 0} \Id_\lambda^s \otimes B_\lambda^s,
	\end{equation*}
	with $\Id_\lambda^s$ the identity on $W_\lambda^s$ and $B_\lambda^s$ a suitable linear operator on $S_\lambda$.
	Thus
	\begin{align*}
		B
		&=
		P_{n,4} 
		\left(
			\bigoplus_{\lambda |\,D_\lambda^+\neq 0} 
			\Id_\lambda \otimes B_\lambda^+,
		\right)
		+
		(\Id - P_{n,4})
		\left(
			\bigoplus_{\lambda |\,D_\lambda^-\neq 0} 
			\Id_\lambda \otimes B_\lambda^-,
		\right),
	\end{align*}
	where $\Id_\lambda$ is the identity on $W_\lambda$. 
	The expressions in parentheses commute with the diagonal representation of $\rmU(d)$ and are thus, by Schur-Weyl duality, linear combinations of the representation of $S_4$, which  permutes the tensor factors. This proves \thref{thm:Main}.
\end{proof}

\subsection{\label{sec:repsofsp}Representations of the discrete symplectic group}

We have argued in \sref{sec:specialcode} that whenever $k$
is a multiple of 4, the stabilizer code $V_{n,k}$ carries a
representation of the symplectic group $\Sp(2n,\FF_2)$.  
For $k=4$, \pref{prop:Wirre} and \tref{tab:IrreDim} imply that
\begin{equation}\label{eqn:irreps_sp}
	V_{n,4} \simeq 
	W_{[4]}^+ 
	\oplus W_{[1,1,1,1]}^+
	\oplus W_{[2,2]}^+\otimes \CC^2
\end{equation}
gives the decomposition of that stabilizer code into irreps of
$\Sp(2n,\FF_2)$. This decomposition is remarkably similar to the decomposition of the complex Weil character $\zeta_n$ of $\Sp(2n,\FF_2)$ as discussed in  Ref.~\cite[pages 4976---4977]{guralnick2004cross},
\begin{equation}
\zeta_n=\alpha_n+\beta_n+2\zeta_n^1.
\end{equation}
Moreover, the dimensions of  $V_{n,4}$, $W_{[1,1,1,1]}$, $W_{[4]}$, and $W_{[2,2]}$ coincide with the degrees of the Weil  characters $\zeta_n, \alpha_n$, $\beta_n$, and  $\zeta_n^1$,  respectively, according to  Table~\ref{tab:IrreDim} and
Table~I in Ref.~\cite{guralnick2004cross}.
The following proposition reveals the reason behind this coincidence.
\begin{proposition}\label{pro:IrrSymp}
 $V_{n,4}$ carries the \emph{complex Weil representation} of $\Sp(2n,\FF_2)$
with character $\zeta_n$ as defined in  
Ref.~\cite[pages 4976---4977]{guralnick2004cross}.
What is more, the characters of 
$W_{[1,1,1,1]}$, $W_{[4]}$, and $W_{[2,2]}$ are their 
 $\alpha_n$, $\beta_n$, and  $\zeta_n^1$, respectively.
\end{proposition}
When $n\geq 4$, \pref{pro:IrrSymp} follows from  Corollary 6.2 in Ref.~\cite{guralnick2004cross}, which states that  $\alpha_n$, $\beta_n$, and  $\zeta_n^1$ are nontrivial characters of $\Sp(2n,\FF_2)$ of three minimal degrees. 
When $n=3$,  $\alpha_n$ and $\beta_n$  are still characters of  the two minimal degrees \cite{TiepZ96}, but there is another  character of $\Sp(6,\FF_2)$ that has the same degree of 21 as  $\zeta_3^1$. When $n=2$, $\Sp(2n,\FF_2)$ is isomorphic to the symmetric group $S_6$. When $n=1$, $\Sp(2n,\FF_2)$ is isomorphic to the symmetric group $S_3$, in which case $W_{[1,1,1,1]}$ has dimension 0, $W_{[2,2]}$ carries the sign representation of $S_3$, and $W_{[4]}$ carries the unique two-dimensional representation. Here we shall give a simple and uniform proof of \pref{pro:IrrSymp}, which does not rely on Corollary 6.2 in Ref.~\cite{guralnick2004cross}. Moreover, we derive an explicit formula for the character afforded by $V_{n,k}$ and determine the sum of squared multiplicities of irreducible components, assuming $k$ is a multiple of 4.

\begin{lemma}\label{lem:CliffdTrace}
	Suppose $F\in \Sp(2n,\bbF_2)$ and $U_F\in\Cli_n$ is a Clifford unitary that induces the symplectic transformation $F$. If $\tr(U_F)\neq0$, then 
	\begin{equation}
	[\tr(U_F)]^4=(-4)^{\dim (\ker(F-1)) }.
	\end{equation} 
\end{lemma} 
As an implication of this lemma,
\begin{equation}
\tr(U_F)=\sqrt{f(F)}\,\rmi^{j}\times\begin{cases}
1 & 2|\dim (\ker(F-1)),\\
\rme^{\pi\rmi/4} & 2\nmid\dim (\ker(F-1)),
\end{cases}
\end{equation}
where $j=1,2,3,4$, and  $f(F):=2^{\dim (\ker(F-1))}$ is the number of fixed points of $F$ on the symplectic space $\bbF_2^{2n}$ \cite{Zhu15MC}. Note that $f(F)=|\tr(U_F)|^2$ if $U_F$ is not traceless. 
 
\begin{proof}
Recall that the Clifford group $\Cli_n$ is generated by  phase gates and  Hadamard gates of respective qubits as well as CNOT gates between all pairs of qubits
(cf.~\sref{sec:PauliCli}). Therefore, the Clifford unitary $U_F$ has
the form $U_F=A/2^j$, where $j$ is a nonnegative integer, and $A$ is a
matrix each entry of which is a linear combination of 1 and $\rmi$
with integer coefficients. Consequently, $\tr(U_F)$ has the form
\begin{equation}
\tr(U_F)=\frac{a+b\rmi}{2^k},
\end{equation}
where $a,b$ are integers and  $k$ is a nonnegative integer. In addition, we may assume that  the greatest common divisor of $a,b$ is odd if $k>0$.  According to \rcite{Zhu15MC},
\begin{equation} 
|\tr(U_F)|^2=f(F)=2^{\dim (\ker(F-1)) }
\end{equation}
whenever $U_F$ is not traceless. 
 If $k>0$, then at least one of  $a, b$ is odd,  so that $|\tr(U_F)|^2=(a^2+b^2)/4^k$ cannot be an integer. It follows that $k=0$,  $\tr(U_F)=a+b\rmi$, and 
\begin{equation}
a^2+b^2=2^{\dim (\ker(F-1)) }.
\end{equation}
If $\dim (\ker(F-1))$ is odd, then $a^2=b^2=2^{\dim (\ker(F-1))-1}$ given that $a,b$ are integers, so that 
	\begin{equation}
	[\tr(U_F)]^4=-4a^2b^2=(-4)^{\dim (\ker(F-1)) }.
	\end{equation} 
If 	$\dim (\ker(F-1))$ is even, then  $ab=0$ and
	\begin{equation}
	[\tr(U_F)]^4=(a^2+b^2)^2=4^{\dim (\ker(F-1)) }=(-4)^{\dim (\ker(F-1)) }.
	\end{equation}
	
\end{proof}

Recall that  the stabilizer code $V_{n,k}$ carries a
representation of the symplectic group $\Sp(2n,\FF_2)$ whenever $k$
is a multiple of 4.  The following lemma yields an explicit formula for the character of this representation. Note that the same formula also applies to any subgroup of $\Sp(2n,\FF_2)$. 
\begin{lemma}\label{lem:SympChar}
	Suppose $F\in \Sp(2n,\bbF_2)$ and $U_F\in\Cli_n$ is a Clifford unitary that induces the symplectic transformation $F$. If $k$ is a multiple of 4,  then 
		\begin{equation}
		\tr\left(U_F^{\otimes k} P_{n,k}\right)=
		\left[\frac{(-4)^{k/4}}{2}\right]^{\dim (\ker(F-1)) }.
		\end{equation} 
		In particular, 
	\begin{equation}
	\tr\left(U_F^{\otimes 4} P_{n,4}\right)=(-2)^{\dim (\ker(F-1)) }.
	\end{equation} 
\end{lemma} 
\begin{proof}
	\begin{equation}
	\tr\left(U_F^{\otimes k} P_{n,k}\right)=\frac{1}{d^2}\tr\left(U_F^{\otimes k}\sum_a W_a^{\otimes k}\right)
	=\frac{1}{d^2}\sum_a[\tr(U_FW_a)]^k.
	\end{equation}	
Note that the $d^2$ operators $U_FW_a$  for $a\in \bbF_2^{2n}$ induce the same symplectic transformation as $U_F$. In addition, $|\tr(U_FW_a)|^2=2^{\dim (\ker(F-1)) } $ when $U_FW_a$ is not traceless and   $\sum_a|\tr(U_FW_a)|^2=d^2$. So among the $d^2$ operators $U_FW_a$ for $a\in \bbF_2^{2n}$, $2^{2n-\dim (\ker(F-1)) }$ of them are not traceless \cite{Zhu15MC}. Now application of \lref{lem:CliffdTrace} to the above equation yields
	\begin{equation}
	\tr\left(U_F^{\otimes 4} P_{n,4}\right)=2^{-2n}\times
	(-4)^{k\dim (\ker(F-1)) /4}\times 2^{2n-\dim (\ker(F-1)) }=
\left[\frac{(-4)^{k/4}}{2}\right]^{\dim (\ker(F-1)) }.
\end{equation}
\end{proof}

\begin{proof}[Proof of \pref{pro:IrrSymp}]
According to \lref{lem:SympChar}, the character afforded by $V_{n,4}$ coincides with $\zeta_n$ discussed in Ref.~\cite{guralnick2004cross}. Consequently, $V_{n,4}$ decomposes into the same irreps as $\zeta_n$. Since the character afforded by   $W_{[2,2]}$ has multiplicity 2, it must correspond to $\zeta_n^1$. Comparison of  the dimensions shows that  the characters of 
$W_{[1,1,1,1]}$ and $W_{[4]}$ correspond to
$\alpha_n$ and  $\beta_n$,  respectively.
\end{proof}

Suppose $R\leq \Sp(2n,\bbF_2)$ and let $G_R$ be the preimage in $\Cli_n$ of $R$ under the homomorphism $\Cli_n/\mathcal{P}_n$. 
Denote by $M_k(R)$ the sum of squared multiplicities of the representation of $R$ or $G_R$ afforded by 
the stabilizer code $V_{n,k}$, assuming $k$ is a multiple of 4.  Then 
 $M_k(R)$ may also be understood as the contribution of $V_{n,k}$ to the $k$th frame potential  $\Phi_k(G_R)$ of $G_R$.

\begin{lemma}\label{lem:SympMul}
	Suppose $R\leq \Sp(2n,\bbF_2)$ and $G_R$ is the preimage in $\Cli_n$ of $R$ under the homomorphism $\Cli_n/\mathcal{P}_n$. 
If $k$ is a multiple of 4, then 
\begin{equation}\label{eq:SympMul}
M_k(R)=\frac{1}{|R|}\sum_{F\in R} f(F)^{k-2}=\Phi_{k-1}(G_R).
\end{equation} 
Moreover, $M_k(R)$ is equal to the number of orbits of	$R$ on $(\bbF_2^{2n})^{\times(k-2)}$.
\end{lemma} 
Surprisingly,  the contribution of $V_{n,k}$ to the $k$th frame potential of $G_R$ is equal to the $(k-1)$th frame potential of $G_R$.

\begin{proof}
According to \lref{lem:SympChar},
\begin{equation}
M_k(R)=\frac{1}{|R|}\sum_{F\in R}\left[\frac{(-4)^{k/4}}{2}\right]^{2\dim (\ker(F-1)) }=\frac{1}{|R|}\sum_{F\in R}2^{(k-2)\dim (\ker(F-1)) }=
\frac{1}{|R|}\sum_{F\in R} f(F)^{k-2}.
\end{equation} 
This proves the first equality in \eref{eq:SympMul}; the second equality follows from Lemma~1 in \rcite{Zhu15MC}. Now according to the well-known orbit-stabilizer relation,  $M_k(R)$ is equal to the number of orbits 	of $R$ on $(\bbF_2^{2n})^{\times (k-2)}$.
\end{proof}
Observing  that $\Sp(2n,\bbF_2)$ has five orbits on $(\bbF_2^{2n})^{\times 2}$ when $n=1$ and six orbits  when $n\geq2$ \cite{Zhu15MC}, we conclude that
\begin{equation}
M_4(\Sp(2n,\bbF_2))=\begin{cases}
5 &n=1,\\
6 &n\geq2,
\end{cases}
\end{equation}  
which agrees with the decomposition in \eref{eqn:irreps_sp}.  
A subgroup $R$ of $\Sp(2n,\bbF_2)$ has the same decomposition
on $V_{n,4}$ iff $R$ has the same number of orbits on  $(\bbF_2^{2n})^{\times 2}$ as $\Sp(2n,\bbF_2)$. This condition  is equivalent to the condition that $G_R$ forms a unitary 3-design \cite{Zhu15MC}.  Technically, this means that $R$ is 2-transitive on $\bbF_2^{2n*}$ when $n=1$ and is a rank-3 permutation group when $n\geq2$ \cite{DixoM96book,Came99book, Zhu15MC}, where $\bbF_2^{2n*}$ is the set of nonzero vectors in $\bbF_2^{2n}$. However, there is no proper subgroup of $\Sp(2n,\bbF_2)$ with this property except when $n=2$, in which case there is a unique counterexample \cite{CameK79,CameK02,Zhu15MC}. Therefore, any proper subgroup of $\Sp(2n,\bbF_2)$ with $n\neq2$ has more irreducible components in $V_{n,4}$ (and also in $V_{n,k}$ as a consequence); in other words, at least one of the characters $\alpha_n, \beta_n, \zeta_n^1$ becomes reducible when restricted to a proper subgroup. Similarly,  any proper subgroup of $\PCli_n$ with $n\neq 2$ has more irreducible components on $V_{n,4}$
than $\PCli_n$, and at least one of the representations $W_{[1,1,1,1]}$, $W_{[4]}$, and $W_{[2,2]}$ becomes reducible when restricted to a proper subgroup of $\PCli_n$.

\section{$t$-designs from Clifford orbits}
\label{sec:designsfromorbits}
In this section  we   determine all Clifford covariant $t$-designs in the case of a single qubit. We then show that random orbits of the Clifford group in general are very good approximations to 4-designs. Furthermore, we introduce several simple and efficient methods for constructing exact 4-designs and approximations with arbitrarily high precision from Clifford orbits.

\subsection{\label{sec:qubit}Clifford covariant $t$-designs for one qubit}
In the case of $n=1$, the $t$-partite symmetric subspace has dimension $t+1$, so the 
frame potential of a qubit  $t$-design is equal to $1/(t+1)$. Since the Clifford group is a unitary 3-design, every orbit of the Clifford group forms a complex projective 3-design. The unique shortest orbit is composed of six stabilizer states, which form a complete set of mutually unbiased bases. When represented on the Bloch sphere, the six states form the vertices of the octahedron.

To derive a simple criterion on the orbit that forms a 4-design, suppose the fiducial state has Bloch vector $(x,y,z)$ with $x^2+y^2+z^2=1$. Then the fourth frame potential of the Clifford orbit is given by
\begin{equation}
\Phi_4(x,y,z)=\frac{21-6(x^4+y^4+z^4)+5(x^4+y^4+z^4)^2}{96}.
\end{equation}
The orbit forms a 4-design iff $x^4+y^4+z^4=3/5$, in which case $\Phi_4(x,y,z)$ attains the minimum of $1/5$.  
One explicit solution is given by
\begin{equation}
x=\sqrt{\frac{5+2\sqrt{10}}{15}}, \quad y=z=\sqrt{\frac{5-\sqrt{10}}{15}}.
\end{equation}
It turns out that the orbit forms a 5-design under the same condition; that is, a Clifford orbit forms a 5-design iff it forms a 4-design. 
As explained in Section~\ref{sec:harmonic}, this may not be a coincidence.
By contrast, $\Phi_4(x,y,z)$ is maximized when $x^4+y^4+z^4=1$, in which case the Bloch vector corresponds to a stabilizer state.

When the condition $x^4+y^4+z^4=3/5$ is satisfied, the sixth and seventh frame potentials satisfy the following equation
\begin{align}
8\Phi_7(x,y,z)-1&=4[7\Phi_6(x,y,z)-1]=\frac{11(1-21x^2+105x^4-105x^6)}{2400}\nonumber\\
&=\frac{11(1-21y^2+105y^4-105y^6)}{2400}
=\frac{11(1-21z^2+105z^4-105z^6)}{2400}\nonumber\\
&=\frac{11[3-7(x^6+y^6+z^6)}{480}.
\end{align}
The orbit forms a 6-design  iff $x^2,y^2,z^2$ are distinct roots of the cubic equation $1-21u+105u^2-105u^3=0$, which are given by 
\begin{equation}
u_j=\frac{1}{3}\Bigl(1+2\sqrt{\frac{2}{5}}\cos\frac{\theta+2j\pi}{3}\Bigr), \quad \theta=\arctan\frac{3\sqrt{10}}{20},\quad j=1,2,3.
\end{equation}
 Equivalently, the orbit forms a 6-design iff $x^6+y^6+z^6=3/7$ or if $x^2y^2z^2=1/105$ (assuming $x^4+y^4+z^4=3/5$).
The same condition also guarantees that the orbit forms a 7-design.
There are 48 solutions in total, which compose two Clifford orbits. When represented on the Bloch sphere, the two orbits can be converted to each other by inversion. The two orbits are not unitarily equivalent, but are equivalent under antiunitary transformations. Actually, the 48 solutions form one orbit under the action of the extended Clifford group \cite{Appl05}, the group generated by the Clifford group and complex conjugation with respect to the computational basis. 
Since any qubit 8-design  has at least 25 elements according to \eref{eq:tdesignLBound}, no Clifford orbit can form an 8-design.

Calculation shows that a random Clifford orbit is approximately a $t$-design for $t$ up to 7. If $(x,y,z)$ is distributed uniformly on the Bloch sphere, 
then the ratio of the average frame potential over the minimum potential is given by
\begin{equation}
(t+1)\rmE[\Phi_t(x,y,z)]={\displaystyle \begin{cases}
1& t=3,\\
\frac{127}{126}& t=4,\\
\frac{43}{42}& t=5,\\
\frac{1795}{1716}& t=6,\\
\frac{1381}{1287} &t=7.\\
\end{cases}}
\end{equation}

\subsection{\label{sec:randomOrbits}Random Clifford orbits are good approximations to 4-designs}

In this section we show that random Clifford orbits are very good approximations to projective 4-designs. 
Recall that  $\Ten^4(\Cli_n)$ has two irreducible components $W^{\pm}_{[4]}$ in  the totally symmetric space $W_{[4]}=\Sym_4(\CC^d)$.  According to \tref{tab:IrreDim},
the dimensions of $W_{[4]}$ and $W^{\pm}_{[4]}$ are 
\begin{equation}
\begin{aligned}
D_{[4]}=& \frac{d(d+1)(d+2)(d+3)}{24},\\
D_+:=& D^+_{[4]}=\frac{(d+1)(d+2)}{6},\\
D_-:=& D^-_{[4]}=\frac{(d-1)(d+1)(d+2)(d+4)}{24}.
\end{aligned}
\end{equation}
The projectors $P_{\pm}$   onto the two irreps $W^{\pm}_{[4]}$ read
\begin{equation}
P_+=P_{n,4}P_{[4]},\qquad P_-=(1-P_{n,4})P_{[4]}.
\end{equation}
where $P_{n,4}$ is the projector onto the stabilizer code $V_{n,4}$ given in \eref{eq:StabCodeProjection} and $P_{[4]}$ is the projector onto $W_{[4]}$.

As an implication of \thref{thm:Main} or  Proposition~\ref{prop:Wirre}, we have 
\begin{corollary}\label{cor:4design}	
  Let $\orb(\psi)$ be the orbit of a vector $\psi\in\CC^{2^n}$ under the action of the Clifford group $\PCli_n$. Then 
	  \begin{equation*}
	  \frac{1}{|\orb(\psi)|}\sum_{\phi\in \orb(\psi)} \big(|\phi\rangle \langle\phi|\big)^{\otimes 4}
	  = \beta_+ P_+ + \beta_- P_-,
	  \end{equation*}
	  where $\beta_+$ and $\beta_-$ satisfy
  \begin{equation}
  \beta_+=\frac{1}{D_+}\tr \bigl[P_+ (|\psi\rangle \langle \psi|)^{\otimes 4}\bigr]=\frac{1}{D_+}\tr \bigl[P_{n,4} (|\psi\rangle \langle \psi|)^{\otimes 4}\bigr],\qquad D_+\beta_++D_-\beta_-=\|\psi\|_2^8.
  \end{equation}
  A  normalized  vector $\psi$  is a fiducial vector of  a 4-design iff 
  \begin{equation}\label{eqn:betacond}
	  \beta_-=\beta_+=1/D_{[4]}.
  \end{equation}
\end{corollary}

In what follows, we will investigate the condition (\ref{eqn:betacond}) from various points of view. 
To this end, we introduce a number of related measures.

Define the \emph{characteristic function} (c.f.\ e.g.\ \rscite{WallM94book,Gros06})  $\Xi(\psi)$ as the vector composed of the  $d^2$ elements
\begin{equation}\label{eqn:characteristic}
	\Xi_a(\psi)=\tr(W_a |\psi\rangle \langle \psi|).
\end{equation}
Recall that the $\ell_p$-norm of a vector is the $p$-th root of the sum of the $p$-th powers of its elements.
For our study, the $\ell_4$-norm of the characteristic function
\begin{equation*}
	\left\|\Xi(\psi)\right\|_{\ell_4}^4 
	= \sum_{a\in\FF_2^{2n}} \left| \tr(W_a |\psi\rangle \langle \psi|)\right|^4
\end{equation*}
turns out to be particularly important.
It follows directly from the definition of $P_{n,4}$ and the symmetry of $\psi^{\otimes 4}$ that
\begin{equation*}
  \alpha_+(\psi)
  :=
  \tr \bigl[ P_+ (|\psi\rangle \langle \psi|)^{\otimes 4}\bigr]
  =
  \tr \bigl[ P_{n,4} (|\psi\rangle \langle \psi|)^{\otimes 4}\bigr]
  =
  \frac{1}{d^2} \left\|\Xi(\psi)\right\|_{\ell_4}^4.
\end{equation*}
We also set
\begin{equation}\label{eq:ONdeviation}
	\epsilon(\psi)
	:=
	\frac{D_{[4]}}{D_+}\alpha_+(\psi)-\|\psi\|_2^8=\frac{d(d+3)}{4}\alpha_+(\psi)-\|\psi\|_2^8.
\end{equation}
The condition (\ref{eqn:betacond}) for a normalized vector $\psi$ to be a 4-design fiducial can now be re-cast  in three  equivalent forms, 
\begin{equation}\label{eqn:4designconds}
\left\|\Xi(\psi)\right\|_{\ell_4}^4=\frac{4d}{(d+3)},\qquad	\alpha_+(\psi)= \frac{D_+}{D_{[4]}}=\frac{4}{d(d+3)},
	\qquad
	\epsilon(\psi)=0.
\end{equation}

From now on, we will assume the normalization condition $\|\psi\|_2=1$.
Then, $\epsilon$ quantifies the deviation of the Clifford orbit of $\psi$ from a 4-design.
More precisely,  $|\epsilon(\psi)|$ is the operator norm of the deviation
\begin{equation}
  \frac{D_{[4]}}{|\orb(\psi)|}\sum_{\phi\in \orb(\psi)} \big(|\phi\rangle \langle\phi|\big)^{\otimes 4}
  -P_{[4]},
\end{equation}
while $2 D_+|\epsilon(\psi)|=(d+1)(d+2)|\epsilon(\psi)|/3$ is the trace norm (or nuclear norm) of the deviation.
In addition, $\epsilon(\psi)$ determines the fourth frame potential of the Clifford orbit as follows,
\begin{equation}\label{eq:FPfromDev}
\Phi_4(\orb(\psi))=\frac{\alpha_+(\psi)^2}{D_+}+\frac{\alpha_-(\psi)^2}{D_-}
=\frac{1}{D_{[4]}}\left[1+\frac{D_+\epsilon(\psi)^2}{D_-}\right]
=\frac{1}{D_{[4]}}\left[1+\frac{4\epsilon(\psi)^2}{(d-1)(d+4)}\right],
\end{equation}
where $\alpha_-(\psi)=1-\alpha_+(\psi)$.

We now turn to clarifying the extremal and typical values of the functions $\alpha_+$ and $\epsilon$.
To this end, note that since  $\{W_a\}$ forms a nice error basis and Hermitian operator basis, we have 
\begin{equation}
\|\Xi(\psi)\|_{\ell_2}^2=\sum_a \tr\bigl[W_a^{\otimes 2} (|\psi\rangle \langle \psi|)^{\otimes 2}\bigr]=d, \qquad  \|\Xi(\psi)\|_{\ell_\infty}=\max_a|\tr(W_a |\psi\rangle \langle \psi|)|\leq 1.
\end{equation}
Consequently,
\begin{equation}\label{eq:Xibound}
\frac{2d}{d+1}\leq\|\Xi(\psi)\|_{\ell_4}^4\leq d,
\end{equation}
which are equivalent to the following inequalities
\begin{equation}\label{eq:Devbound}
\frac{2}{d(d+1)}\leq\alpha_+(\psi)\leq \frac{1}{d}, \qquad -\frac{d-1}{2(d+1)} \leq \epsilon(\psi)\leq\frac{d-1}{4}.
\end{equation}
The upper bound in \eref{eq:Xibound} follows from the H\"older inequality; it is saturated iff $\Xi(\psi)$ has $d$ entries equal to 1 and all other entries equal to 0; this can happen iff
 $\psi$ is a stabilizer state (c.f.~\lref{lem:ProjBound} in the appendix).
The lower bound is saturated iff 
\begin{equation}
\Xi_a(\psi)=\frac{1}{\sqrt{d+1}}\quad \forall a\neq0,
\end{equation}
in which case the $d^2$  vectors $W_a |\psi\rangle$ for $a\in \bbF_2^{2n}$ define a symmetric informationally complete measurement (SIC) \cite{Zaun11,ReneBSC04,ScotG10,ApplFZ15G}, which happens to be a minimal 2-design \cite{Scot06}.  It is known that SIC fiducial vectors of the $n$-qubit Pauli group cannot  exist except for  $n=1,3$ \cite{GodsR09}, so the lower bounds in \esref{eq:Xibound} and \eqref{eq:Devbound} cannot be saturated except for  $n=1,3$. As an implication of \esref{eq:FPfromDev} and \eqref{eq:Devbound},
the frame potential satisfies
\begin{equation}\label{eq:FPbound}
\frac{1}{D_{[4]}}\leq\Phi_4(\orb(\psi))\leq \frac{1}{D_{[4]}}\left(1+\frac{d-1}{4(d+4)}\right),
\end{equation}
where the lower bound is saturated iff the orbit forms a 4-design, and the upper bound is saturated iff  $\psi$ is a stabilizer state.

Next, we  show that random Clifford orbits are very good approximations to 4-designs. To this end, we first compute the  variance of the deviation parameter $\epsilon(\psi)$. 
Suppose  $\psi$ is distributed according to the uniform measure. Then the first and second moments of $\alpha_+(\psi)$ are given by
\begin{align}
\rmE[\alpha_+(\psi)]&=\tr(P_{n,4} \rmE[(|\psi\rangle \langle \psi|)^{\otimes 4}])=\frac{1}{D_{[4]}}\tr(P_{n,4} P_{[4]})=\frac{4}{d(d+3)},\\
\rmE[\alpha_+(\psi)^2]&=\frac{1}{D_{[8]}}\tr(P_{n,4}^{\otimes2} P_{[8]})=\frac{16(d^2+15d+68)}{d^2(d+3)(d+5)(d+6)(d+7)},\label{eq:SecondMoment}
\end{align}
where the last equality was derived  in  \aref{sec:SecondMoment}. 
The variance of $\alpha_+(\psi)$ reads
\begin{equation}
\begin{aligned}
\var[\alpha_+(\psi)]&=\rmE[\alpha_+(\psi)^2]-\rmE[\alpha_+(\psi)]^2=\frac{96(d-1)}{d^2(d+3)^2(d+5)(d+6)(d+7)}.
\end{aligned}
\end{equation}
As an immediate consequence,
\begin{equation}\label{eq:4designDeviation}
\rmE[\epsilon(\psi)]=0,\qquad 
\rmE[\epsilon(\psi)^2]=\frac{\var[\alpha_+(\psi)]}{\rmE[\alpha_+(\psi)]^2}=\frac{6(d-1)}{(d+5)(d+6)(d+7)}.
\end{equation}
Since the function $\epsilon(\psi)$ is continuous, the equality $\rmE[\epsilon(\psi)]=0$ guarantees the existence of a root (actually many roots) of $\epsilon(\psi)$, so  exact fiducial vectors of 4-designs always exist.  In addition, \eref{eq:4designDeviation} shows that  the typical value of $|\epsilon(\psi)|$ is around $\sqrt{6}/d$ when $d$ is large, which is  much smaller then the upper bound $(d-1)/4$. 
Application of  the Chebyshev inequality further  implies that
\begin{equation}\label{eq:DevChebyshevBound}
\mathrm{Prob}\{|\epsilon(\psi)|\geq \xi\}\leq \frac{6(d-1)}{(d+5)(d+6)(d+7)\xi^2} \quad \forall \xi> 0.
\end{equation}
 For example, 
\begin{equation}
\mathrm{Prob}\{|\epsilon(\psi)|\geq 1/2\}=\mathrm{Prob}\{\epsilon(\psi)\geq 1/2\}=\mathrm{Prob}\left\{\alpha_+(\psi)\geq \frac{6}{d(d+3)}\right\}\leq \frac{24(d-1)}{(d+5)(d+6)(d+7)}.
\end{equation}
This particular bound is of interest to studying the distinguishability of quantum states under measurements constructed from Clifford orbits. 
In \rcite{PovmNorm16}, it was shown that Clifford orbits of $\psi$ with $\alpha_+(\psi)\geq 6/d(d+3)$ can achieve almost the same POVM norm constants as 4-designs. Therefore, random Clifford orbits are very good approximations to 4-designs in this concrete setting.

In conjunction with \Eref{eq:FPfromDev}, we can also  determine
the ratio of the average fourth frame potential over the minimum frame potential (the potential for a 4-design) and bound the probability of large deviation,
\begin{align}
&D_{[4]}\rmE[\Phi_4(\orb(\psi))]=1+\frac{4}{(d-1)(d+4)}\rmE[\epsilon(\psi)^2]=1+\frac{24}{(d+4)(d+5)(d+6)(d+7)},\\
&\mathrm{Prob}\left\{D_{[4]}\rmE[\Phi_4(\orb(\psi))]\geq 1+\frac{4 \xi^2}{(d-1)(d+4)}\right\}\leq \frac{6(d-1)}{(d+5)(d+6)(d+7)\xi^2}.
\end{align}

In the rest of this section we derive another   large-deviation bound based on   Levy's lemma \cite{Ledo01}.
\begin{lemma}[Levy]
 Let $f : S^{2d-1} \rightarrow  \bbR$ be Lipschitz-continuous with Lipschitz constant $\eta$, that is, 
 \begin{equation}
 |f(x)-f(y)|\leq \eta \|x-y\|,
 \end{equation}
 where $\|x-y\|$ is the Euclidean norm in the surrounding space $\bbR^{2d}$ of $S^{2d-1}$.
 Suppose  $x$ is drawn randomly according to the uniform measure on the sphere $S^{2d-1}$. Then 
 \begin{equation}
 \mathrm{Prob}\{|f(x)-\rmE[f(x)]|\geq \xi\}\leq 2\exp\Bigl(\frac{-d\xi^2}{9\pi^3\eta^2}\Bigr)\quad \forall \xi\geq 0.
 \end{equation}
\end{lemma}

\begin{lemma}\label{lem:LipConstant}
	The functions $\alpha_+(\psi)$ and $\epsilon(\psi)$ are  Lipschitz-continuous with Lipschitz constants $5.4/d$ and $5.4(d+3)/4$, respectively, that is,
\begin{equation}
|\alpha_+(\psi)-\alpha_+(\varphi)|\leq  \frac{5.4}{d}\|\psi-\varphi\|,\qquad |\epsilon(\psi)-\epsilon(\varphi)|\leq  \frac{5.4(d+3)}{4}\|\psi-\varphi\|.
\end{equation}
\end{lemma}
This lemma is proved in the appendix.  Note that the second inequality is an immediate consequence of the first one and \eref{eq:ONdeviation}.
We guess that the two Lipschitz constants can be improved to $1/d$ and $d/4$, respectively. 
The following proposition is an immediate consequence of \lref{lem:LipConstant} and Levy's lemma. 
\begin{proposition}
Suppose  $\psi$ is drawn randomly according to the uniform measure on the complex sphere $\bbC^d$. Then 
 \begin{align}
 \mathrm{Prob}\{|\alpha_+(\psi)-\rmE[\alpha_+(\psi)]|\geq \xi\}\leq 2\exp\Bigl(-\frac{d^3 \xi^2}{8138}\Bigr),\qquad
  \mathrm{Prob}\{|\epsilon(\psi)|\geq \xi\}\leq 2\exp\Bigl(-\frac{d \xi^2}{509(d+3)^2}\Bigr) \quad \forall \xi\geq 0.
\end{align}
\end{proposition}
Here the bound on $\mathrm{Prob}\{|\epsilon(\psi)|\geq \xi\}$ is tighter than that given in \eref{eq:DevChebyshevBound} only when $\epsilon(\psi)$ is very large, that is, $\epsilon(\psi)\gg \sqrt{d}$.

Although random Clifford orbits are good approximations to 4-designs with respect to a number of measures, such as the frame potential and operator-norm deviation. They are not good enough according to certain other measures. For example, the second moment of the trace norm deviation is given by  
\begin{equation}
(2D_+)^2\rmE[\epsilon(\psi)^2]=\frac{2(d-1)(d+1)^2(d+2)^2}{3(d+5)(d+6)(d+7)},
\end{equation}
where the equality follows from \eref{eq:4designDeviation}. When $d$ is large, the typical deviation with respect to the trace norm is around $\sqrt{2/3}\,d$, while it is  desirable that the deviation does not grow with the dimension for some applications. This observation motivates us to search for exact 4-designs or  approximations with higher precision.

\subsection{Fiducial vectors of exact 4-designs up to five qubits}
In this section we  propose a method for constructing  exact fiducial vectors of 4-designs of the Clifford group.  Solutions up to five qubits are presented explicitly.

Recall that an $n$-qubit state vector $\psi$ is a fiducial vector of a 4-design iff $\|\Xi(\psi)\|_{\ell_4}^4=4d/(d+3)$; see \eref{eqn:4designconds}.
Suppose $\psi=\psi_1\otimes \psi_2$ is a tensor product of an  $n_1$-qubit state vector and an $n_2$-qubit state vector with $n_1+n_2=n$. Then $\|\Xi(\psi)\|_{\ell_4}^4=\|\Xi(\psi_1)\|_{\ell_4}^4 \|\Xi(\psi_2)\|_{\ell_4}^4$ since
$P_{n,4}$  decomposes in the same way   $P_{n,4}=P_{n_1,4}\otimes P_{n_2,4}$.
In the case of a single qubit, let   $\psi(x,y,z)$ be a fiducial vector  with Bloch vector $(x,y,z)$ with $x^2+y^2+z^2=1$; then
\begin{equation}
\|\Xi(\psi)\|_{\ell_4}^4=1+x^4+y^4+z^4.
\end{equation}
The vector generates a 4-design iff $x^4+y^4+z^4=3/5$ as pointed out in \sref{sec:qubit}. Let $\psi_\rmT$ be the magic state with Bloch vector $(1,1,1)/\sqrt{3}$ \cite{BravK05} (which is also a SIC fiducial). Then fiducial vectors of 4-designs for $n=2,3,4$ can be constructed as follows,
\begin{equation}
\begin{cases}
\psi_\rmT\otimes \psi(x,y,z),\quad x^4+y^4+z^4=5/7,  & n=2; \\
\psi_\rmT^{\otimes2}\otimes \psi(x,y,z),\quad x^4+y^4+z^4=7/11,  & n=3;\\
\psi_\rmT^{\otimes3}\otimes \psi(x,y,z),\quad x^4+y^4+z^4=8/19,  & n=4.
\end{cases}
\end{equation}
Many  other constructions are also available.

In dimension 8, the set of Hoggar lines forms a SIC that is covariant with respect to the three-qubit Pauli group \cite{Hogg98, Zaun11, Zhu15S}. One fiducial vector of the SIC is given by
\begin{equation}\label{eq:HoggarLines}
\psi_{\rm{Hog}}=\frac{1}{\sqrt{6}}(1+\rmi,0,-1,1,-\rmi,-1,0,0)^\rmT.
\end{equation}
According to \eref{eq:Xibound},  $\|\Xi(\psi_{\rm{Hog}})\|_{\ell_4}^4=16/9$ attains the minimum over all three-qubit state vectors.
This observation enables us to construct fiducial vectors of 4-designs for $n=4,5$,
\begin{equation}
\begin{cases}
\psi_{\rm{Hog}}\otimes \psi(x,y,z),\quad x^4+y^4+z^4=17/19,  & n=4;\\
\psi_{\rm{Hog}}\otimes \psi_\rmT\otimes \psi(x,y,z),\quad x^4+y^4+z^4=8/19,  & n=5.
\end{cases}
\end{equation}

The tensor-product construction of fiducial vectors of 4-designs also has a limitation. Consider tensor powers of $\psi_\rmT$ and $\psi_{\rm{Hog}}$ for example,
\begin{equation}
\begin{aligned}
&\|\Xi(\psi_{\rmT}^{\otimes n})\|_{\ell_4}^4=\Bigl(\frac{4}{3}\Bigr)^n.\\
&\|\Xi(\psi_{\rm{Hog}}^{\otimes n/3})\|_{\ell_4}^4=\Bigl(\frac{16}{9}\Bigr)^{n/3}=\Bigl(\frac{4}{3}\Bigr)^{2n/3},\quad 3|n.
\end{aligned}
\end{equation}
As $n$ increases, $\|\Xi(\psi_{\rmT}^{\otimes n})\|_{\ell_4}^4$ and $\|\Xi(\psi_{\rm{Hog}}^{\otimes n/3})\|_{\ell_4}^4$  increase exponentially with $n$. By contrast, the value required for a 4-design approaches the constant 4. The following proposition clarifies this limitation; see \aref{sec:TensorFiducial} for a proof. 
\begin{proposition}\label{pro:TensorFiducial}
Suppose a 4-design fiducial vector of  the $n$-qubit Clifford group is a tensor product of $m\geq2$ vectors $\psi=\otimes_{j=1}^m\psi_j$,  where $\psi_j$ is an $n_j$-qubit state vector with $\sum_j n_j=n$ and $n_1\geq n_2\geq \cdots \geq n_m$. Then $m\leq 3$
except when $n=4$ and $n_1=n_2=n_3=n_4=1$. If $m=3$, then $n_2=n_3=1$, except when $(n_1,n_2,n_3)=(2,2,1)$ or $(n_1,n_2,n_3)=(3,2,1)$.
\end{proposition}
More explicitly, this proposition implies that $(n_1, n_2, \ldots, n_m)$ can only admit one of the following forms $(1,1,1,1)$, $(3,2,1)$, $(2,2,1)$ $(n_1,1,1)$, and $(n_1, n_2)$.

\subsection{Algorithms for constructing  projective 4-designs}
In this section we present two algorithms
for constructing fiducial vectors of 4-designs. Also presented is a method for constructing  exact weighted 4-designs from two Clifford orbits.

Let $\psi$ be an $n$-qubit state vector.
 Recall that $\psi$ is a fiducial vector of a 4-design iff $\|\Xi(\psi)\|_{\ell_4}^4=4d/(d+3)$ or, equivalently, iff $\epsilon(\psi)=0$; cf.~\eref{eqn:4designconds}. 
 
 The first algorithm is based on the tensor-product construction discussed in the previous section.
 
 \textbf{Algorithm 1:}
 \begin{enumerate}
 	
 	\item Generate an $(n-1)$-qubit state vector $\psi_{n-1}$ such that $\|\Xi(\psi_{n-1})\|_{\ell_4}^4\leq  3d/(d+3)$, where $d=2^n$.

 	\item Let $c=4d/\left[(d+3)\|\Xi(\psi_{n-1})\|_{\ell_4}^4\right]$. Choose a qubit state vector $\psi$  with Bloch vector $(x,y,z)$ which satisfies $x^4+y^4+z^4=c-1$. Then $\psi_{n-1}\otimes \psi$ is a fiducial vector of a 4-design.

 \end{enumerate}
 	The vector required in Step~2 can always be found since  $1/3\leq c-1\leq 2(d+2)/(d+3)-1<1$, given that $2d/(d+2)\leq \|\Xi(\psi_{n-1})\|_{\ell_4}^4\leq  3d/(d+3)$, 
 	where the lower bound follows from \eref{eq:Devbound}.
 	
 	In general, it is still not clear whether there exists an $(n-1)$-qubit state vector   $\psi_{n-1}$ which satisfies $\|\Xi(\psi_{n-1})\|_{\ell_4}^4\leq  3d/(d+3)$, but we believe that the answer is positive. Actually, any eigenstate of a Singer unitary might satisfy the requirement; see the next section. In addition, one may try to minimize $\|\Xi(\psi_{n-1})\|_{\ell_4}^4$ numerically as in the search of SICs \cite{ReneBSC04,ScotG10}. Note that here the task is much simpler since the target $3d/(d+3)$ is much larger than the value $2d/(d+1)$ required for a SIC.

 Given two $n$-qubit state vectors $\psi_1,\psi_2$ with $\epsilon(\psi_1)>0$ and $\epsilon(\psi_2)<0$, then any continuous curve of state vectors connecting $\psi_1$ and $\psi_2$ contains a 4-design fiducial vector.  The following bisection algorithm is based on this simple observation. Suppose $\epsilon_0$ is the target precision.

\textbf{Algorithm 2:}
\begin{enumerate}
\item Generate two state vectors $\psi_1,\psi_2$ such  that $\epsilon(\psi_1)>0$, $\epsilon(\psi_2)<0$, and $\langle\psi_1|\psi_2\rangle\neq0$. Choose suitable phase factors so  that $\langle\psi_1|\psi_2\rangle>0$.

\item  Let $\psi'_3=(\psi_1+\psi_2)/2$ and $\psi_3=\psi'_3/\sqrt{\langle\psi'_3|\psi'_3\rangle}$. Stop if $|\epsilon(\psi_3)|\leq \epsilon_0$.

\item If $\epsilon(\psi_3)\geq 0$, then replace $\psi_1$ with $\psi_3$; otherwise, replace $\psi_2$ with $\psi_3$. Repeat Steps 2,3.

\end{enumerate}
\begin{remark}
A candidate for $\psi_1$ is any stabilizer state, while a potential candidate for $\psi_2$ is  an eigenstate of a Singer unitary introduced in the next section. In Step 2 we may also use a weighted sum of $\psi_1,\psi_2$, say
\begin{equation}
\psi'_3=\frac{\epsilon(\psi_1)\psi_1-\epsilon(\psi_2)\psi_2}{\epsilon(\psi_1)-\epsilon(\psi_2)}.
\end{equation} 
\end{remark}

Given two $n$-qubit state vectors $\psi_1,\psi_2$ with $\epsilon(\psi_1)>0$ and $\epsilon(\psi_2)<0$ as above, we can also construct an exact weighted  4-design from two Clifford orbits. Note that 
\begin{equation*}
	  \frac{|\epsilon(\psi_2)|}{|\orb(\psi_1)|}\sum_{\phi\in \orb(\psi_1)} \big(|\phi\rangle \langle\phi|\big)^{\otimes 4}+ \frac{|\epsilon(\psi_1)|}{|\orb(\psi_2)|}\sum_{\phi\in \orb(\psi_2)} \big(|\phi\rangle \langle\phi|\big)^{\otimes 4}
= \frac{P_{[4]}}{D_{[4]}}
 \end{equation*}
according to \crref{cor:4design} and the definition of $\epsilon(\psi)$ [c.f.~\eref{eq:ONdeviation}]. Therefore, the union of $\orb(\psi_1)$ and  $\orb(\psi_2)$ forms an exact weighted 4-design provided that the  vectors in $\orb(\psi_1)$  and that in $\orb(\psi_2)$ have the following weights respectively,
\begin{equation}
\frac{|\epsilon(\psi_2)|}{|\orb(\psi_1)|[|\epsilon(\psi_1)|+|\epsilon(\psi_2)|]},\qquad \frac{|\epsilon(\psi_1)|}{|\orb(\psi_2)|[|\epsilon(\psi_1)|+|\epsilon(\psi_2)|]}.
\end{equation}
Similar construction also applies to more than two Clifford orbits.

\subsection{Approximate fiducial vectors of 4-designs from  MUB cycler}

In this section we reveal an interesting connection between approximate  4-designs and eigenstates of certain special unitary transformations in the Clifford group. While these states and unitary transformations have been found useful in a number of contexts, the connection with 4-designs seems to be unexplored. We hope our preliminary observation will stimulate further progress.

Let $\{\psi_j^r\}_{r,j}$ be a set of MUB \cite{DurtEBZ10}, where $r$ labels the basis, and $j$ labels each element in a basis. A \emph{balanced state} $\psi$ with respect to $\{\psi_j^r\}_{r,j}$ is a state that looks the same from every basis in the set, that is, the set of probabilities $\{|\langle \psi_j^r| \psi\rangle|^2\}_j$ is independent of $r$ \cite{AmbuSSW14, ApplBD15}. If there exists a unitary operator that cycles through all the bases, then any eigenstate of the unitary operator is a balanced state. For example, the complete set of MUB constructed by Wootters and Fields \cite{WootF89} has a cycler when the dimension is a power of 2, that is, $d=2^n$. Here each MUB cycler is a special unitary transformation in the Clifford group, which
 is also known as a Singer unitary \cite{Zhu15Sh}. The group generated by a Singer unitary is called a Singer unitary group. All Singer unitary groups are conjugated to each other in the Clifford group; in particular all of them have the same order of $d+1$ (modular phase factors). In addition, each Singer unitary has a nondegenerate spectrum, so the eigenbasis is well-defined. In the case of a qubit, each Singer unitary has order 3, and each eigenstate of a Singer unitary is a SIC fiducial and a magic state.

When $n$ is a power of 2,  a simple construction of Singer unitaries (MUB cyclers) was presented in \rcite{SeyfR11}. Here we are interested in constructing approximate fiducial vectors of 4-designs from the eigenvectors of a Singer unitary.
For $n=1,2, 4, 8$, numerical calculation shows that all eigenvectors $\psi_n$ of a Singer unitary for given $n$ have the same value of the deviation parameter $\epsilon(\psi_n)$ [cf.~\eref{eq:ONdeviation}].  Let  $\psi_\rmT$ be a single qubit magic state vector. Calculation shows that
\begin{equation}
-\epsilon(\psi_n\otimes \psi_\rmT)=
\begin{cases}
\frac{2}{9} &n=1,\\
0.12 &n=2,\\
0.0312&n=4,\\
0.0020 &n=8.
\end{cases}
\end{equation}
The magnitude of the deviation  $\epsilon(\psi_n\otimes \psi_\rmT)$  is around $1/2^{n+1}$, which has the same order of magnitude as the standard deviation of $\epsilon(\psi)$ of a random $(n+1)$-qubit state vector $\psi$; cf. \eref{eq:4designDeviation}.
The orbit generated from $\psi_n\otimes \psi_\rmT$ is a very good approximation to a 4-design.  Exact 4-design fiducial vectors can be constructed using algorithm~1 in the previous section. In addition, $\psi_n$ or $\psi_n\otimes \psi_\rmT$ can serve as an input to Algorithm~2 presented in the previous section.

\begin{conjecture}\label{con:Singer4design}
Suppose  $\psi_n$ is any eigenvector of a Singer unitary operator in the $n$-qubit Clifford group. Then
\begin{equation}\label{eq:Singer4design}
\lim_{n\rightarrow \infty} \epsilon(\psi_n\otimes \psi_\rmT)=0.
\end{equation}
\end{conjecture}
This conjecture implies
 that the orbit generated by the  $(n+1)$-qubit Clifford group from $\psi_n\otimes \psi_\rmT$ converges to a 4-design with respect to the operator norm as $n$ grows. \Eref{eq:Singer4design} has several  equivalent formulations; a succinct  alternative  reads
\begin{equation}
\lim_{n\rightarrow \infty} \|\Xi(\psi_n)\|_{\ell_4}^4=3.
\end{equation}

\subsection{Harmonic invariants, connections to the real-valued theory, and 5-designs}
\label{sec:harmonic}

One original motivation \cite[Section~1.E]{KuenG15} for this work came from a result on the \emph{real} Clifford group $\mathrm{RC}_n$.
This is the group generated by tensor products of the real Pauli matrices $\sigma_{(0,0)}, \sigma_{(0,1)}, \sigma_{(1,0)}$, together with the (real) Hadamard matrix
\begin{equation}\label{eqn:hadamardr}
H_{\RR}=\frac{1}{\sqrt 2}\begin{pmatrix}
1 & 1\\
1&-1
\end{pmatrix}
\end{equation}
and the CNOT matrix as in (\ref{eqn:cnot}) between each pair of qubits.
In Refs.~\cite{Side99,Rung96,NebeRS01,NebeRS06}, the authors studied  invariant polynomials of $\mathrm{RC}_n$ and  real spherical designs \cite{DelsGS77,Bann79} that appear as the group's orbits.

Using methods from classical invariant theory, they showed \cite[Corollary~4.13]{NebeRS01} that there are no invariant harmonic polynomials of $\mathrm{RC}_n$ of degree $2t$ for $t=1,2,3,5$, and -- up to scalar multiples -- a single harmonic invariant for $t=4$ (c.f.~Appendix~\ref{sec:polynomials}).
It follows that the orbit of any vector forms a real spherical design of strength $2\cdot 3 +1$.
Furthermore, the orbit of any real root of the unique harmonic invariant of degree $2\cdot 4$
forms a spherical design of strength $2\cdot 5 + 1$.
The existence of real roots follows from an averaging argument similar to the one we employ in \sref{sec:randomOrbits}.

References~\cite{NebeRS01,NebeRS06} also treat the complex Clifford group $\Cli_n$.
However, it seems that these works only characterize the invariant polynomials in $\Hom_{(2t)}(\CC^d)$ rather than the ones in $\Hom_{(t,t)}(\CC^d)$ investigated here (c.f.~Appendix~\ref{sec:polynomials}).
To the present authors, these two cases seem significantly different and we are not aware of any way that would allow one to directly apply the \emph{complex} results from \rscite{NebeRS01,NebeRS06} in our setting (however, see below for corollaries of their \emph{real}-valued statements).

It was therefore an initial goal of this work to see whether methods from quantum information theory (such as stabilizer codes) could be used to find similar statements to the ones summarized above.
The theory developed in the previous sections largely achieves this goal.
The following proposition reformulates our results in a way that emphasizes the similarities.

\begin{proposition}\label{prop:harmoniclanguage}
	The Clifford group $\Cli_n$ has no non-trivial harmonic invariants
	of degrees $(1,1), (2,2)$, or $(3,3)$.
	All harmonic invariants of degree $(4,4)$ are multiples of $\epsilon$ as defined in \eref{eq:ONdeviation}.
	The orbit of a normalized vector $\psi$ forms a $4$-design if and only
	if it is a root of $\epsilon$.
\end{proposition}

\begin{proof}
	Because the Clifford group forms a unitary $3$-design, it follows that for $t=1,2,3$, the
	commutant
	$L(\Sym_t(\CC^d))^{\Cli_n}$
	of $\Cli_n$ acting on
	$\Sym_t(\CC^d)$
	is given by  multiples of $P_{[t]}$.
	By Eqs.~(\ref{eqn:complexharmonicdecomp}) and  (\ref{eqn:complexSymmetricEmbedding}) in \aref{sec:polynomials}, these are just the embeddings of $H_{(0,0)}$ into $L(\Sym_t(\CC^d))$ (corresponding to the polynomials $\psi \mapsto \|\psi\|_2^{2t}$).
	This proves the first part.
	
	From this and (\ref{eqn:complexharmonicdecomp}), we have that
	\begin{equation*}
	L(\Sym_4(\CC^d))^{\Cli_n}
	\simeq
	H_{(0,0)}^{\Cli_n} \oplus H_{(4,4)}^{\Cli_n}
	=
	H_{(0,0)} \oplus H_{(4,4)}^{\Cli_n}.
	\end{equation*}
	At the same time, Theorem~\ref{thm:Main} implies that $L(\Sym_4(\CC^d))^{\Cli_n}$ is spanned by the the two projectors $P_\pm$ onto $W_{(4)}^\pm$ defined in \sref{sec:randomOrbits}.
	As in the first part, $P_{[4]}=P_++P_-$ spans $H_{(0,0)}$. Clearly, the operator
	\begin{equation*}
	A:= \frac{D_{[4]}}{D_+} P_+ - P_{[4]}
	\end{equation*}
	is an element of the commutant and orthogonal to $P_{[4]}$.
	As such, $A$ must span $H_{(4,4)}^{\Cli_n}$.
	But $\epsilon$ is the polynomial $p_A$ associated with $A$ in the sense of Lemma~\ref{lem:complexPolarization}.
	
	The final statement of \pref{prop:harmoniclanguage} is just \eref{eqn:4designconds}.
\end{proof}

The results on the real Clifford group mentioned above strongly suggest upper bounds on the dimensions of the spaces of harmonic invariants of $\Cli_n$.
Indeed, up to slightly different phase conventions for the Hadamard gate [Eq.~(\ref{eqn:hadamardc}) vs Eq.~(\ref{eqn:hadamardr})], which are immaterial for the present discussion, the real Clifford group $\mathrm{RC}_n$ is a subgroup of the complex one $\Cli_n$.
Now let $p_A\in\left(\Harm_{(i,i)}\right)^{\Cli_n}$ be a $\Cli_n$-invariant polynomial.
It is clearly also invariant under any subgroup of $\Cli_n$, in particular, under $\mathrm{RC}_n$.
Let $A=A_{\Re} + i A_{\Im}$, for $A_{\Re}, A_{\Im}$ real matrices be the decomposition of $A$ into its real and imaginary part.
Since the action of $\mathrm{RC}_n$ does not mix the real and the imaginary components, it follows that the restrictions of $p_{A_{\Re}}$ and $p_{A_{\Im}}$ to real arguments are $\mathrm{RC}_n$-invariant polynomials.
Using (\ref{eqn:totallysymmetric}), they can easily be checked to lie in the harmonic space $H_{2i}(\RR^d)$.
(The the restriction of  $p_{A A^\dagger}$ to real arguments also gives a real polynomial -- but it need not be harmonic, even if $p_A$ was.)
We can therefore convert invariant harmonic polynomials of $\Cli_n$ into those of $\mathrm{RC}_n$.
Unfortunately, the resulting real polynomials may turn out to be zero:
In the language of \rcite{DoheW13} (Appendix~\ref{sec:polynomials}), it could happen that the matrices $A_{\Im}, A_{\Re}$ -- while elements of $L(\Sym_{i}(\RR^d))$ -- are orthogonal to the \emph{totally symmetric matrices} $\mathrm{MSym}_i(\RR^d))$.
This technical problem prevents us from directly inferring the absence of harmonic invariants of $\Cli_n$ of bi-degree $(t,t)$ for $t=1,2,3,5$ from the absence of real harmonic invariants of $\mathrm{RC}_n$ of degree $2 \cdot t$ for the same $t$'s.

We conjecture, however, that this potential problem is not realized for the Clifford group, at least not for degree $(5,5)$.
In this case, \cite[Corollary~4.13]{NebeRS01} -- stating the absence of harmonic invariants of $\mathrm{RC}_n$ of degree $2\cdot 5$ --  would imply the following:

\begin{conjecture}\label{conj:5}
	Let $\psi$ be a normalized vector.
	Its orbit forms a complex projective $5$-design iff it is a root of $\epsilon$.
\end{conjecture}

There are several pieces of evidence in favor of this conjecture:
\begin{enumerate}
	\item
	\Cref{conj:5} holds when $n=1$ according to  the discussion in \sref{sec:qubit}. It is also supported  by numerical calculation when $n=2$
	
	\item
	The argument works for $t=1,2,3$.
	\item
	The set of matrices in $L(\Sym_i(\RR^d))$ that is orthogonal to $\operatorname{MSym}_{i}(\RR^d)$ is of measure zero.
\end{enumerate}
It would be interesting to verify this conjecture, as well as to re-prove the statement of Ref.~\cite{NebeRS01} using just the tools of the present paper.

Even if a simple way of turning general harmonic invariants of the complex Clifford group into those for the real Clifford group could be constructed, the results of the present work and those of Refs.~\cite{Side99,Rung96,NebeRS01,NebeRS06} would still differ in scope.
On the one hand, our results are stronger, as they allow for a decomposition of the entire space $\left(\CC^d\right)^{\otimes 4}$ under $\Cli_n$, as opposed to just the totally symmetric subspace.
On the other hand, the cited references are stronger by giving a characterization of the invariant polynomials of \emph{any} degree (in terms of weight enumerator polynomials of certain binary codes), while we restrict attention to degree~4.

\section{Summary}

The most prominent unitary $t$-design considered in quantum
information is the multi-qubit Clifford group, which is a unitary 3-design, but, unfortunately,  not a 4-design. Accordingly, Clifford orbits are 3-designs, but generally not 4-designs. The lack of an explicit family of well structured 4-designs has been a major limitation in the applications of $t$-designs for derandomizing constructions that rely on random vectors.

In this work we showed  that although Clifford orbits do not constitute $4$-designs, their $4$th moments are well-behaved such  that
for several major applications,  including phase retrieval and quantum state discrimination,
typical Clifford orbits turn out to perform  as well
as $4$-designs or Gaussian random vectors would. Moreover, we gave various constructions of exact 4-designs and approximations  of arbitrarily high precision to serve for more demanding applications.
In order to achieve this goal, we determined all irreducible components that appear in
the $4$th tensor power of the Clifford group. It turns out that the structure of these representations is completely captured by Schur-Weyl duality and a special stabilizer code. In addition to the applications mentioned above, our results  may  help construct exact unitary 4-designs or better approximations. In the course of our study, we also discovered several  results concerning the representations of the discrete symplectic group, which may be of interest to pure mathematician.

Our work also leaves several open problems, which deserve further study.
\begin{enumerate}
	\item Is there any orbit of the Clifford group that forms a $t$-design for $t>4$ (c.f.~\cref{conj:5})? The answer is positive when $n=1$. It seems that the same could hold for larger $n$.
	
	\item What is the maximal  $t$ such that there is an orbit of the Clifford group that forms a $t$-design. The answer is 7 when $n=1$. How about approximate $t$-designs?
	
	\item Prove  \cref{con:Singer4design}.
	
	\item Construct unitary 4-designs based on the Clifford group.
\end{enumerate}

More generally, it would be desirable to give an explicit description of the commutant of higher tensor powers of the Clifford group -- maybe similar to the characterization of invariant polynomials of the Clifford group in terms of weight enumerator polynomials described in Refs.~\cite{Rung96,NebeRS01,NebeRS06}.
There is a potentially simpler problem.
Central to our construction was the stabilizer projector $P_{n,k}$.
It belongs to a stabilizer code in $(\CC^d)^{\otimes 4}$ that is not a tensor product itself.
Similarly, the recent work \rcite{NezaW16} identifies an element of the commutant of a tensor power of the Clifford group, that is itself a non-factoring Clifford operation on the tensor product space.
If a general explicit description of the commutant of powers of the Clifford group might not be realistically available, one could ask how far one can go by classifying those commuting elements that can themselves be expressed in terms of Clifford theory or related constructions. 

\section*{Acknowledgments}

This work has been supported by the Excellence Initiative of the German Federal and State Governments (Grant ZUK 81), the ARO under contract W911NF-14-1-0098 (Quantum Characterization, Verification, and Validation), and the DFG (SPP1798 CoSIP). 
Major parts of this project were undertaken while DG and RK participated in the \emph{Mathematics of Signal Processing} program of the Hausdorff Research Institute of Mathematics at the University of Bonn.

\appendix

\section{Alternative proof of  \lref{lem:dimensions}}
In this appendix,  we present an alternative approach for computing the dimensions of $W_\lambda^+$ defined in \sref{sec:Main}, thereby yielding an alternative proof  of  \lref{lem:dimensions}. In the course of study, we also  construct explicit orthonormal bases for $W_{[4]}^+$ and $W_{[1^4]}^+$. 

To achieve our goal, we first construct   an orthonormal basis for $V_{n,4}$ and determine the orbits of basis elements under the action of the symmetric group $S_4$. When $n=1$, one orthonormal basis of $V_{n,4}$ is composed of the following four states,
\begin{equation}
\begin{aligned}
|\phi_0\rangle &:= \frac{1}{\sqrt{2}}(|0000\rangle + |1111\rangle), \\
|\phi_1\rangle &:= \frac{1}{\sqrt{2}}(|1001\rangle + |0110\rangle), \\
|\phi_2\rangle &:= \frac{1}{\sqrt{2}}(|0101\rangle + |1010\rangle), \\
|\phi_3\rangle &:= \frac{1}{\sqrt{2}}(|0011\rangle + |1100\rangle).
\end{aligned}
\end{equation} 
The symmetric group $S_4$ (permuting the four tensor factors) fixes $|\phi_0\rangle$ and acts like
$S_3$ on  $|\phi_1\rangle, |\phi_1\rangle, |\phi_2\rangle$. Since
$V_{n,4}=V_{1,4}^{\otimes n}$ for general $n$, one orthonormal basis  of $V_{n,4}$ is composed of the following $4^n$ states,
\begin{equation}
|\phi_{i_1i_2,\ldots,i_n}\rangle=|\phi_{i_1}\rangle\otimes \cdots \otimes |\phi_{i_n}\rangle, \quad i_1, i_2, \ldots, i_n\in  \{0,1,3,4\}.
\end{equation}
Here each state in the basis is labeled by a length-$n$ string $i_1, \dots,
i_n$  with $i_j \in \{0,1,3,4\}$. Each permutation in the symmetric group $S_4$ induces  a permutation on the basis states and a corresponding permutation on the strings, which acts on all letters simultaneously. The orbits on the strings divide into three types as described as follows.
\begin{enumerate}
	\item
	One orbit containing $0^{\times n}$, referred to as type I  orbit below.
	\item
	Any string in $\{0,i\}^{\times n}$ (for given $i\in\{1,2,3\}$) excluding $0^{\times n}$
	generates an orbit of length $3$. There are $2^n-1$ such
	orbits of length 3,  referred to as type II  orbits below.
	\item	The
	remaining strings have either  two or three
	distinct non-zero letters and  are partitioned into orbits of length $6$,  referred to as type III orbits below. The number of such orbits is 
	\begin{equation}
	\frac{4^n-3\times 2^n+2}{6}=
	\frac{(2^n-2)(2^n-1)}{6}=\frac{(d-2)(d-1)}{6}.
	\end{equation}	
\end{enumerate}
The  total number of orbits is 
\begin{equation}
2^n + \frac{4^n-3\times 2^n+2}{6}=
\frac{4^n+3\times 2^n+2}{6}=
\frac{(2^n+2)(2^n+1)}{6}=\frac{(d+2)(d+1)}{6}.
\end{equation}
The strings corresponding to the three types of orbits are  referred to as type I, II, III strings, respectively.  The stabilizer  of a type I string is $S_4$, that  of a type II string is  a Sylow-2 subgroup of $S_4$, and that of a type III string is  the unique  order-4 normal subgroup of $S_4$.

Now we are ready to compute the dimensions of  $W_{\lambda}^+$.  Let $\orb(s)$ denote the orbit of the string $s\in \{0,1,2,3\}^n$ under the action of $S_4$. According to \eref{eq:Plambda}, 
\begin{equation}\label{eq:Plambdaphi}
P_\lambda|\phi_s\rangle=\frac{d_\lambda}{24}\sum_{\sigma\in S_4} \chi_\lambda(\sigma) U_\sigma|\phi_s\rangle=\frac{d_\lambda}{24}\sum_{r\in \orb(s)}\left( \sum_{\sigma\in S_4| \sigma(s)=r}\chi_\lambda(\sigma)\right) |\phi_r\rangle,
\end{equation}
where $ \sum_{\sigma\in S_4| \sigma(s)=r}\chi_\lambda(\sigma)$ is the sum of $\chi_\lambda(\sigma)$ over a coset of the stabilizer of $s$. For example,
\begin{equation}\label{eq:ProjExample}
P_\lambda|\phi_{0\cdots0}\rangle=\begin{cases}
|\phi_{0\cdots0}\rangle & \lambda=[4],\\
0&\mbox{otherwise}.
\end{cases}
\end{equation}

When $\lambda=[4]$, $d_\lambda=1$ and $\chi_\lambda(\sigma)=1$ for all $\sigma\in 
S_4$. Consequently,
\begin{equation}
P_{[4]}|\phi_s\rangle=\frac{1}{|\orb(s)|}\sum_{r\in \orb(s)}|\phi_r\rangle.
\end{equation}
Note that $P_{[4]}|\phi_s\rangle\in W_{[4]}^+$ only depends on  $\orb(s)$ and that the states corresponding to different orbits are orthogonal. Let $\mathscr{S}$ be a subset of $\{0,1,3,4\}^n$ that contains exactly one string from each orbit. Then the set 
\begin{equation}\label{eq:SymBasis}
\{\sqrt{|\orb(s)|}P_{[4]}|\phi_s\rangle \,|\,  s\in \mathscr{S} \}
\end{equation}
forms an orthonormal basis for $W_{[4]}^+$. In particular, the dimension of $W_{[4]}^+$ is equal to the total number of orbits of strings, that is,
\begin{equation}
D_{[4]}^+=\dim(W_{[4]}^+)=\frac{(d+2)(d+1)}{6}.
\end{equation}

Now consider the subspace $W_{[1^4]}^+$. Note that $P_{[1^4]}|\phi_s\rangle=0$ when $s$ is an type I or type II string. An orthonormal basis for $W_{[1^4]}^+$ is given by
\begin{equation}
\{\sqrt{|\orb(s)|}P_{[1^4]}|\phi_s\rangle\,|\,  \mbox{$s\in \mathscr{S}$ is of type III}  \}.
\end{equation}
The dimension of $W_{[1^4]}^+$ is equal to the number of type III orbits, that is,
\begin{equation}
D_{[1^4]}^+=\dim(W_{[1^4]}^+)=\frac{(d-2)(d-1)}{6}.
\end{equation}

It is more involved to compute the dimension of  $W_{[2,2]}^+$. Fortunately,  this task can be avoided if we can compute the dimensions of $W_{[2,1,1]}^+$ and $W_{[3,1]}^+$. It turns out that $P_{[2,1,1]}|\phi_s\rangle=0$ and $P_{[3,1]}|\phi_s\rangle=0$ for all strings $s\in \{0,1,2,3\}^n$. This conclusion follows from \eref{eq:ProjExample} when $s$ is a type I string, that is $s=0\cdots0$. When $s$ is a type II or III string, this conclusion follows from \eref{eq:Plambdaphi} and \lsref{lem:CosetCharSylow} , \ref{lem:CosetChar} below, recall that the stabilizer of a type II string is  a Sylow-2 subgroup of $S_4$ and that of a type III string is  the unique  order-4 normal subgroup of $S_4$. Consequently, both $W_{[2,1,1]}^+$ and $W_{[3,1]}^+$ have dimension 0, so  that 
\begin{equation}
D_{[4]}^++D_{[1^4]}^++2D_{[2,2]}^+=d^2,
\end{equation}
which implies that $D_{[2,2]}^+=(d^2-1)/3$.

\begin{lemma}\label{lem:CosetCharSylow}
	Let $G$ be a Sylow 2-subgroup of $S_4$. Then $\sum_{\sigma\in g G}\chi_\lambda(\sigma)=\sum_{\sigma\in  Gg}\chi_\lambda(\sigma)=0$ for $\lambda=[2,1,1], [3,1]$ and all $g\in S_4$.
\end{lemma}
\begin{lemma}\label{lem:CosetChar}
	Let $H$ be the unique order-4 normal subgroup of $S_4$. Then $\sum_{\sigma\in g H}\chi_\lambda(\sigma)=\sum_{\sigma\in  Hg}\chi_\lambda(\sigma)=0$ for $\lambda=[2,1,1], [3,1]$ and all $g\in S_4$.
\end{lemma}
\begin{remark}
	\Lref{lem:CosetCharSylow} follows from \lref{lem:CosetChar} since each coset of a Sylow 2-subgroup of $S_4$ is a union of two cosets of  the unique order-4 normal subgroup of $S_4$.
	The two lemmas can be verified directly based on \tref{tab:CharS4}. Nevertheless, we shall present more instructive proofs below. 
\end{remark}

\begin{proof}[Proof of \lref{lem:CosetCharSylow}]
	Suppose $\lambda=[2,1,1]$ or $\lambda=[3,1]$. 
	Note that $G$ is isomorphic to the order-8 dihedral group; it 
	has one element of cycle type $(1^4)$, three elements of cycle type $(2^2)$, two elements of cycle type $(2,1^2)$, and two elements of cycle type $(4)$. Therefore, $\sum_{\sigma\in  G}\chi_\lambda(\sigma)=0$  according to \tref{tab:CharS4}. Let $g$ be any order-3 element in $S_4$;  then $G, gG, g^{-1}G$ are three distinct left  cosets of $G$. Since  $\sum_{\sigma\in  G}\chi_\lambda(\sigma)=0$ and  $\sum_{\sigma\in  S_4}\chi_\lambda(\sigma)=0$, it follows that 
	\begin{equation}\label{eq:CharS}
	\sum_{\sigma\in  gG}\chi_\lambda(\sigma)+\sum_{\sigma\in  g^{-1}G}\chi_\lambda(\sigma)=0.
	\end{equation}
	On the other hand,   by conjugation $G$ acts transitively on the eight order-3 elements in $S_4$, so there exists an element $h$ in $G$ such that $hgh^{-1}=g^{-1}$, that is, $hg G h^{-1}=g^{-1}G$. It follows that 
	\begin{equation}\label{eq:CharE}
	\sum_{\sigma\in  gG}\chi_\lambda(\sigma)=\sum_{\sigma\in  g^{-1}G}\chi_\lambda(\sigma),
	\end{equation}
	which, together with \eref{eq:CharS}, implies that 
	\begin{equation}
	\sum_{\sigma\in  gG}\chi_\lambda(\sigma)=\sum_{\sigma\in  g^{-1}G}\chi_\lambda(\sigma)=0. 
	\end{equation}
	In conclusion, $\sum_{\sigma\in g G}\chi_\lambda(\sigma)=0$ for $\lambda=[2,1,1], [3,1]$ and all $g\in S_4$. The equality $\sum_{\sigma\in  Gg}\chi_\lambda(\sigma)=0$ follows from the same reasoning.
\end{proof}

\begin{proof}[Proof of \lref{lem:CosetChar}]
	Suppose $\lambda=[2,1,1]$ or $\lambda=[3,1]$. 
	Note that $H$ has one element of cycle type $(1^4)$ and three elements of cycle type $(2^2)$.  Therefore, $\sum_{\sigma\in  H}\chi_\lambda(\sigma)=0$  according to \tref{tab:CharS4}. The symmetric group $S_4$ has three Sylow 2-subgroups, each of which is the union of two cosets of $H$. Let $G$ be any Sylow 2-subgroup of $S_4$, then $\sum_{\sigma\in  G}\chi_\lambda(\sigma)=0$ according to \lref{lem:CosetCharSylow}, which implies that $\sum_{\sigma\in  G\setminus H}\chi_\lambda(\sigma)=0$.

	Let $g_j$ for $j=1, 2, 3, 4, 5, 6$ be the coset representatives of $H$, with $g_1$ being the identity. Above analysis shows that $\sum_{\sigma\in  g_j H}\chi_\lambda(\sigma)=0$ for four of the six cosets, say, $j=1,2,3,4$, so that  
	\begin{equation}\label{eq:CharS2}
	\sum_{\sigma\in  g_5H}\chi_\lambda(\sigma)+\sum_{\sigma\in  g_6H}\chi_\lambda(\sigma)=0.
	\end{equation}
	In addition, the two coset representatives $g_5, g_6$ necessarily have order 3 since otherwise they would belong to certain Sylow 2-subgroups of $S_4$. Observing that $H$ is normal in $S_4$ and that all order-3 elements in $S_4$ are conjugated to each other, we conclude that 
	\begin{equation}\label{eq:CharE2}
	\sum_{\sigma\in  g_5H}\chi_\lambda(\sigma)=\sum_{\sigma\in  g_6H}\chi_\lambda(\sigma),
	\end{equation}
	which, together with \eref{eq:CharS2},  implies that 
	\begin{equation}
	\sum_{\sigma\in  g_5H}\chi_\lambda(\sigma)=\sum_{\sigma\in  g_6H}\chi_\lambda(\sigma)=0.
	\end{equation}
	In conclusion, $\sum_{\sigma\in g H}\chi_\lambda(\sigma)=0$ for $\lambda=[2,1,1], [3,1]$ and all $g\in S_4$. As an immediate consequence, $\sum_{\sigma\in  Hg}\chi_\lambda(\sigma)=0$ since  left cosets and right cosets of $H$ coincide.
\end{proof}

\section{Two natural sets of vectors in the stabilizer code $V_{n,4}$}

\subsection{An interesting basis for the stabilizer code}

Recall the basis-dependent \emph{vectorization map} which sends
matrices to tensors
\begin{align*}
\operatorname{vec}: L(\CC^d) &\to \CC^d\otimes \CC^d \\
\sum_{i,j} L_{i,j} |e_i\rangle \langle e_j| 
&\mapsto
\sum_{i,j} L_{i,j} |e_i\rangle \otimes  | e_j\rangle. 
\end{align*}
It fulfils 
\begin{equation*}
\operatorname{vec}(A B C) = A \otimes C^\rmT \operatorname{vec}(B),
\end{equation*}
where the transpose is to be taken with respect to the same basis
in which the vectorization map is defined.

With this notion, note that
\begin{equation*}
\mathcal{B} := \{ \operatorname{vec}(W_a) \otimes \operatorname{vec}(W_a) \,|\, a \in \FF_2^{2n} \}
\end{equation*}
defines a set of $d^2$ orthogonal vectors in $(\CC^d)^{\otimes 4}$. 
One easily verifies that $\mathcal{B}$ is contained in the stabilizer
code $V_{n,4}$:
\begin{align*}
\Ten^4(W_b)\,\big(\vectorize(W_a) \otimes \vectorize(W_a)\big)
&=
\vectorize(W_b W_a W_b^\rmT) 
\otimes
\vectorize(W_b W_a W_b^\rmT)  \\
&=
(\pm \vectorize(W_a)) \otimes
(\pm \vectorize(W_a))  \\
&= \vectorize(W_a)\otimes \vectorize(W_a).
\end{align*}
It thus forms an orthogonal basis of the code.
A similar calculation shows that the real elements of the Clifford
group act on this basis by permutation
\begin{equation*}
\Ten^4(U)\,\big(\vectorize(W_a) \otimes \vectorize(W_a)\big)
=
\vectorize(W_{Fa}) \otimes \vectorize(W_{Fa}),
\end{equation*}
where $F\in \Sp(2n,\FF_2)$ is the symplectic map associated with
$U/\mathcal{P}_n$. Complex elements of the Clifford group $\Cli_n$ still act by signed permutation on the basis -- i.e.\ they permute the elements and
may multiply them with signs $\pm 1$. The latter fact can be  verified explicitly by inspecting the action of  those generators of the Clifford group as discussed in \sref{sec:PauliCli}, all of which are real except for the phase gate. 
In particular, all Clifford unitaries act \emph{monomially}.\footnote{
	Certain monomial representations of the Clifford group have been studied
	before in Ref.~\cite{ApplBBG12}. However, their results are
	incomparable to our findings, as they classified monomial
	representations that also contain a faithful representation of
	the Pauli group $\mathcal{P}_n$.
} The above discussion also implies that the representation of the  symplectic group $\Sp(2n,\FF_2)$ afforded by the stabilizer code $V_{n,4}$ is a signed permutation representation in the above basis.

We have not used this basis affording a monomial representation of
the Clifford group and the symplectic group in the present paper.
However, we speculate that one might use it to give a more explicit
derivation of the characters described in \sref{sec:repsofsp}.

\subsection{An interesting orbit}

Here, we describe a Clifford orbit of vectors in $W^+_{[4]}$ that is
naturally labeled by isotropic subspaces $M\subset \FF_2^{2n}$ of
dimension $\dim M = n$ (such spaces are called \emph{maximally
	isotropic}). 
The authors are not aware of any application of this particular
configuration of vectors.
However, we feel that the construction is sufficiently canonic to
deserve a mention.

Choose a maximally isotropic space $M\subset \FF_2^{2n}$ and define
\begin{align*}
S_n'(M) &:=
\{ W_a \otimes W_a \otimes \Id\otimes \Id \,|\, a \in M
\}, \\
S_n''(M) &:=
\{ 
W_a \otimes \Id \otimes W_a \otimes \Id
\,|\, a \in M
\}.
\end{align*}
Then the union $S_{n,4} \cup S_n' \cup S_n''$ generates a maximal stabilizer group
on $4n$ qubits and thus determines a stabilizer state $|\psi_M\rangle \in V_{n,4}$, where $S_{n,4}$ is defined in \eref{eqn:stabgroup}.

Consider the concrete example 
\begin{equation*}
M_Z=\{ (p_1, 0, p_2, 0, \dots, p_n, 0) \,|\, p_1, p_2, \ldots, p_n \in \FF_2 \}.
\end{equation*}
Then $W_a$ for $a\in M_Z$ is an element of $\{\sigma_{(0,0)}, \sigma_{(1,0)}
\}^{\otimes n}$. In this particular case, one verifies that
\begin{equation*}
\ket{\psi_{M_Z}} 
= 
\frac1{2^{n/2} }
\sum_{x\in \FF_2^n} \ket{x}^{\otimes 4} 
= 
\frac1{2^{n/2} }
|\phi_0\rangle^{\otimes n} 
\in H_{[4]}^+	
\simeq
W_{[4]}^+,
\end{equation*}
where $|\phi_0\rangle = (|0000\rangle + |1111\rangle)/\sqrt{2}$.

Now consider a Clifford unitary $U\in C_n$, associated with the
symplectic transformation $F$. Then, for any maximally isotropic subspace $M$ of $\bbF^{2n}$, 
\begin{align*}
\Ten^4(U) S'_n(M) \Ten^4(U)^\dagger
&=
S'_n(F M)
\end{align*}
and the same is true for $S''_n(M)$. We conclude that
\begin{equation*}
\Ten^4(U) \ket{\psi_M} \propto \ket{\psi_{F M}}.
\end{equation*}
Because the symplectic group acts transitively on maximal isotropic
subspaces, the Clifford group acts monomially (i.e.\ by permutation
and possibly multiplication with a phases) on the set 
\begin{equation*}
X := \{ \ket{\psi_M} \,|\, M \text{ max.\ isotropic} \}.
\end{equation*}
As $X$ is -- up to phases -- an orbit of the Clifford group generated
from $|\psi_{M_Z}\rangle \in W_{[4]}^+$, the entire set is contained in
$W_{[4]}^+$.

Because the projectors $|\psi_M\rangle\langle \psi_M|$ form an orbit
under the action by conjugation of the irreducible representation of
$\rmC_n$ on $W_{[4]}^+$, it follows from Schur's Lemma that $X$ is a
\emph{tight frame} on that space.
It is known (e.g.\ \cite{Gros06}) that
\begin{equation*}
|X| 
= |\{ \text{max.\ isotropic subspaces of } \FF_2^{2n} \}| 
= \prod_{i=1}^n (2^i + 1),
\end{equation*}
which implies that $|X|>D_{[4]}^+$. If $n\geq3$, then 
\begin{equation*}
|X| = 3D_{[4]}^+ \, \prod_{i=1}^{n-2} (2^i+1).
\end{equation*}
So $X$ is overcomplete as a frame.
We can compute the squared inner products between elements of $X$ from
the intersection of their stabilizer groups (c.f.\ e.g.\
\rcite{KuenG15}):
\begin{equation*}
|\braket{\psi_M}{\psi_{N}}|^2
=
\frac1{2^{4n}} 
2^{2n+2\dim(M\cap N)}
=
\left(
2^{\dim(M\cap N)-n}
\right)^2.
\end{equation*}
That number is  the square of what one would obtain for the
overlap-squared between $n$ qubit stabilizer states taken from bases
associated with, respectively, $M$ and $N$ \cite{KuenG15}.

\section{\label{sec:ProjBound}Proof of a generalization of  \eref{eq:Xibound}}
Here we prove a generalization of  \eref{eq:Xibound}, which also provides some insight on the entanglement property of the stabilizer code $V_{n,4}$.

\begin{lemma}\label{lem:ProjBound}
	Suppose $\psi_j$ for $j=1,2,3,4$ are four normalized  state vectors in dimension $d=2^n$. Then 
	\begin{equation}\label{eq:ProjBound}
	0\leq\tr\left(P_{n,4}\bigotimes_{j=1}^4 |\psi_j\rangle\langle\psi_j| \right)\leq \frac{1}{d}.
	\end{equation}
	If the upper bound is saturated, then $\psi_j$ for $j=1,2,3,4$ are stabilizer states that belong to a same stabilizer basis.
\end{lemma}
The upper bound in \eref{eq:ProjBound} means that the stabilizer code $V_{n,4}$ contains no product state. Moreover, it 
sets an upper bound $1/d$ for the fidelity between any pure state in $V_{n,4}$ and any product state, that is, a lower bound for the geometric measure of entanglement
of any  pure state in $V_{n,4}$. 
The upper bound
can be saturated if $\psi_j$ for $j=1,2,3,4$ are identical stabilizer states, but this is not necessary. For example, it can also be saturated if $\psi_1=\psi_2$ and $\psi_3=\psi_4$ are two orthogonal stabilizer states that belong to a same stabilizer basis. By contrast, the lower bound in \eref{eq:ProjBound} can be saturated if $\psi_1=\psi_2=\psi_3$ and $\psi_4$ are two orthogonal stabilizer states that belong to a same stabilizer basis.

\begin{proof}The lower bound is trivial since both $P_{n,4}$ and $\bigotimes_{j=1}^4 |\psi_j\rangle\langle\psi_j|$ are positive semidefinite. The upper bound can be derived as follows. 
	\begin{align}\label{eq:ProjBoundProof}
	&d^2\tr\left(P_{n,4}\bigotimes_{j=1}^4 |\psi_j\rangle\langle\psi_j| \right)=\sum_a \prod_{j=1}^4 \langle \psi_j| W_a|\psi_j\rangle
	\leq  \left[\prod_{j=1}^4 \sum_a(\langle \psi_j| W_a|\psi_j\rangle)^4\right]^{1/4}\nonumber\\
	&\leq \left[\prod_{j=1}^4 \sum_a(\langle \psi_j| W_a|\psi_j\rangle)^2\right]^{1/4}=d.
	\end{align}
	Here 
	the first inequality follows from repeated applications of the Cauchy inequality or the H\"older inequality. It 
	is saturated iff 
	\begin{equation}\label{eq:ProjBoundProof2}
	|\langle \psi_1| W_a|\psi_1\rangle| =|\langle \psi_2| W_a|\psi_2\rangle|=|\langle \psi_3| W_a|\psi_3\rangle|=|\langle \psi_4|W_a|\psi_4\rangle|,\qquad \prod_{j=1}^4 \langle \psi_j| W_a|\psi_j\rangle \geq0 \quad\forall a.
	\end{equation}
	The second inequality in \eref{eq:ProjBoundProof} is saturated iff each $\langle\psi_j| W_a|\psi_j\rangle$  takes on only one of the three values $0,\pm1$. In that case, each $\Xi(\psi_j)$ (recall that  $\Xi_a(\psi_j)=\langle \psi_j| W_a|\psi_j\rangle$) has exactly $d$ entries equal to $\pm1$, given that $\sum_a(\langle\psi_j| W_a|\psi_j\rangle)^2=d$ for $j=1,2,3,4$.
	Note that the set $\{a\in \bbF_2^{2n}\,|\,\langle\psi_j| W_a|\psi_j\rangle=\pm1\}$ for given $j$ must form a maximal isotropic subspace of the symplectic vector space $\bbF_2^{2n}$, so each $\psi_j$ is an eigenvector of a stabilizer group and is thus a stabilizer state by definition. If the two inequalities in \eref{eq:ProjBoundProof} are saturated simultaneously, then $\psi_j$ for $j=1,2,3,4$ must  be eigenvectors of a common stabilizer group due to \eref{eq:ProjBoundProof2}. In other words, they belong to a same stabilizer basis. 	
\end{proof}

\section{\label{sec:SecondMoment}Derivation of \eref{eq:SecondMoment}}

In this appendix, we derive the second moment of $\alpha_+(\psi)$, as presented in \eref{eq:SecondMoment}.
\begin{align}\label{aeq:SecondMoment}
\rmE[\alpha_+(\psi)^2]&=\frac{1}{D_{[8]}}\tr(P_{n,4}^{\otimes2} P_{[8]})=\frac{1}{d^4 D_{[8]}}\sum_{a,b} \tr\left[ P_{[8]} (W_a^{\otimes4}\otimes  W_b^{\otimes4})\right]=\frac{16(d^2+15d+68)}{d^2(d+3)(d+5)(d+6)(d+7)},
\end{align}
where $P_{n,4}$ is the  projector onto the stabilizer code $V_{n,4}$ discussed in \sref{sec:specialcode},  $P_{[k]}$ is the projector onto the $k$-partite symmetric subspace $\Sym_k(\bbC^d)$ of $(\bbC^d)^{\otimes k}$ with $d=2^n$,  $D_{[k]}$ is the rank of $P_{[k]}$ or the dimension of $\Sym_k(\bbC^d)$, and $W_a, W_b$ are $n$-qubit Pauli operators. In deriving the las equality in \eref{aeq:SecondMoment}, we have made use of the following formula
\begin{equation}\label{eq:TraceFormula}
\tr\left[ P_{[8]} (W_a^{\otimes4}\otimes  W_b^{\otimes4})\right]=\begin{cases}
D_{[8]} & W_a=W_b=1,\\
\frac{D_{[8]}}{D_{[4]}}\frac{3d^2+6d}{24} &W_a=1, W_b\neq1\; \mbox{or}\; W_b=1, W_a\neq1,\\
\frac{1}{2688}(7d^4+84d^3+308d^2+336d) &W_a=W_b\neq1,\\
\frac{1}{4480}(d^4+28d^3+236d^2+560d) &W_a,W_b\neq1, W_aW_b =W_bW_a,\\
\frac{1}{4480}(d^4+12d^3+44d^2+48d) &W_a,W_b\neq1, W_aW_b =-W_bW_a.\\
\end{cases}
\end{equation}
\begin{table}
	\caption{\label{tab:S8}Permutations of $S_8$ without cycle of odd length. $N_1$ is the number of permutations of a given cycle type; $N_2$ is the number of balanced  permutations (those in $\mathscr{A}$) of a given cycle type; $N_3=N_{3+}-N_{3-}$,  where $N_{3\pm}$ is the number  of permutations of a given cycle type  that belong to $\mathscr{A}_\pm$.  The sets $\mathscr{A}$ and  $\mathscr{A}_\pm$ are defined in the text. Note that $N_{3+}+N_{3-}=N_2$.	}
	
	\centering
	\begin{math}
	\begin{array}{c|ccccc}
	\hline\hline
	\mbox{cycle type} & (2^4) & (2^2, 4) & (4^2) & (2, 6) & (8) \\
	N_1 & 105 & 1260 & 1260 & 3360 & 5040 \\
	N_2& 9 & 252 & 684 & 1440 & 5040 \\
	N_3  & 9 & 108 & 108 & 288 & 432\\
	\hline\hline
	\end{array}
	\end{math}
\end{table}
To derive \eref{eq:TraceFormula}, we 
recall  the following facts,
\begin{align}
P_{[k]}&=\frac{1}{k!}\sum_{\sigma\in S_k}U_\sigma,  \label{eq:PSk}  \qquad \tr_k P_{[k]}=\frac{D_{[k]}}{D_{[k-1]}}P_{[k-1]},
\end{align}
where $\tr_k$ means the partial trace over party $k$. 
If  $a\neq 0$, then
\begin{equation}\label{eq:UWa}
\tr(U_\sigma W_a^{\otimes k})=
\begin{cases}
0 & \mbox{$\sigma$ contains a cycle of odd length},\\
d^{l(\sigma)} & \mbox{otherwise}.
\end{cases}
\end{equation}
where  $l(\sigma)$ is the number of cycles in $\sigma$ of even lengths. The cycle types of elements in  $S_8$ without cycle of odd length are listed in \tref{tab:S8}.

The first case in \eref{eq:TraceFormula} is trivial. When $W_b=1, W_a\neq1$,
\begin{equation}
\tr\left[ P_{[8]} (W_a^{\otimes4}\otimes  W_b^{\otimes4})\right]=\frac{D_{[8]}}{D_{[4]}}\tr\left( P_{[4]} W_a^{\otimes4}\right)=\frac{D_{[8]}}{D_{[4]}}\frac{3d^2+6d}{24},
\end{equation}
recall that  the symmetric group $S_4$ has three permutations of cycle type $(2^2)$, six permutations of cycle type (4), and all other permutations contain at least one cycle of odd length (cf.~\sref{tab:CharS4}). The case $W_a=1, W_b\neq1$ has the same result. When $W_b= W_a\neq1$, the result follows from \esref{eq:PSk}, \eqref{eq:UWa}, and \tref{tab:S8}.

To settle the last two cases in \eref{eq:TraceFormula}, we need to introduce some terminology. A permutation in $S_8$ is balanced if each cycle involves even number of parties both in the first four parties and in the second four parties. Define $\mathscr{A}$ as the subset of balanced permutations in $S_8$. 
Each permutation in $S_8$ induces a permutation on the vector $v=(a,a,a,a,b,b,b,b)$. Define
\begin{align}
\mathscr{A}_+&=\{\sigma\in \mathscr{A}|\mbox{ $\sigma$ induces even number of transpositions between $a$ and $b$}\},\\
\mathscr{A}_-&=\{\sigma\in \mathscr{A}|\mbox{ $\sigma$ induces odd number of transpositions between $a$ and $b$}\}.
\end{align}
Note that $\mathscr{A}=\mathscr{A}_+\cup\mathscr{A}_-$. 
If $W_b, W_a\neq1$,   $W_b\neq W_a$, and $W_b W_a=W_a W_b$, then
\begin{equation}\label{eq:UWabCom}
\tr\left[U_\sigma (W_a^{\otimes 4}\otimes W_b^{\otimes 4})\right]=
\begin{cases}
d^{l(\sigma)} & \sigma\in\mathscr{A},\\
0 & \sigma\notin\mathscr{A}.
\end{cases}
\end{equation}
If   $W_b W_a=-W_a W_b$, then
\begin{equation}\label{eq:UWaAntiCom}
\tr\left[U_\sigma (W_a^{\otimes 4}\otimes W_b^{\otimes 4})\right]=
\begin{cases}
d^{l(\sigma)} & \sigma\in\mathscr{A}+,\\
-d^{l(\sigma)} & \sigma\in\mathscr{A}_-,\\
0 & \mbox{otherwise}.
\end{cases}
\end{equation}
Now the  last two cases in \eref{eq:TraceFormula} can be determined by virtue of \tref{tab:S8} and the above two equations.

\section{Proof  of \lref{lem:LipConstant}}
\begin{proof}
Let $X=|\psi\rangle\langle \psi|- |\varphi\rangle\langle\varphi|$ and suppose $X$ has spectral decomposition $X=\lambda(|\mu\rangle\langle\mu|- |\nu\rangle\langle\nu|)$ with $0\leq \lambda\leq 1$. Then $\|X\|_1=2\lambda$, $\|X\|_2=\sqrt{2}\lambda$, and we have
\begin{align}
\|\psi-\varphi\|_2^2&=2-\langle \psi|\varphi\rangle-\langle \varphi|\psi\rangle\geq 2-2 |\langle \psi|\varphi\rangle|\nonumber\\
\lambda^2&=1-|\langle \psi|\varphi\rangle|^2\leq 2-2 |\langle \psi|\varphi\rangle|\leq \|\psi-\varphi\|_2^2,
\end{align}
which implies that $\lambda\leq \|\psi-\varphi\|_2$.
In addition,
\begin{align}
&\alpha_+(\psi)-\alpha_+(\varphi)=  \tr\left[P_{n,4}(|\psi\rangle\langle \psi|)^{\otimes 4} \right]-\tr\left[P_{n,4}(|\varphi\rangle\langle \varphi|)^{\otimes 4} \right]\nonumber\\
&=\tr\left[P_{n,4}(|\varphi\rangle\langle \varphi|+X)^{\otimes 4} \right]-\tr\left[P_{n,4}(|\varphi\rangle\langle \varphi|)^{\otimes 4} \right]\nonumber\\
&=  4\tr\left\{P_{n,4}\left[(|\varphi\rangle\langle \varphi|)^{\otimes 3}\otimes X\right]\right\}\
+6\tr\left\{P_{n,4}\left[(|\varphi\rangle\langle \varphi|)^{\otimes 2}\otimes X^{\otimes 2}\right]\right\}+4\tr\left\{P_{n,4}\left[(|\varphi\rangle\langle \varphi|)\otimes X^{\otimes 3}\right]\right\}\
\nonumber\\
&\quad+\tr\left(P_{n,4} X^{\otimes 4}\right).
\end{align}
According to \lref{lem:bound1234} below, if $0\leq \lambda\leq 1/2$, then 
\begin{align}\label{aeq:bound0}
&\alpha_+(\psi)-\alpha_+(\varphi)\leq \frac{1}{d}\left[4\lambda(1-\lambda) +6\sqrt{2}\lambda^2(1+\lambda) +8\lambda^3(1-\lambda)+8\lambda^4\right]=\frac{1}{d}\lambda f(\lambda)\leq \frac{1}{d}f(\lambda)\|\psi-\varphi\|_2,
\end{align}
where 
\begin{align}
f(\lambda):=4 +(6\sqrt{2}-4)\lambda +(8+6\sqrt{2})\lambda^2.
\end{align}
Note that $f(\lambda)$ increases monotonically with $\lambda$ when $\lambda\geq0$. 
If $0\leq \lambda\leq 1/5.4$, then $f(\lambda)\leq f(1/5.4)<5.4$, so that $\alpha_+(\psi)-\alpha_+(\varphi)\leq 5.4\|\psi-\varphi\|_2/d$. If $1/5.4\leq \lambda\leq 1$, then 
\begin{equation}
\alpha_+(\psi)-\alpha_+(\varphi)\leq  \alpha_+(\psi) \leq \frac{1}{d}\leq \frac{5.4\lambda }{d}\leq \frac{5.4 }{d}\|\psi-\varphi\|_2,
\end{equation}
where we have applied \eref{eq:Devbound}. 
By symmetry we have 
	\begin{equation}
	|\alpha_+(\psi)-\alpha_+(\varphi)|\leq  \frac{5.4}{d}\|\psi-\varphi\|,
	\end{equation}
which confirms the first inequality in  \lref{lem:LipConstant}. The second inequality in the lemma is an immediate consequence of the first one and \eref{eq:ONdeviation}.
\end{proof}

\begin{lemma}\label{lem:bound1234}
	Suppose $\psi,\varphi$ are two normalized $n$-qubit state vectors, $X=|\psi\rangle\langle \psi|- |\varphi\rangle\langle\varphi|$, and 
	$\lambda=\|X\|_1/2$. Then 
\begin{align}
&\tr\left\{P_{n,4}\left[(|\varphi\rangle\langle \varphi|)^{\otimes 3}\otimes X\right]\right\} \leq
\begin{cases}
\frac{\lambda(1-\lambda)}{d} &0\leq \lambda\leq \frac{1}{2}, \\
\frac{1}{4d} & 0\leq \lambda\leq 1;
\end{cases}\\
&\tr\left\{P_{n,4}\left[(|\varphi\rangle\langle \varphi|)^{\otimes 2}\otimes X^{\otimes 2}\right]\right\}
\leq \min\left\{\frac{2\lambda^2}{d}, \frac{ \sqrt{2}\lambda^2(1+\lambda)}{d}\right\};\\
&\tr\left\{P_{n,4}\left[(|\varphi\rangle\langle \varphi|)\otimes X^{\otimes 3}\right]\right\}\leq
\begin{cases}
\frac{2\lambda^3(1-\lambda)}{d} &0\leq \lambda\leq\frac{1}{2}, \\
\frac{\lambda^2}{2d} &\frac{1}{2}\leq \lambda\leq 1;
\end{cases}\\
&\tr\left(P_{n,4} X^{\otimes 4}\right)\leq \frac{8\lambda^4}{d}.
\end{align}
\end{lemma}
\begin{proof}
Suppose $X$ has spectral decomposition $X=\lambda(|\mu\rangle\langle\mu|- |\nu\rangle\langle\nu|)$. Then 
\begin{align}
&\tr\left\{P_{n,4}\left[(|\varphi\rangle\langle \varphi|)^{\otimes 2}\otimes X^{\otimes 2}\right]\right\}\leq\lambda^2\tr\left\{P_{n,4}\left[(|\varphi\rangle\langle \varphi|)^{\otimes 2}\otimes (|\mu\rangle\langle\mu|)^{\otimes 2}+(|\varphi\rangle\langle \varphi|)^{\otimes 2}\otimes (|\nu\rangle\langle\nu|)^{\otimes 2} \right]\right\}\leq \frac{2\lambda^2}{d},\\
&\tr\left(P_{n,4} X^{\otimes 4}\right)
\leq\lambda^4\tr\left\{P_{n,4}\left[(|\mu\rangle\langle \mu|)^{\otimes 4} +(|\nu\rangle\langle \nu|)^{\otimes 4}  +6(|\mu\rangle\langle \mu|)^{\otimes 2}\otimes (|\nu\rangle\langle \nu|)^{\otimes 2}\right]\right\}\leq \frac{8\lambda^4}{d},
\end{align}
where we have applied \lref{lem:ProjBound}.

According to \lref{lem:bound5} below  with $W$ being a  Pauli operator, we also have
\begin{align}
&\tr\left\{P_{n,4}\left[(|\varphi\rangle\langle \varphi|)^{\otimes 3}\otimes X\right]\right\}=\frac{1}{d^2}\sum_a \tr\left\{W_a^{\otimes 4} \left[(|\varphi\rangle\langle \varphi|)^3\otimes X\right]\right\}=\frac{1}{d^2}\sum_a (\langle \varphi| W_a|\varphi\rangle)^3 \tr(W_aX)\nonumber\\
&\leq 
\begin{cases}
\frac{\lambda(1-\lambda)}{d^2}\sum_a |\langle \varphi| W_a|\varphi\rangle^2 =\frac{\lambda(1-\lambda)}{d}& 0\leq \lambda\leq \frac{1}{2} \\
\frac{1}{4d^2}\sum_a |\langle \varphi| W_a|\varphi\rangle^2 =\frac{1}{4d}&  \frac{1}{2}\leq \lambda\leq 1.
\end{cases} 
\end{align}	

\begin{align}\label{eq:bound2}
&\tr\left\{P_{n,4}\left[(|\varphi\rangle\langle \varphi|)^{\otimes 2}\otimes X^{\otimes 2}\right]\right\}=\frac{1}{d^2}\sum_a \tr\left\{W_a^{\otimes 4} \left[(|\varphi\rangle\langle \varphi|)^2\otimes X^{\otimes 2}\right]\right\}=\frac{1}{d^2}\sum_a (\langle \varphi| W_a|\varphi\rangle)^2 [\tr(W_aX)]^2\nonumber\\
&\leq \frac{\lambda(1+\lambda) }{d^2}\sum_a |\langle \varphi| W_a|\varphi\rangle \tr(W_aX)|\leq \frac{\lambda(1+\lambda) }{d^2}
\left\{\sum_a (\langle \varphi| W_a|\varphi\rangle)^2
\sum_b  [\tr(W_bX)]^2\right\}^{1/2}
\nonumber\\
&=\frac{\lambda(1+\lambda) }{d} \|X\|_2=\frac{\sqrt{2}\lambda^2(1+\lambda)}{d}.
\end{align}

\begin{align}\label{eq:bound3}
&\tr\left\{P_{n,4}\left[(|\varphi\rangle\langle \varphi|)\otimes X^{\otimes 3}\right]\right\}=\frac{1}{d^2}\sum_a \tr\left\{W_a^{\otimes 4} \left[(|\varphi\rangle\langle \varphi|)\otimes X^{\otimes 3}\right]\right\}=\frac{1}{d^2}\sum_a \langle \varphi| W_a|\varphi\rangle [\tr(W_aX)]^3\nonumber\\
&\leq
\begin{cases}
\frac{\lambda(1-\lambda)}{d^2}\sum_a [\tr(W_aX)]^2=\frac{\lambda(1-\lambda)}{d}\|X\|_2^2=\frac{2\lambda^3(1-\lambda)}{d}& 0\leq \lambda\leq \frac{1}{2},\\
\frac{1}{4d^2}\sum_a [\tr(W_aX)]^2=\frac{1}{4d}\|X\|_2^2=\frac{\lambda^2}{2d}& \frac{1}{2}\leq \lambda\leq 1. 
\end{cases} 
\end{align}
\end{proof}

\begin{lemma}\label{lem:bound5}
	Suppose $\psi,\varphi$ are two normalized  state vectors, $X=|\psi\rangle\langle \psi|- |\varphi\rangle\langle\varphi|$,
	$\lambda=\|X\|_1/2$,  and  $W$ is an operator satisfying $-\Id\leq W\leq \Id$. Then 
	\begin{equation}
-\lambda(1+\lambda)\leq 	\langle \varphi| W|\varphi\rangle \tr(WX)\leq 
\begin{cases}
\lambda(1-\lambda) & 0\leq \lambda\leq \frac{1}{2},\\
\frac{1}{4} &\frac{1}{2}\leq \lambda\leq 1.
\end{cases}
	\end{equation}
\end{lemma}
\begin{proof}
Let $P=(\Id+W)/2$, then $0\leq P\leq\Id$. In addition,
\begin{align}
	\langle \varphi| W|\varphi\rangle \tr(WX)
	&=	\langle \varphi| (2P-\Id)|\varphi\rangle \tr[(2P-\Id)X]
	=	(2\langle \varphi| P|\varphi\rangle-1) \tr(2PX)=	(2\langle \varphi| Q|\varphi\rangle-1) \tr(2QX).
\end{align}
where $Q$ is the projection of $P$ onto the two-dimensional subspace spanned by $\psi$ and $\varphi$, which satisfies $0\leq Q\leq \Id$. Now the lemma follows from \lref{lem:bound6} below.
\end{proof}

\begin{lemma}\label{lem:bound6}
	Suppose $\rho_1$ and $\rho_2$ are two qubit pure states,  $X=\rho_2-\rho_1$, and $\lambda=\|X\|_1/2$. Let $Q$ be a positive operator on the qubit which satisfies $Q\leq \Id$. Then 
	\begin{equation}\label{aeq:bound6}
	-\lambda(1+\lambda)\leq [2\tr(Q\rho_1)-1] \tr(2QX)\leq 
	\begin{cases}
	\lambda(1-\lambda) & 0\leq \lambda\leq \frac{1}{2},\\
	\frac{1}{4} &\frac{1}{2}\leq \lambda\leq 1.
	\end{cases}
	\end{equation}
	Here both the lower bound and the upper bound can be saturated. 	
\end{lemma}
\begin{proof}
With a suitable unitary transformation if necessary, $\rho_1,\rho_2$ can be written as follows,
\begin{align}
\rho_1=\frac{1+\vec{r}_1\cdot\vec{\sigma}}{2},\quad
\vec{r}_1=(\sin\theta,0,\cos\theta); \qquad \rho_2=\frac{1+\vec{r}_2\cdot\vec{\sigma}}{2}, \quad \vec{r}_2=(-\sin\theta,0,\cos\theta),
\end{align}
where $\vec{\sigma}$ is the vector composed of the three Pauli matrices $\sigma_x, \sigma_y, \sigma_z$ and  $\theta=\arcsin\lambda$ with $0\leq \theta\leq \pi/2$. Similarly,  the operator  $Q$ can be expanded in terms of the Pauli basis. Moreover, to achieve the maximum or minimum in \eref{aeq:bound6}, the coefficient of $\sigma_y$ can be set to zero, so we have
\begin{equation}
Q=\frac{1}{2}(1+a+b\sigma_x+c\sigma_z),
\end{equation}
where $a,b,c $ are real constants. The   constraint $0\leq Q\leq \Id$ amounts to the inequality
\begin{equation}
|a|+\sqrt{b^2+c^2}\leq 1,
\end{equation}
which defines the region of feasible solutions, denoted by $\mathcal{R}$ below. 
 Geometrically,  the region $\mathcal{R}$ is the intersection of two opposite cones, which is convex. In addition $\mathcal{R}$ is symmetric under inversion of each of the three coordinates $a,b,c$. With this notation, we have
\begin{align}\label{aeq:fabc}
[2\tr(Q\rho_1)-1] \tr(2QX)=-\sin\theta f(a,b,c),
\end{align}
where
\begin{equation}
f(a,b,c):=2b(a+b\sin\theta+c\cos\theta),
\end{equation}

Since $f(a,b,c)$ is invariant under inversion, to determine the maximum of $f(a,b,c)$ over $\mathcal{R}$, we may assume that $b\geq0$, and then it suffices to consider the case $a,c\geq0$. 
In addition, for given $b$, $f(a,b,c)$ is linear in $a,c$, so the maximum of  $f(a,b,c)$ over the convex region $\mathcal{R}$ can be attained at the boundary of $\mathcal{R}$. Define $s:=\sqrt{b^2+c^2}$ and $\phi:=\arccos(c/s)$; then  $0\leq s\leq 1$, $0\leq \phi\leq \pi/2$,    $b=s\sin\phi$, $c=s\cos\phi$, and $a=1-s$ under the above assumptions. Consequently,
\begin{align}
f(a,b,c)&=2s\sin\phi[1-s+s\cos(\phi-\theta)]=2\sin\phi\{-s^2 [1-\cos(\phi-\theta)]  +s\}.
\end{align}
If $\cos(\phi-\theta)\leq 1/2$, then 
\begin{align}
f(a,b,c)\leq \frac{\sin\phi}{2[1-\cos(\phi-\theta)]}\leq \sin\phi\leq 1+\sin\theta.
\end{align}
If $\cos(\phi-\theta)\geq 1/2$, then 
\begin{align}
f(a,b,c)&\leq 2\sin\phi\cos(\phi-\theta)
=\sin \theta+\sin(2\phi-\theta)\leq 1+\sin \theta.
\end{align}

Similarly,  to determine the minimum of $f(a,b,c)$ over $\mathcal{R}$, we may assume that $b\geq0$ and  $a,c\leq 0$. 
In addition, it suffices to consider the boundary of $\mathcal{R}$, so that $a=\sqrt{b^2+c^2}-1$. Define $s:=\sqrt{b^2+c^2}$ and $\phi:=\arccos(c/s)$; then  $0\leq s\leq 1$, $\pi/2\leq \phi\leq \pi$,  $b=s\sin\phi$, $c=s\cos\phi$, and $a=s-1$. Consequently,
\begin{align}
f(a,b,c)&=2s\sin\phi[s-1+s\cos(\phi-\theta)]=2\sin\phi\{s^2 [1+\cos(\phi-\theta)]  -s\}.
\end{align}
If $\cos(\phi-\theta)\geq -1/2$, that is, $\phi\leq(2\pi/3)+\theta$, then 
\begin{align}
f(a,b,c)\geq -\frac{\sin\phi}{2[1+\cos(\phi-\theta)]}
\geq-\frac{\sin(\pi-\theta)}{2[1+\cos(\pi-\theta-\theta)]}
=-\frac{1}{4\sin\theta}.
\end{align}
If $\cos(\phi-\theta)\leq -1/2$, that is, $\phi\geq(2\pi/3)+\theta$, then 
\begin{align}
f(a,b,c)&\geq 2\sin\phi\cos(\phi-\theta)=\sin \theta+\sin(2\phi-\theta)\geq \sin \theta-1\geq-\frac{1}{4\sin\theta}.
\end{align}
If $0\leq \theta\leq \pi/6$ and $\phi\leq(2\pi/3)+\theta$ then 
\begin{align}
f(a,b,c)\geq -\frac{\sin\phi}{2[1+\cos(\phi-\theta)]}
\geq -\frac{\sin\left(\frac{2\pi}{3}+\theta\right)}{2[1+\cos(2\pi/3)]}=-\sin\left(\frac{2\pi}{3}+\theta\right)\geq \sin \theta-1.
\end{align}

In summary we have
\begin{equation}\label{aeq:fbound}
\begin{cases}
-(1-\sin\theta)\leq f(a,b,c)\leq 1+\sin\theta &0\leq \theta\leq\pi/6,\\
-\frac{1}{4\sin\theta}\leq f(a,b,c)\leq 1+\sin\theta &\pi/6\leq \theta\leq\pi/2,
\end{cases}
\end{equation}
which implies  \eref{aeq:bound6} given  \eref{aeq:fabc} and the inequality $\sin\theta=\lambda$. The upper bound in \eref{aeq:fbound} is  saturated when $s=1$ and $\phi=(\pi+2\theta)/4$, in which case $a=0$, $b=\sqrt{(1+\sin\theta)/2}$, $c=\sqrt{(1-\sin\theta)/2}$, and
\begin{equation}
Q=\frac{1}{2}\left(1+\sqrt{\frac{1+\sin\theta}{2}}\sigma_x+\sqrt{\frac{1-\sin\theta}{2}}\sigma_z\right).
\end{equation}
If  $0\leq \theta\leq\pi/6$, the lower bound in \eref{aeq:fbound} is  saturated when $s=1$ and $\phi=(3\pi+2\theta)/4$, in which case $a=0$, $b=\sqrt{(1-\sin\theta)/2}$, $c=-\sqrt{(1+\sin\theta)/2}$, and
\begin{equation}
Q=\frac{1}{2}\left(1+\sqrt{\frac{1-\sin\theta}{2}}\sigma_x-\sqrt{\frac{1+\sin\theta}{2}}\sigma_z\right).
\end{equation}
If $\pi/6\leq \theta\leq\pi/2$, the lower bound is saturated when $s=1/(4\sin^2\theta)$ and $\phi=\pi-\theta$, in which case $a=[1/(4\sin^2\theta)]-1$, $b=1/(4\sin\theta)$, $c=-\cos\theta/(4\sin^2\theta)$, and 
\begin{equation}
Q=\frac{1}{2}\left(\frac{1}{4\sin^2\theta}+\frac{1}{4\sin\theta}\sigma_x-\frac{\cos\theta}{4\sin^2\theta}\sigma_z\right)
=\frac{1}{8\sin^2\theta}\left(1+\sin\theta\sigma_x-\cos\theta\sigma_z\right).
\end{equation}
Therefore, 	 both the lower bound and the upper bound in \eref{aeq:bound6} can be saturated. 
\end{proof}

\section{\label{sec:TensorFiducial}Proof of \pref{pro:TensorFiducial}}
\begin{proof}
	Let $d_j=2^{n_j}$ and suppose  $m=4$.  Then
	\begin{equation}
	\|\Xi(\psi)\|_{\ell_4}^4=\prod_{j=1}^4\|\Xi(\psi_j)\|_{\ell_4}^4\geq \prod_{j=1}^4 \frac{2d_j}{d_j+1}\geq \Bigl(\frac{4}{3}\Bigr)^3\frac{2^{n-2}}{2^{n-3}+1},
	\end{equation}
	where the first inequality follows from \eref{eq:Xibound}.
	If $n\geq 5$, then
	\begin{equation}
	\Bigl(\frac{4}{3}\Bigr)^3\frac{2^{n-2}}{2^{n-3}+1}-\frac{4d}{d+3}=\frac{5\times 2^{n+2}(2^n-24)}{27(2^n+3)(2^n+8)}>0.
	\end{equation}
	So the vector $\psi$ cannot generate a 4-design, recall that $\psi$ is a 4-design fiducial iff $	\|\Xi(\psi)\|_{\ell_4}^4=4d/(d+3)$, where $d=2^n$.
	
	Now suppose $m=3$, so that $n_1+n_2+n_3=n$. If $n_3=2$, then \begin{equation}
	\|\Xi(\psi)\|_{\ell_4}^4\geq \Bigl(\frac{8}{5}\Bigr)^3=\frac{512}{125}\geq 4>\frac{4d}{d+3},
	\end{equation}
so $\psi$ cannot be a 4-design fiducial. If $n_3=1, n_1,n_2\geq 3$, then $n\geq 7$,
	\begin{equation}
	\|\Xi(\psi)\|_{\ell_4}^4-\frac{4d}{d+3}\geq \prod_{j=1}^4 \frac{2d_j}{d_j+1}-\frac{2^{n+2}}{2^n+3}\geq \frac{4}{3}\times \frac{16}{9}\times\frac{2^{n-3}}{2^{n-4}+1}-\frac{2^{n+2}}{2^n+3}=\frac{ 2^{n+2}(5\times2^n-336)}{27(2^n+3)(2^n+16)}>0.
	\end{equation}
	So $\psi$ cannot be a 4-design fiducial.
	If $n_3=1, n_1,n_2\geq 2$, then $n\geq 5$,
	\begin{equation}
	\|\Xi(\psi)\|_{\ell_4}^4-\frac{4d}{d+3}\geq \prod_{j=1}^4 \frac{2d_j}{d_j+1}-\frac{2^{n+2}}{2^n+3}\geq \frac{4}{3}\times\frac{8}{5}\times\frac{2^{n-2}}{2^{n-3}+1}-\frac{2^{n+2}}{2^n+3}=\frac{ 2^{n+2}(2^n-72)}{15(2^n+3)(2^n+8)}.
	\end{equation}
	If in addition $n\geq 7$, then
	$\|\Xi(\psi)\|_{\ell_4}^4> 4d/(d+3)$, so that $\psi$ cannot be a 4-design fiducial.  This observation completes the proof of the proposition.
\end{proof}

\section{Notes on multivariate polynomials}
\label{sec:polynomials}

\subsection{Real case}

In this section, we recall several standard facts about multivariate polynomials. 
This is intended mainly for the benefit of readers from quantum information theory, where the notions discussed below do not seem to be commonly known.
The material presented here is based on Refs.~\cite{Proc07book,Rudi80book,GoodW09book,RoyS14, DoheW13}.

We start by considering the real vector space $\mathbb{R}^d$.
Let $\Hom_t(\RR^d)$ be the set of polynomial functions on
$\RR^d$ that is homogeneous of degree $t$. 

There is a close connection between polynomials, symmetric multilinear
forms, and totally symmetric tensors.  
Recall that the symmetric group $S_t$ acts on the tensor product space
$(\RR^d)^{\otimes t}$ by permuting tensor factors:
\begin{equation}\label{eqn:perms}
	\pi (u_1\otimes\dots\otimes u_t) 
	= 
	u_{\pi_1} \otimes \dots\otimes u_{\pi_k}
	\qquad
	\forall u_i\in\RR^d, \pi\in S_t.
\end{equation}
Let $\Sym_t\big( \mathbb{R}^d \big)$ be the space of \emph{totally symmetric
degree-$t$ tensors}:
\begin{equation*}
	\Sym_t\big( \mathbb{R}^d \big) = \{ f \in \mathbb{R}^{\otimes t}
	\,|\, \pi f = f \,\forall \pi \in S_t \}.
\end{equation*}
Every $f\in\Sym_t\big( \mathbb{R}^d \big)$ specifies a symmetric
$t$-linear form
on $\mathbb{R}^d$ by setting
\begin{equation*}
	f(u_1, \dots, u_t) := \langle f, u_1 \otimes \dots \otimes
	u_t \rangle \qquad u_i\in\RR^d. 
\end{equation*}
Conversely, every symmetric $t$-linear form arises in this way, and we
will not distinguish between symmetric tensors and symmetric forms in
the following\footnote{
	If, instead of $\RR^d$, one starts with a general linear space $V$
	for which no canonical scalar product has been specified, it would
	be cleaner to work with the symmetric tensor powers $\Sym_t( V^* )$
	of the \emph{dual space} $V^*$. Here, however, we find it advantageous to to
	identify $(\RR^d)^*$ with $\RR^d$ via the standard scalar product
	whenever necessary.
}.
By restricting all $u_i$ to be equal, we obtain a homogeneous
order-$t$ polynomial $p_f\in\Hom_t(\RR^d)$:
\begin{equation}\label{eqn:pf}
	p_f(u) := 
	\langle f, u\otimes \dots \otimes u  \rangle.
\end{equation}

It is a less trivial fact that this relation between symmetric
$t$-linear forms and homogeneous order-$t$ polynomials is one-one,
i.e.\ that one can recover $f$ from the restriction $p_f$.
The map $p_f \mapsto f$ is called the \emph{polarization map}
(c.f.\ \cite[Chapter~3.2]{Proc07book}).
To see how this works, we construct the polarization map explicitly
for a basis of $\Hom_t(\RR^d)$.
Let $\mu\in\mathbb{N}_0^d$ be a vector of $d$ non-negative integers
summing to $t$. 
Then $\mu$ is called a \emph{partition of $t$ into $d$ parts}. 
The polynomial $x^\mu\in\Hom_t(\RR^d)$ defined by
\begin{equation*}
	\RR^d \ni x = (x_1, \dots, x_d) \mapsto
	x^\mu:= \prod_{i=1}^d x_i^{\mu_i}
\end{equation*}
is the \emph{monomial associated with $\mu$}. 
By definition, the degree-$t$ monomials span $\Hom_t(\RR^d)$ and they
are easily seen to be linearly independent as functions on $\RR^d$
(c.f.\ Lemma~\ref{lem:complexPolarization}, where we give a proof for
the complex case).
The symmetric vector
\begin{equation*}
	e_\mu:=\frac{1}{|S_t|} \sum_{\pi\in S_t} \pi\,( e_1^{\otimes
	\mu_1}\otimes \dots\otimes e_d^{\otimes \mu_d}) \quad \in
	\Sym_t(\RR^d)
\end{equation*}
clearly fulfils
\begin{equation*}
	\langle e_\mu, x^{\otimes t}\rangle = x^\mu \qquad \forall x\in\RR^d
\end{equation*}
and thus constitutes the polarization of $x^\mu$. In fact, we have:

\begin{lemma}\label{lem:realPolarization}
	The relations
	\begin{equation}\label{eqn:realpol}
		x^\mu 
		\mapsto e_\mu\qquad \forall 
		\mu \in\mathbb{N}_0^d, \mu \text{ partition of }t
	\end{equation}
	define a linear isomorphism from $\Hom_t(\RR^d)$ to
	$\Sym_t(\RR^d)$.
	Its inverse is
	\begin{equation*}
		f \mapsto p_f, \qquad p_f(x):= \langle f, x^{\otimes t} \rangle.
	\end{equation*}
\end{lemma}

\begin{proof}
	It remains to be shown that the map (\ref{eqn:realpol}) is onto,
	i.e.\ that the $\{e_\mu\}_{\mu}$ span $\Sym_t(\RR^d)$. 
	That is true because the set $\{e_\mu\}$ constitutes the projection
	onto $\Sym_t(\RR^d)$ of the standard tensor product basis
	$\{e_{i_1}\otimes \dots \otimes e_{i_t}\}_{i_j \in\RR^d}$. 
\end{proof}

As a corollary, we see directly that the set of tensor powers spans
the totally symmetric space:
\begin{equation}\label{eqn:outerspan}
	\Sym_t(\RR^d) = 
	\langle 
		\{ u^{\otimes t} \,|\, u \in \RR^d \}
	\rangle.
\end{equation}

We now turn to the relevant symmetries.
By definition, the orthogonal group $\rmO(d)$ acts on $\RR^d$. 
The action extends to the \emph{diagonal action} on degree-$t$ tensors
$(\RR^d)^{\otimes t}$: 
\begin{equation*}
	O:
	u_1\otimes\dots\otimes u_t \mapsto 
	(O\otimes\dots\otimes O)	\,
	(u_1\otimes\dots\otimes u_t)
	=
	(Ou_{1}) \otimes
	\dots\otimes (Ou_{t}) \qquad u_i\in\RR^d.
\end{equation*}
The diagonal representation of $\rmO(d)$ commutes with the action of the
symmetric group $S_t$ defined in Eq.~(\ref{eqn:perms}). This implies
in particular that $\Sym_t(\RR^d)$ is an invariant subspace 
and thus carries a representation of $\rmO(d)$. 
By Eq.~(\ref{eqn:pf}), the action coincides with the natural action
\begin{equation*}
	O: p( \cdot ) \mapsto p(O^{-1} \cdot) =  p( O^{T} \cdot )
\end{equation*}
of $\rmO(d)$ on polynomials in $\Hom_t(\RR^d)$. The representation of $\rmO(d)$ on
$\Sym_t(\RR^d) \simeq \Hom_t(\RR^d)$ is reducible (this is an important
conceptual distinction to the complex case, described below).  We will
now describe the irreducible representations for the case of even
degree
-- first as subspaces of $\Sym_{t}(\RR)$ and then viewed as subspaces
of $\Hom_t(\RR^d)$.

Choose an ortho-normal basis $\{e_1, \dots, e_d\}$ of $\RR^d$. In the
case of of $t=2$, it follows directly from the defining property $O
O^\rmT = \mathbb{I}$ of the orthogonal group that
\begin{equation*}
	v_0 = \sum_{i=1}^d e_i \otimes e_i \in \Sym_2(\RR^d)
\end{equation*}
is an invariant vector. 
A fruitful way to think about this fact lies in the relation
\begin{equation}\label{eqn:v0innerprod}
	\langle v_0, u_1\otimes u_2\rangle = \langle u_1, u_2
	\rangle \qquad \forall\, u_i \in\RR^d,
\end{equation}
together with the fact that the orthogonal group preserves inner
products like those appearing on the right hand side. If we expand
tensors in coordinates
\begin{equation*}
	u = \sum_{i,j} u_{i,j}\,(e_i \otimes e_j)\,\in (\RR^d)^{\otimes 2},
\end{equation*}
then the inner product with $v_0$ corresponds to a contraction of the
indices
\begin{equation*}
	\langle v_0, u \rangle = \sum_i u_{i,i}.
\end{equation*}
We can apply these findings to higher orders $t>2$ by ``contracting
only two of the indices with $v_0$''. 
More precisely, define the \emph{contraction map} 
\begin{equation*}
	C: (\RR^d)^{\otimes t} \to (\RR^d)^{\otimes (t-2)}
\end{equation*}
by its actions  on product tensors as follows:
\begin{equation}\label{eqn:contraction}
	C: u_1\otimes \dots \otimes u_t \mapsto 
	\langle u_1, u_2\rangle \,
	(u_{3} \otimes \dots \otimes u_t).
\end{equation}
In coordinates, it corresponds to a contraction of the first two
indices
\begin{equation*}
	u_{i_1, \dots, i_t} \mapsto \sum_{i} u_{i,i,i_3, \dots, i_t}.
\end{equation*}
From (\ref{eqn:contraction}), it is clear that the kernel $\ker C$ is
an invariant subspace of the diagonal representation $O \mapsto
O^{\otimes t}$ of $\rmO(d)$. 
The same is true for its intersection 
\begin{equation*}
	H_t(\RR^d):=\ker C \cap \Sym_t(\RR^d)
\end{equation*}
which turns out to carry an irreducible representation.
We refer to $H_t$ as the \emph{harmonic totally symmetric
space}
of degree $t$.
From Eq.~(\ref{eqn:outerspan}), it follows that $C$ maps $\Sym_t(\RR^d)$
onto $\Sym_{t-2}(\RR^d)$ and thus
\begin{equation*}
	\Sym_{t-2}(\RR^d) \simeq 
	\Sym_{t}(\RR^d) / \ker C \simeq 
	(\ker C)^\perp,
\end{equation*}
where the ortho-complement is taken within $\Sym_t(\RR^d)$.
Therefore, we can embed $H_{t-2}(\RR^d)$ into $\Sym_t(\RR^d)$, as the
kernel of $C\circ C$ intersected with $(\ker C)^\perp$. Iterating this
procedure and setting
\begin{equation*}
	H_0(\RR^d) = \Sym_0(\RR^d)=\RR, 
\end{equation*}
one obtains the decomposition
\begin{equation}\label{eqn:harmonicdecomp}
	\Sym_t(\RR^d) 
	\simeq 
	\bigoplus_{i=0}^{t/2} H_{2i}(\RR^d).
\end{equation}
of $\Sym_t(\RR^d)$ into irreducible representations of $\rmO(d)$. 
The embedding $H_{t-2j}\to
\Sym_t(\RR^d)$ is given explicitly by 
\begin{align}\label{eqn:embeddingf}
	f \mapsto 
	P_{[t]}	(v_0^{\otimes j} \otimes f\big),
	\qquad
	P_{[t]}:=\frac1{t!}\sum_{\pi \in S_t} \pi.
\end{align}
In particular, the one-dimensional invariant space
$H_0(\RR^d)$ is realized as a multiple of
\begin{align}\label{eqn:embeddingfinv}
	v_0^{(t)} = P_{[t]}\, v_0^{\otimes t/2}.
\end{align}
We will use these embeddings implicitly from now on and treat $H_{2i}(\RR^d)$ as a subset of $\Sym_t(\RR^d)$.
The same convention will apply to related spaces for polynomials, and in the complex case.
It turns out that these embeddings are much more transparent from the point of view of
polynomials, as described next.

Using the language of homogeneous polynomials, one finds
that up to normalization, the contraction $C$ corresponds to the action of the \emph{Laplacian}.
Recall that the \emph{Laplacian} differential operator
\begin{equation*}
	\Delta=\sum_{i=1}^d 
	\frac{\partial^2}{\partial{x_i^2}}: \Hom_t(\RR^d) \to \Hom_{t-2}(\RR^d).
\end{equation*}
Simple direct calculations yield
\begin{align*}
	\Delta x^\mu &= \sum_{i=1}^d \mu_i(\mu_i-1) x^{\mu - 2 e_i}, \\
	C e_\mu &= \frac1{t(t-1)} \, \sum_{i=1}^d \mu_i (\mu_i-1) e_{\mu - 2 e_i}.
\end{align*}
Therefore, up to normalization, the contraction operator $C$ corresponds to the Laplacian in the sense that
\begin{equation*}
	\Delta p_f 
	=
	t(t-1)\,
	p_{C f} 
	\qquad \forall f \in \Sym_t(\RR^d).
\end{equation*}
Polynomials in the kernel of the Laplacian are called \emph{harmonic}
and we write 
\begin{equation*}
	\Harm_t(\RR^d) := \ker \Delta \cap \Hom_t(\RR^d).
\end{equation*}
As above, this gives rise to the decomposition
\begin{equation}\label{eqn:harmonicpolydecomp}
	\Hom_t(\RR^d)
	\simeq
	\bigoplus_{i=0}^{t/2} \Harm_{2i}(\RR^d)
\end{equation}
of the set of all homogeneous polynomials of degree $t$ into
irreducible spaces, equivalent to harmonic polynomials of lower
degrees.
The embedding of $\Harm_{t-2j}(\RR^d)\to\Hom_t(\RR^d)$ used in
Eq.~(\ref{eqn:harmonicpolydecomp}) maps $p\mapsto p'$, where
\begin{align}
	p'(u) &=
	\|u\|_2^{2j} p(u)\qquad \forall u\in \RR^d,
	\label{eqn:embedding}
\end{align}
and the one-dimensional space $\Harm_{0}(\RR^d)$ is realized 
within $\Hom^t(\RR^d)$
as  multiples of
\begin{equation}\label{eqn:embeddinginv}
	u \mapsto \|u\|_{2}^{t}.
\end{equation}
Equations (\ref{eqn:embedding}) and (\ref{eqn:embeddinginv}) are the
polynomial analogues of
Eqs.~(\ref{eqn:embeddingf}) and (\ref{eqn:embeddingfinv}),
respectively.

We are frequently concerned with the restriction
$\Hom_t(S^{d-1})$ of $\Hom_t(\RR^d)$ to the unit sphere
$S^{d-1}\subset \RR^d$. 
There, the embedding in Eq.~(\ref{eqn:embedding}) becomes trivial, so
that $\Harm_{2i}(S^{d-1})$ is a subset of $\Hom_t(S^{d-1})$. For any
homogeneous polynomial $p$ of degree $t$, it is true by definition
that
\begin{equation*}
	p(u)= \|u\|_2^{t}\,p\left(\frac{u}{\|u\|_2}\right).
\end{equation*}
Thus, homogeneous polynomials  are fully specified by their
restriction to $S^{d-1}$ and therefore $\Hom_t(\RR^d)\simeq
\Hom_t(S^{d-1})$.

Given an even number $t\geq0$,
a \emph{spherical $t$-design}  \cite{DelsGS77} for $S^{d-1}$ is a set of vectors $X\subset S^{d-1}$ such that
\begin{equation*}
	\frac1{|X|} \sum_{\psi \in X} \psi^{\otimes t} \in H_0(\RR^d).
\end{equation*}

It is a $(t+1)$-design if in addition 
\begin{equation*}
	\sum_{\psi \in X} \psi^{\otimes t+1} = 0.	
\end{equation*}
Clearly, this equation is satisfied automatically when the set $X$ is symmetric under inversion, that is, $-\psi\in X$ whenever $\psi\in X$.

Let $G\subset \rmO(d)$ be a subgroup of $\rmO(d)$. 
For simplicity, we take $G$ to be finite.
The spaces $H_{2i}$ are $G$-invariant. 
For each $i$, let $\{ b^{(2i)}_{j} \}_j$ be an orthonormal basis of the space of  $G$-invariant vectors in $H_{2i}(\RR^d)$.
If now $X = G \cdot \psi_0$ is a $G$-orbit, then 
\begin{equation*}
	\frac1{|X|} \sum_{\psi \in X} \psi^{\otimes t}  
	=
	\sum_{j,i} b^{(2i)}_j \langle b^{(2i)}_j, \psi_0^{\otimes t} \rangle
	=
	\sum_{j,i} b^{(2i)}_j \, p_{b^{(2i)}_j}(\psi_0).
\end{equation*}
In particular, $G\cdot \psi_0$ is a $t$-design if and only if $p_{b^{(2i)}_j}(\psi_0)=0$ for all $i\geq 1$ and all $j$.
This is equivalent to saying that $\psi_0$ is a root of all $G$-invariant harmonic polynomials.
If $G$ affords no harmonic invariants of degree $s$ for $1\leq s \leq t$, then the orbit of any vector $\psi_0$ constitutes a spherical $t$-design.

\subsection{Complex case}

In analogy to the real case, we define $\Hom_{t}(\CC^d)$ to be the
complex vector space of homogeneous polynomials in $d$ complex variables.  
The
definitions of $\Sym_t(\CC^d)$, the monomials $x^\mu$, and symmetric
tensors $e_\mu$ carry over from the real case. 
The map
$x^\mu \to e_\mu$ defines 
an (anti-linear)
isomorphism
$\Hom_t(\CC^d)\to\Sym_t(\CC^d)$. 
Thus, the complex analogue 
\begin{equation*}
	\Sym_t(\CC^d) = \langle \{u^{\otimes t} \,|\, u\in\CC^d\}\rangle
\end{equation*}
of Eq.~(\ref{eqn:outerspan}) holds.

In the main part of this paper, we frequently work with a notion
of ``sesquilinear'' polynomials.
To define this concept, let $p\in\Hom_{2t}(\CC^d)$. 
Then we define a function that is homogeneous of degree $t$ in
the coordinates $x$ with respect to the complex coordinates $x_1,
\dots, x_d$ on $\CC^d$ and also homogeneous of degree $t$ in the
complex conjugates $\bar x_1, \dots, \bar x_d$ via
\begin{equation*}
	x \mapsto p(\bar x, x).
\end{equation*}
We denote the set of all such polynomials of \emph{bi-degree $(t,t)$}
on $\CC^d$ by $\Hom_{(t,t)}(\CC^d)$.
The following lemma establishes the right notion of  polarization for
functions in $\Hom_{(t,t)}(\CC^d)$.

\begin{lemma}
\label{lem:complexPolarization}
	The relations
	\begin{equation}\label{eqn:complexpol}
		x^\mu \bar x^\nu 
		\mapsto e_\mu\, e_\nu^* \qquad \forall \mu, \nu
		\in\mathbb{N}_0^d\text{ partitions of }t
	\end{equation}
	define a linear isomorphism from $\Hom_{(t,t)}(\CC^d)$ to the set of
	linear maps $L(\Sym_t(\CC^d))$ on $\Sym_t(\CC^d)$.
	Its inverse is
	\begin{equation*}
		A \mapsto p_A, \qquad p_A(\bar x, x):=
		\tr \big(A ( x^{\otimes t}) \, (\bar x^{\otimes t})^\rmT\big).
	\end{equation*}
\end{lemma}

Our proof of the polarization relation relies on the notion of
\emph{Wirtinger derivatives}, introduced next. On $\CC^d$, one can
define real coordinates $(x_1, \dots, x_d, y_1, \dots, y_d)$ by
writing
\begin{equation*}
	\CC^d \ni z = (z_1, \dots, z_n) = (x_1 + \rmi y_1, \dots, x_n + \rmi y_n).
\end{equation*}
The Wirtinger derivatives are
\begin{equation*}
	\frac{\partial}{\partial{z_i}} := 
	\frac12\left(
		\frac{\partial}{\partial{x_i}} + \rmi 
		\frac{\partial}{\partial{y_i}}
	\right),
	\qquad
	\frac{\partial}{\partial{\bar z_i}} := 
	\frac12\left(
		\frac{\partial}{\partial{x_i}} - \rmi 
		\frac{\partial}{\partial{y_i}}
	\right).
\end{equation*}
The motivation for this definition is that the complex functions $z_1,
\dots, z_n$ and $\bar z_1, \dots, \bar z_n$ behave like $2n$
independent variables with respect to these derivatives in the sense
that
\begin{equation*}
	\frac{\partial}{\partial{z_i}}  z_j = \delta_{i,j}, \qquad
	\frac{\partial}{\partial{z_i}}  \bar z_j = 0, \qquad
	\frac{\partial}{\partial{\bar z_i}}  z_j = 0, \qquad
	\frac{\partial}{\partial{\bar z_i}}  \bar z_j = \delta_{i,j},
\end{equation*}
as one can easily verify. For  multi-indices $\mu,\nu\in\mathbb{N}_0^d$,
define
\begin{equation*}
	\partial_{(\mu,\nu)} := 
	\prod_{i=1}^d 
	\frac{1}{\mu_i! \nu_i!}\,
	\frac{\partial^{\mu_i}}{\partial{z_i}^{\mu_i}}  
	\frac{\partial^{\nu_i}}{\partial{\bar z_i}^{\nu_i}}.
\end{equation*}
Below, we will also use the \emph{complex Laplacian}
\begin{equation*}
	\Delta_\CC 
	:= 
	\sum_{i=1}^d 
	\frac{\partial}{\partial{z_i}}  
	\frac{\partial}{\partial{\bar z_i}}.
\end{equation*}

\begin{proof}[Proof of Lemma~\ref{lem:complexPolarization}]
	We first show that the set of monomials $\{x^\mu \bar x^\nu\}_{\mu,
	\nu}$ is linearly independent, where $\mu,\nu$ range over
	partitions of $t$ into $d$ parts. 
	Indeed, if $(\alpha,\beta)$ are two partitions of $t$, then 
	\begin{equation*}
		\partial_{(\alpha,\beta)}\, 
		x^{\mu}\bar x^{\nu} 
		= \delta_{\alpha, \mu}\delta_{\beta,\nu},
	\end{equation*}
	because for each $(\mu,\nu)\neq (\alpha,\beta)$, at least one of the
	variables has lower degree than the corresponding derivative. 
	Consequently,
	\begin{equation*}
		\partial_{(\alpha,\beta)}\, 
			\sum_{\mu,\nu} c_{\mu,\nu}\,x^{\mu}\bar x^{\nu}  
		= c_{\alpha,\beta},
	\end{equation*}
	and therefore, a linear combination of monomials cannot be zero
	unless all the coefficients are.

	It follows that the linear map of \eref{eqn:complexpol} is
	well-defined.
	It is onto, because the $\{e_\mu\}_\mu$ form a basis of
	$\Sym_t(\CC^d)$ and the rank-one outer products
	$\{ e_\mu e_\nu^* \}_{\mu,\nu}$  constitute a basis
	of the space of linear maps $L(\Sym_t(\CC^d))$. 
\end{proof}

For completeness, we mention a misconception that we have encountered
more than once. Namely, while
\begin{equation}\label{eqn:Ltensor}
	L\left( (\CC^d)^{\otimes t} \right)
	\simeq
	\left( L(\CC^d) \right)^{\otimes t},
\end{equation}
we only have
\begin{equation*}
	L\left( \Sym_t(\CC^d) \right) 
	\subset 
	\Sym_t\left( L(\CC^d) \right),
\end{equation*}
where the inclusion is proper. To understand the equivalence in
(\ref{eqn:Ltensor}), we make use of the fact that for any vector space
$V$, the space of  linear maps on $V$ is a tensor product space:
$L(V)\simeq V\otimes V^*$.  
Then the isomorphism
$L\left( (\CC^d)^{\otimes t} \right) \to 
	\left( L(\CC^d) \right)^{\otimes t}$
just amounts to  a re-ordering of tensor factors:
\begin{equation}\label{eqn:LtensorDef}
	(u_1\otimes \dots \otimes u_t)
	\otimes 
	(\alpha_1 \otimes \dots \otimes \alpha_t)
	\mapsto
	(u_1\otimes \alpha_1) \otimes \dots 
	\otimes 
	(u_t \otimes \alpha_t),
	\qquad
	u_i \in \CC^d, \quad \alpha_i \in (\CC^d)^*.
\end{equation}
We make liberal and implicit use of this identification throughout the paper (e.g.\ in 
statement
\emph{\ref{item:notquiteright}.}\ in 
Proposition~\ref{prop:designequiv}).

The isomorphism restricts to a map 
\begin{equation}\label{eqn:SymReorder}
	L\left( \Sym_t(\CC^d) \right) 
	\to
	\Sym_t\left( L(\CC^d) \right),
	\qquad
	(u\otimes \dots \otimes u)
	\otimes 
	(\alpha \otimes \dots \otimes \alpha)
	\mapsto
	(u\otimes \alpha) \otimes \dots 
	\otimes 
	(u \otimes \alpha).
\end{equation}
However, while the l.h.s.\ does span $L(\Sym_t(\CC^d))$, there is no
reason to believe that the r.h.s.\ spans $\Sym_t(L(\CC^d))$. In
fact, 
\begin{equation*}
	\dim 
	L\left( \Sym_t(\CC^d) \right) 
	=
	\binom{d+t-1}{t}^2
	<
	\binom{d^2+t-1}{t}
	=
	\dim \Sym_t\left( L(\CC^d) \right) \qquad \forall\,d,t \geq 2,
\end{equation*}
so the map cannot be onto.
Thus, the ``order of $L$ and $\Sym$'' in the above lemma cannot be
interchanged. 

For $\CC^d$, the relevant symmetry group is $\rmU(d)$. 
We now discuss its representation on $L(\Sym_t(\CC^d))$, where we will identify an analogue of harmonic symmetric tensors.

The complex version of the contraction map $C$ is the \emph{partial
trace} $\tr_1$. It is defined by
\begin{equation*}
	\tr_1: 
	L\big( (\CC^d)^{\otimes t} \big)
	\to
	L\big( (\CC^d)^{\otimes (t-1)} \big),
	\qquad
	(u_1\otimes \dots \otimes u_t)
	\otimes 
	(\alpha_1 \otimes \dots \otimes \alpha_t)
	\mapsto
	\alpha_1(u_1) \,
	(u_2\otimes \dots \otimes u_t)
	\otimes 
	(\alpha_2 \otimes \dots \otimes \alpha_t).
\end{equation*}
Using the isomorphism (\ref{eqn:LtensorDef}), we can equivalently
define the partial trace by its action on $(L(\CC^d))^{\otimes t}$,
where it takes the form
\begin{equation*}
	\tr_1: 
	A_1\otimes\dots\otimes A_t 
	\mapsto
	(\tr A_1)\, A_2 \otimes \dots \otimes A_t,
	\qquad
	A_i \in L(\CC^d).
\end{equation*}
The origin of the name ``partial trace'' becomes most apparent in this
formulation.  In any case, it is clear that the space $\ker \tr_1$ is
invariant under the action by conjugation
\begin{equation*}
	A \mapsto U^{\otimes t} A (U^\dagger)^{\otimes t}
\end{equation*}
of $\rmU(d)$ on $L( (\CC^d)^{\otimes t})$. 
We now define the \emph{complex harmonic tensors of bi-degree $(t,t)$}
to be 
\begin{equation*}
	H_{(t,t)}(\CC^d):= \ker \tr_1 \cap L(\Sym_t(\CC^d)).
\end{equation*}
As in the real case, this gives rise to a decomposition 
\begin{equation}\label{eqn:complexharmonicdecomp}
	L(\Sym_t(\CC^d))
	\simeq 
	\bigoplus_{i=0}^{t} H_{(i,i)}(\CC^d)
\end{equation}
in terms of irreducible \cite[Chapter~12.2]{Rudi80book} representations of $\rmU(d)$.

In the language of polynomials, the partial trace maps to the complex Laplacian up to a normalization factor:
\begin{align*}
	\Delta_\CC x^\mu \bar x^\nu
	&=
	\sum_{i=1}^d 
	\mu_i \nu_i 
	\,x^{\mu - e_i}\bar x^{\nu - e_i}, \\
	\tr_1 e_\mu e^*_\nu
	&=
	\frac1{t^2}
	\sum_{i=1}^d 
	\mu_i \nu_i
	\,e_{\mu - e_i}e^*_{\nu - e_i},
\end{align*}
so that
\begin{equation*}
	\Delta_\CC\,p_A
	=
	t^2\,
	p_{(\tr_1 A)} \qquad \forall A \in L(\Sym_{t}(\CC^d)).
\end{equation*}
The \emph{harmonic polynomials of bi-degree $(t,t)$} are
\begin{equation*}
	\Harm_{(t,t)}(\CC^d) := \ker \Delta_\CC \cap \Hom_{(t,t)}(\CC^d).
\end{equation*}
An embedding of $\Harm_{(t-j,t-j)}(\CC^d) \to \Hom_{(t,t)}(\CC^d)$ is given by
$p\mapsto p'$, with
\begin{equation*}
	p'(\bar x, x) = \|x\|_2^{2j} p(\bar x, x).
\end{equation*}
For $H_{(t-j,t-j)} \to L(\Sym_t(\CC^d))$, this corresponds to
\begin{equation}\label{eqn:complexSymmetricEmbedding}
	A \mapsto P_{[t]} \left(\Id^{\otimes j} \otimes A\right)P_{[t]}, 
\end{equation}
where $P_{[t]}$ is the projector onto $\Sym_{t}(\CC^d)$.

In the current language, a \emph{complex projective $t$-design} for $\CC^d$ is a set $X$ of normalized vectors such that
\begin{equation*}
	\frac1{|X|} \sum_{\psi \in X} \left(\ket\psi\bra\psi\right)^{\otimes t} \in H_{(0,0)}(\CC^d).
\end{equation*}
Let $G\subset \rmU(d)$ be a finite subgroup of $\rmO(d)$. 
As in the real case, let
$\{ B^{(2i)}_{j} \}_j$ be an orthonormal basis of the $G$-invariant linear maps in $H_{(i,i)}(\CC^d)$.
For $X = G \cdot \psi_0$, we get 
\begin{equation*} \frac1{|X|} \sum_{\psi \in X} \left(\ket\psi\bra\psi\right)^{\otimes t} 
	=
	\sum_{j,i} B^{(2i)}_j \tr\left(B^{(2i)}_j (\ket{\psi_0}\bra{\psi_0})^{\otimes t}\right)
	=
	\sum_{j,i} B^{(2i)}_j \, p_{B^{(2i)}_j}(\bar\psi_0,\psi_0).
\end{equation*}
In particular, $G\cdot \psi_0$ is a complex projective $t$-design if and only if $p_{B^{(2i)}_j}(\bar\psi_0,\psi_0)=0$ for all $i\geq 1$ and all $j$.
This is equivalent to saying that $\psi_0$ is a root of all $G$-invariant harmonic polynomials.

Comparing these conditions with the analogues ones of the real case, we note that a \emph{real spherical $t$-design} is defined in terms of powers $\psi^{\otimes t}$ of degree $t$, while \emph{complex projective $t$-designs} depend on the behavior of $(\ket{\psi}{\bra{\psi}})^{\otimes t}$ of \emph{bi-}degree $t$. 
Therefore, complex projective $t$-designs are often conceptually close to real spherical designs of degree $2t$ or even $2t+1$. 
One example of such a connection is discussed in Section~\ref{sec:harmonic}.

If $G$ affords no harmonic invariants of degree $(s,s)$ for $1\leq s \leq t$, then the orbit of any vector $\psi_0$ constitutes a complex projective $t$-design.
Using Schur's Lemma and the identification of $\Hom_{(t,t)}(\CC^d)$ with $L(\Sym_t(\CC^d))$, this condition is equivalent to demanding that $G$ acts irreducibly on $\Sym_t(\CC^d)$.

For real actions and spherical designs, irreducibility is sufficient but  \emph{not} necessary \cite{Harp04}.
The crucial difference seems to be that
\begin{equation*}
	\Hom_{(t,t)}(\CC^d) \simeq L(\Sym_t(\CC^d)),
	\qquad
	\text{while}
	\qquad
	\Hom_{t}(\RR^d) \varsubsetneqq L(\Sym_{t/2}(\RR^d)).
\end{equation*}
Indeed, let $f\in \Sym_t(\RR^d)\subset (\RR^d)^{\otimes t}$, for $t$ even.
The standard inner product on $\RR^d$ induces a linear isomorphism $\RR^d \simeq (\RR^d)^*$. 
Applying this isomorphism to the ``rear $t/2$'' factors, we obtain an isomorphism
\begin{equation*}
	(\RR^d)^{\otimes t}	
	\simeq
	(\RR^d)^{\otimes t/2}\otimes ((\RR^d)^*)^{\otimes t/2}
	\simeq
	L( (\RR^d)^{\otimes t/2} ).
\end{equation*}
Explicitly,
\begin{equation}\label{eqn:totallysymmetric}
	u_1 \otimes \dots \otimes u_t 
	\mapsto
	(u_1\otimes \dots \otimes u_{t/2})
	(u_{t/2+1}\otimes \dots \otimes u_{t})^*.
\end{equation}
The image of $\Sym_{t}(\RR^d)$ under this isomorphism is a subset 
\begin{equation*}
	\operatorname{MSym}_{t/2}\subset L(\Sym_{t/2}(\RR^d)),
\end{equation*}
which, following \cite{DoheW13}, we refer to as the set of \emph{maximally symmetric matrices}.
These are linear maps whose range and support are on the totally symmetric subspace, and which are in addition invariant under the partial transpose operation. 
With respect to a product basis, these matrices are invariant not only under permutations of covariant \emph{or} contravariant indices amongst themselves, but in addition under permutations of arbitrary indices \cite{DoheW13}.

\bibliographystyle{abbrv}

\bibliography{all_references}

\end{document}